\documentclass[11pt,english]{article}
\usepackage[T1]{fontenc}
\usepackage[latin9]{inputenc}
\usepackage{geometry}
\usepackage{lmodern}

\usepackage{array}
\usepackage{verbatim}
\usepackage{multirow}
\usepackage{amsmath}
\usepackage{amsthm}
\usepackage{amssymb}
\usepackage[unq]{unique}
\usepackage{thmtools}

\let\oldFootnote\footnote
\newcommand\nextToken\relax

\renewcommand\footnote[1]{%
	\oldFootnote{#1}\futurelet\nextToken\isFootnote}

\newcommand\isFootnote{%
	\ifx\footnote\nextToken\textsuperscript{,}\fi}

\makeatletter

\providecommand{\tabularnewline}{\\}

\declaretheoremstyle[
headfont=\normalfont\bfseries, 
bodyfont=\normalfont,
spaceabove=1em plus 0.75em minus 0.25em,
spacebelow=1em plus 0.75em minus 0.25em,
qed={$\triangle$},
]{defstyle2}

\usepackage{xcolor}
\usepackage{xspace}
\usepackage{nameref}
\definecolor{ForestGreen}{rgb}{0.1333,0.5451,0.1333}
\definecolor{DarkRed}{rgb}{0.8,0,0}
\definecolor{Red}{rgb}{1,0,0}
\usepackage[linktocpage=true,
pagebackref=true,colorlinks,
linkcolor=DarkRed,citecolor=ForestGreen,
bookmarks,bookmarksopen,bookmarksnumbered]
{hyperref}

\usepackage{cleveref}

\usepackage{enumitem}

\usepackage{microtype} %

\declaretheorem[numberwithin=section,refname={Theorem,Theorems},Refname={Theorem,Theorems},name={Theorem}]{thm}
\declaretheorem[numberlike=thm,refname={Theorem,Theorems},Refname={Theorem,Theorems},name={Theorem}]{theorem}
\declaretheorem[numberlike=thm,refname={Lemma,Lemmas},Refname={Lemma,Lemmas},name={Lemma}]{lem}
\declaretheorem[numberlike=thm,refname={Corollary,Corollaries},Refname={Corollary,Corollaries},name={Corollary}]{cor}
\declaretheorem[numberlike=thm,refname={Question,Questions},Refname={Question,Questions},name={Question}]{question}
\declaretheorem[numberlike=thm,refname={Fact,Facts},Refname={Fact,Facts},name={Fact}]{fact}

\declaretheorem[numberlike=thm,refname={Proposition,Propositions},Refname={Proposition,Propositions},name={Proposition}]{prop}

\declaretheorem[numberlike=thm,refname={Assumption,Assumptions},Refname={Assumption,Assumptions}]{assumption}
\declaretheorem[style=defstyle2,numberlike=thm,refname={Example,Examples},Refname={Example,Examples}]{example}
\declaretheorem[style=defstyle2,numberlike=thm,refname={Definition,Definitions},Refname={Definition,Definitions},name={Definition}]{defn}
\declaretheorem[style=remark,numberwithin=section,refname={Remark,Remarks},Refname={Remark,Remarks},name={Remark}]{rem}
\declaretheorem[style=remark,numberlike=thm,refname={Claim,Claims},Refname={Claim,Claims}]{claim}

\renewcommand\equationautorefname{\@gobble}

\AtBeginDocument{%
\let\ref\Cref
}

\makeatother

\makeatletter
\renewcommand{\paragraph}{%
	\@startsection{paragraph}{4}%
	{\z@}{1.25ex \@plus 1ex \@minus .2ex}{-1em}%
	{\normalfont\normalsize\bfseries}%
}
\makeatother

\usepackage{fixfoot}
\DeclareFixedFootnote{\repQuasi}{Again, in the word-RAM model, quasi-polynomial preprocessing time and space is allowed.}
\DeclareFixedFootnote{\repOblivious}{A note to readers who are familiar with the {\em oblivious adversaries assumption} for randomized dynamic algorithms: This assumption plays no role for decision problems, since an algorithm that is correct with high probability (w.h.p.) under this assumption is also correct w.h.p. without the assumption (in other words, its output reveals no information about its randomness). Because of this, we do not discuss this assumption in this paper.}

 \ifdefined\DEBUG

\def\thatchaphol#1{\marginpar{$\leftarrow$\fbox{T}}\footnote{$\Rightarrow$~{\sf\textcolor{purple}{#1 --Thatchaphol}}}}
\def\danupon#1{\marginpar{$\leftarrow$\fbox{D}}\footnote{$\Rightarrow$~{\sf\textcolor{orange}{#1 --Danupon}}}}
\def\sayan#1{\marginpar{$\leftarrow$\fbox{S}}\footnote{$\Rightarrow$~{\sf\textcolor{green}{#1 --Sayan}}}}

\newcommand\dan[1]{{\sf \textcolor{orange}{#1}}}

\else

\def\danupon#1{}
\def\thatchaphol#1{}
\def\sayan#1{}
\newcommand\dan[1]{}

\fi

\begin{document}

\title{Coarse-Grained Complexity for Dynamic Algorithms}
\date{}

\author{
Sayan Bhattacharya\thanks{
	University of Warwick, UK} 
\and 
Danupon Nanongkai\thanks{
 	KTH Royal Institute of Technology, Sweden}
\and
Thatchaphol Saranurak\thanks{
	Toyota Technological Institute at Chicago, USA. Work partially done while at KTH Royal Institute of Technology.}
}

\global\long\def\R{\mathcal{R}}
\global\long\def\L{\mathcal{L}}
\global\long\def\I{\mathcal{I}}
\global\long\def\A{\mathcal{A}}
\global\long\def\M{\mathcal{M}}
\global\long\def\D{\mathcal{D}}
\global\long\def\O{\mathcal{O}}
\global\long\def\V{\mathcal{V}}
\global\long\def\T{\mathcal{T}}
\global\long\def\F{\mathcal{F}}
\global\long\def\U{\mathcal{U}}
\global\long\def\X{\mathcal{X}}
\global\long\def\Y{\mathcal{Y}}
\global\long\def\Z{\mathcal{Z}}
\global\long\def\P{\mathcal{P}}
\global\long\def\G{\mathcal{G}}
\global\long\def\E{\mathcal{E}}
\global\long\def\X{\mathcal{X}}
\global\long\def\C{\mathcal{C}}
\global\long\def\opt{\textsf{opt}}
\global\long\def\B{\mathcal{B}}

\global\long\def\drankNP{\mathsf{NP}^{dy}}
\global\long\def\drankcoNP{\mathsf{rankcoNP}^{dy}}
\global\long\def\poly{\mathrm{poly}}
\newcommand{\polylog}{\operatorname{polylog}}
\global\long\def\dDNF{\text{DNF}^{dy}}
\global\long\def\sDNF{\text{Search-DNF}^{dy}}
\global\long\def\dOV{\text{OV}^{dy}}
\global\long\def\dAW{\text{Allwhite}^{dy}}
\global\long\def\dIndep{\text{Indep}^{dy}}

\global\long\def\fDNF{\text{First-DNF}^{dy}}
\global\long\def\fDT{\text{First-DT}^{dy}}

\global\long\def\mem{\mathtt{mem}}
\global\long\def\init{\mathtt{prep}}
\global\long\def\time{\mathsf{Time}}
\global\long\def\inp{\mathtt{in}}
\global\long\def\qsize{\mathsf{BlowUp}}
\global\long\def\outp{\mathtt{out}}
\global\long\def\out{\mathtt{out}}
\global\long\def\work{\mathtt{work}}
\global\long\def\prf{\mathtt{pf}}
\global\long\def\inst{\mathtt{inst}}
\global\long\def\priv{\mathtt{priv}}

\global\long\def\kconn{k\mbox{-}\mathsf{EdgeConn}^{dy}}
\global\long\def\conn{\mathsf{Conn}^{dy}}
\global\long\def\spanningForest{\mathsf{SpanningForest}^{dy}}
\global\long\def\kminTreeCut{k\mbox{-}\mathsf{MinTreeCut}^{dy}}

\newcommand{\dP}{{\sf P}^{dy}}
\newcommand{\dPram}{{\sf RAM-P}^{dy}}
\newcommand{\dPprb}{{\sf Probe-P}^{dy}}
\newcommand{\dFP}{{\sf FP}^{dy}}
\newcommand{\dprmP}{{\sf promise-P}^{dy}}

\newcommand{\dBPP}{{\sf BPP}^{dy}}
\newcommand{\dBPPram}{{\sf RAM-BPP}^{dy}}
\newcommand{\dBPPprb}{{\sf Probe-BPP}^{dy}}
\newcommand{\dFBPP}{{\sf FBPP}^{dy}}

\newcommand{\dNP}{{\sf NP}^{dy}}
\newcommand{\dNPram}{{\sf RAM-NP}^{dy}}
\newcommand{\dNPprb}{{\sf Probe-NP}^{dy}}
\newcommand{\dFNP}{{\sf FNP}^{dy}}
\newcommand{\dTFNP}{{\sf TFNP}^{dy}}

\newcommand{\dMA}{{\sf AM}^{dy}}
\newcommand{\dcoMA}{{\sf coAM}^{dy}}
\newcommand{\dAM}{{\sf AM}^{dy}}
\newcommand{\dcoAM}{{\sf coAM}^{dy}}

\newcommand{\dcoNP}{{\sf coNP}^{dy}}
\newcommand{\dcoNPram}{{\sf RAM-coNP}^{dy}}
\newcommand{\dcoNPprb}{{\sf Probe-coNP}^{dy}}

\newcommand{\drankNPram}{{\sf RAM-rankNP}^{dy}}
\newcommand{\drankNPprb}{{\sf Probe-rankNP}^{dy}}

\newcommand{\dPH}{{\sf PH}^{dy}}
\newcommand{\dPHram}{{\sf RAM-PH}^{dy}}
\newcommand{\dPHprb}{{\sf Probe-PH}^{dy}}
\newcommand{\dprmPH}{{\sf promise-PH}^{dy}}
\newcommand{\dSigma}{\Sigma^{dy}}
\newcommand{\dPi}{\Pi^{dy}}

\newcommand{\relaxed}{\mbox{\sf relaxed}\text{-}}
\newcommand{\patrascu}{P{\v a}tra{\c s}cu \xspace}

\begin{titlepage}	
	\pagenumbering{roman}
	\maketitle
	\ifdefined\DEBUG
	\begin{center}
		{\centering\huge\textcolor{red}{DEBUG VERSION}}
	\end{center}
	\fi
	
	\begin{abstract}

	To date, the only way to argue {\em polynomial lower bounds} for dynamic algorithms is via {\em fine-grained complexity arguments}. These arguments rely on strong assumptions about {\em specific problems} such as the Strong Exponential Time Hypothesis (SETH) and the Online Matrix-Vector Multiplication Conjecture (OMv). While they have led to many exciting discoveries,  dynamic algorithms still miss out some benefits and lessons from the traditional ``coarse-grained'' approach that relates together classes of problems such as P and NP. In this paper we initiate the study of coarse-grained complexity theory for dynamic algorithms. Below are among questions that this theory can answer.

{\bf What if dynamic Orthogonal Vector (OV) is easy in the cell-probe model?} 
A research program for proving {\em polynomial unconditional lower bounds} for dynamic OV in the cell-probe model is motivated by the fact that many conditional lower bounds can be shown via reductions from the dynamic OV problem (e.g. [Abboud, V.-Williams, FOCS 2014]). 
Since the cell-probe model is more powerful than word RAM and has historically allowed smaller upper bounds (e.g. [Larsen, Williams, SODA 2017; Chakraborty, Kamma,  Larsen, STOC 2018]), it might turn out that dynamic OV is easy in the cell-probe model, making this research direction infeasible. 
Our theory implies that if this is the case, there will be very interesting algorithmic consequences: If dynamic OV can be maintained in polylogarithmic worst-case update time in the cell-probe model, then so are several important dynamic problems such as $k$-edge connectivity,  $(1+\epsilon)$-approximate mincut, $(1+\epsilon)$-approximate matching, planar nearest neighbors, Chan's subset union and 3-vs-4 diameter. The same conclusion can be made when we replace dynamic OV by, e.g.,  subgraph connectivity, single source reachability, Chan's subset union, and 3-vs-4 diameter. 

{\bf Lower bounds for $k$-edge connectivity via dynamic OV?} The ubiquity of reductions from dynamic OV 
raises a question whether we can prove conditional lower bounds for, e.g., $k$-edge connectivity, approximate mincut, and approximate matching, via the same approach. Our theory provides a method to refute such possibility (the so-called {\em non-reducibility}). In particular, we show that there are no ``efficient'' reductions (in both cell-probe and word RAM models) from dynamic OV to $k$-edge connectivity under an assumption about the classes of dynamic algorithms whose analogue in the static setting is widely believed. 
We are not aware of any existing assumptions that can play the same role. (The NSETH of Carmosino et al. [ITCS 2016] is the closest one, but is not enough.)
To show similar results for other problems, one only need to develop efficient randomized {\em verification protocols} for such problems. 
\end{abstract}

	\newpage 
	
	\setcounter{tocdepth}{2}
	\tableofcontents
	
\end{titlepage}

\newpage

\part{Extended Abstract}
\label{part:extended:abstract}
\pagenumbering{arabic}

\section{Introduction}\label{sec:intro}

In  a {\em dynamic problem}, we first get an input instance for {\em preprocessing}, and subsequently we have to handle a sequence of {\em updates} to the input.
For example, in the graph connectivity problem \cite{HenzingerK99,KapronKM13}, an $n$-node graph $G$ is given to an algorithm to preprocess. 
Then the algorithm has to answer whether $G$ is connected or not after each edge insertion and deletion to $G$. (Some dynamic problems also consider {\em queries}. (For example, in the connectivity problem an algorithm may be queried whether two nodes are in the same connected component or not.) Since queries can be phrased as input updates themselves, we will focus only on updates in this paper 
(see~\ref{sec:model} for formal definitions and examples.) 
Algorithms that handle dynamic problems are known as {\em dynamic algorithms}. 
The {\em preprocessing time} of a dynamic algorithm is the time it takes to handle the initial input, whereas the worst-case {\em update time} of a dynamic algorithm is the {\em maximum} time it takes to handle any update.
Although dynamic algorithms are also analyzed in terms of their {\em amortized update times}, %
we emphasize that the results in this paper deal only with worst-case update times.
A holy grail for many dynamic problems -- especially those concering dynamic {\em graphs}  under edge deletions and insertions -- is to design algorithms with
{\em polylogarithmic update times}. From this perspective, the computational status of many classical dynamic problems still remain widely open.

\paragraph{Example: Family of Connectivity Problems.} 
A famous example of a widely open question is for the family of connectivity problems:
(i) The problem of maintaining whether the input dynamic graph is connected (the dynamic {\em  connectivity} problem) admits a {\em randomized} algorithm with polylogarithmic worst-case update time. It is an active, unsettled line of research to determine whether it admits  deterministic polylogarithmic worst-case update time (e.g.~\cite{Frederickson85,EppsteinGIN92,HenzingerK99,HolmLT98,Thorup00,PatrascuD06,KapronKM13,Wulff-Nilsen17,NanongkaiS17,NanongkaiSW17}). 
(ii) The problem of maintaining whether the input dynamic graph can be disconnected by deleting an edge (the dynamic {\em $2$-edge connectivity} problem) admits polylogarithmic amortized update time \cite{HolmLT98,HolmRT18}, but its worst-case update time (even with randomization) remains polynomial \cite{Thorup01}. (iii) For dynamic $k$-edge connectivity with $k\geq 3$, the best update time -- amortization and randomization allowed -- suddenly jumps to $\tilde O(\sqrt{n})$ where 
$\tilde O$ hides polylogarithmic terms. Indeed, it is a major open problem to maintain a $(1+\epsilon)$-approximation to the value of global minimum cut in a dynamic graph in polylogarithmic update time~\cite{Thorup01}. Doing so for $k$-edge connectivity with $k=O(\log n)$ is already sufficient to solve the general case.

Other dynamic  problems that are not known to admit polylogarithmic update times include approximate matching, shortest paths, diameter, max-flow, etc.~\cite{Thorup01,Sankowski07}.
Thus, it is natural to ask: {\em can one argue that these problems do not admit efficient dynamic algorithms?}

A traditionally popular approach to  answer the question above is to use information-theoretic arguments in the {\em bit-probe/cell-probe model}. 
In this model of computation,  all the operations are free {\em except} memory accesses. (In more details, the bit-probe model concerns the number of bits accessed, while the cell-probe model concerns the number of accessed cells, typically of logarithmic size.) Lower bounds via this approach are usually {\em unconditional}, meaning that it does not rely on any assumption. 
Unfortunately, this approach could only give small lower bounds so far; and getting a super-polylogarithmic lower bound for {\em any} natural dynamic problem is an outstanding open question is this area~\cite{Larsen17boolean}. 

More recent advances towards answering this question arose from a new area called {\em fine-grained complexity}. While traditional complexity theory (henceforth we refer to it as  {\em coarse-grained complexity}) focuses on classifying problems based on resources and relating together resulting classes (e.g. P and NP), fine-grained complexity gives us {\em conditional} lower bounds in the word RAM model based on various assumptions about specific problems. For example, assumptions that are particularly useful for dynamic algorithms are the Strong Exponential Time Hyposis (SETH), which concerns the running time for solving SAT,  and the Online Matrix-Vector Multiplication Conjecture (OMv), which concerns the running time of certain matrix multiplication methods (more at, e.g., \cite{Patrascu10,AbboudW14,OMV}). 
In sharp contrast to cell-probe lower bounds, these assumptions often lead to polynomial lower bounds in the word RAM model, many of which are tight.

While the  fine-grained complexity approach has led to many exciting lower bound results, there are a number of traditional results in the static setting that seem to have {\em no} analogues in the dynamic setting. 
For example, one reason that makes the $P\neq NP$ assumption so central in the static setting is that proving and disproving it will both lead to stunning consequences:  If the assumption is false then hundreds of problems in NP and bigger complexity classes like the polynomial hierarchy (PH) admit efficient algorithms; otherwise the situation will be the opposite.\footnote{For further consequences see, e.g., \cite{Aaronson17,Cook2006p,Cook03} and references therein.}
In contrast, we do not see any immediate consequence to dynamic algorithms if someone falsified SETH, OMv, or any other assumptions.\footnote{An indirect consequence would be that some barriers were broken and we might hope to get better upper bounds. This is however different from when P=NP where many hard problems would immediately admit efficient algorithms. Note that some consequences of falsifying SETH have been shown recently (e.g. \cite{AbboudBDN18,Williams13,GaoIKW17,CyganDLMNOPSW16,SanthanamS12,JahanjouMV15,DantsinW13}); however, we are not aware of any consequence to dynamic algorithms. It might also be interesting to note that Williams~\cite{Williams18-likelihood} estimates the likelihood of SETH to be only 25\%.}
As another example, comparing complexity classes allows us to speculate on various situations such as non-reducibility (e.g. \cite{AharonovR05,KlivansM02,Condon92}), the existence of NP-intermediate problems \cite{Ladner75} and the derandomization possibilities (e.g. \cite{ImpagliazzoW97}). (See more in~\ref{sec:open}.)
We cannot anticipate results like these in the dynamic setting without the coarse-grained approach, i.e. by considering analogues of P, NP, BPP and other complexity classes that are defined based on computational resources.

\paragraph{Our Main Contributions.} We initiate a systematic study of coarse-grained complexity theory for dynamic problems in the bit-probe/cell-probe model of computation. We now mention a couple of concrete implications that follow from this study.

Consider  the {\em dynamic Orthogonal Vector} (OV) problem (see~\ref{main:def:OV}). Lower bounds conditional on SETH for many natural problems (e.g.  Subgraph connectivity, ST-reachability, Chan's subset union,  3-vs-4 Diameter) are based on reductions from dynamic OV~\cite{AbboudW14}.
This suggests two research directions: (I) Prove strong {\em unconditional} lower bounds for many natural problems in one shot by proving a polynomial cell-probe lower bound for dynamic OV. (II) Prove lower bounds conditional on SETH for the family of connectivity problems mentioned in the previous page via reductions (in the word RAM model) from dynamic OV. 
Below are some questions about the feasibility of these research directions that our theory can answer. We are not aware of any other technique in the existing literature that can provide similar conclusions.

\smallskip
\noindent {\bf (I) What if dynamic OV is easy in the cell-probe model?}
For the first direction, there is a risk that dynamic OV might turn out to admit a polylogarithmic update time algorithm in the cell-probe model. This is because  lower bounds  in the word RAM model do not necessarily extend to the cell-probe model. For example, it was shown by Larsen and Williams~\cite{LarsenW17} and later by Chakraborty~et~al~\cite{ChakrabortyKL17} that the OMv conjecture \cite{OMV} is false in the cell-probe model. 

Will all the efforts be wasted if dynamic OV turns out to admit polylogarithmic update time in the cell-probe model? Our theory implies that this will also lead to a very interesting algorithmic consequence: If dynamic OV admits polylogarithmic update time in the cell-probe model, so do several important dynamic problems such as $k$-edge connectivity,  $(1+\epsilon)$-approximate mincut, $(1+\epsilon)$-approximate matching, planar nearest neighbors, Chan's subset union and 3-vs-4 diameter.
The same conclusion can be made when we replace dynamic OV by, e.g., subgraph connectivity, single source reachability, Chan's subset union, and 3-vs-4 diameter (see~\ref{thm:intro:results_summary}). 
Thus, there will be interesting consequences regardless of the outcome of this line of research.%

Roughly, we reach the above conclusions by proving in the dynamic setting an analogue of the fact that if P=NP (in the static setting), then the polynomial hierarchy (PH) collapses. This is done by carefully defining the classes $\dP$, $\dNP$ and $\dPH$ as dynamic analogues of P, NP, and PH, so that we can prove such statements, along with $\dNP$-completeness and $\dNP$-hardness results for natural dynamic problems including dynamic OV. We sketch how to do this in~\ref{sec:intro:contribution},~\ref{sec:overview}.

\smallskip
\noindent {\bf (II) Lower bounds for $k$-edge connectivity via dynamic OV?} As discussed above, whether dynamic $k$-edge connectivity admits  polylogarithmic update time for  $k\in [3, O(\log n)]$ %
 is a very important open question. There is  a hope to answer this question negatively via reductions (in the word RAM model) from dynamic OV. 
Our theory provides a method to refute such a possibility (the so-called {\em non-reducibility}). First, note that any reduction from dynamic OV in the word RAM model will also hold in the (stronger) cell-probe model. Armed with this simple observation, we show that there are no ``efficient'' reductions from dynamic OV to $k$-edge connectivity under an assumption about the complexity classes for dynamic problems in the cell-probe model, namely $\dPH \not\subseteq \dMA \cap \dcoMA$ (see~\ref{thm:intro:ruleout_reductions}). We defer defining the classes $\dMA$ and $\dcoMA$, but note two things. (i) Just as the classes AM and  coAM (where AM stands for {\em Arthur-Merlin}) extend NP in the static setting, the classes $\dMA$ and $\dcoMA$ extend the class $\dNP$ in a similar manner. (ii) In the static setting it is widely believed that PH $\not\subseteq$ AM $\cap$ coAM because otherwise the PH collapses.
Roughly, the phrase ``{\em efficient reduction}" from problems $X$ to $Y$  refers to a way of processing each update for problem $X$ by quickly feeding polylogarithmic number of updates as input to an algorithm for $Y$ (see~\ref{sec:reduction} for details). 
All reductions from dynamic OV in the literature%
that we are aware of are efficient reductions. 

\paragraph{Remark.} We define our complexity classes in the cell-probe model, whereas the reductions from dynamic OV are in the word RAM model. This does not make any difference, however, since any reduction in the  word RAM reduction continues to have the same  guarantees in the (stronger) cell-probe model.

\medskip
To show a similar non-reducibility result for any problem $X$, one needs to prove that $X\in \dMA \cap \dcoMA$, which boils down to developing efficient randomized {\em verification protocols} for such problems. We explain this in more details in~\ref{sec:intro:contribution} and~\ref{sec:rand_classes}.

We are not aware of any existing assumptions that can lead the same conclusion as above. 
To our knowledge, the only conjecture that can imply results of this nature is the  Nondeterministic Strong Exponential Time Hypothesis (NSETH) of~\cite{CarmosinoGIMPS16}. However, it needs a stronger property of $k$-edge connectivity that is not yet known to be true. (In particular,~\ref{thm:intro:ruleout_reductions} follows from the fact that $k$-edge connectivity is in  $\dMA \cap \dcoMA$. To use NSETH, we need to show that it is in $\dNP\cap \dcoNP$.) Moreover, even if such a property holds it would only rule out deterministic reductions since NSETH only holds for deterministic algorithms.

\smallskip
\noindent {\bf Paper Organization.} \ref{part:extended:abstract} of the paper is organized as follows.~\ref{sec:intro:contribution} explains our contributions in details, including the conclusions above and beyond. We discuss related works and future directions in~\ref{sec:related:work},~\ref{sec:open}. An overview of our main $\dNP$-completeness proof is in~\ref{sec:overview}. 

In~\ref{part:formalization} we formalize the required notions. Although some of them are tedious, they are crucial to get things right.~\ref{part:completeness} provides detailed proofs of our $\dNP$-completeness results.~\ref{part:further_classes} extends our study to other classes, including those with randomization. This is needed for the non-reducibility result mentioned above, discussed in details in~\ref{part:connectivity}.

The focus of this paper is on the cell-probe model. Our results in the word RAM model are mostly similar. We discuss in more details in~\ref{part:RAM}.

\section{Our Contributions in Details} 
\label{sec:intro:contribution}

We show that coarse-grained complexity results similar to the static setting can be obtained for dynamic problems in the bit-probe/cell-probe model of computation,\footnote{Throughout the paper, we use the cell-probe and bit-probe models interchangeably since the complexity in these models are the same up to polylogarithmic factors.}   provided the notion of ``nondeterminism'' is carefully defined.
Recall that the cell-probe model is similar to the word RAM model, but the time complexity is measured by the number of memory reads and writes (probes); other operations are free (see~\ref{sec:model:dynamic}). 
Like in the static setting, we only consider {\em decision} dynamic problems,
meaning that the output after each update is either ``yes'' or ``no''.
Note the following remarks. 
\begin{itemize}
\item Readers who are familiar with the traditional complexity theory may wonder why we do not consider the Turing machine. This is because the Turing machine is not suitable for implementing dynamic algorithms, since we cannot access an arbitrary tape location in $O(1)$ time. There is no efficient algorithm even for a basic data structure like the binary search tree. 
\item Our results for decision problems extend naturally to {\em promised problems} which are useful when we discuss approximation algorithms. We do not discuss promised problems here to keep the discussions simple. 
\item  Readers who are familiar with the oblivious adversaries assumption for randomized dynamic algorithms may wonder if we consider this assumption here.
This assumption plays no role for decision problems, since an algorithm that is correct with high probability (w.h.p.)
under this assumption is also correct w.h.p. without the assumption (in other words, its output reveals no information
about its randomness). Because of this, we do not discuss this assumption in this paper.
\end{itemize}

We start with our main results which can be obtained with 
appropriate definitions of complexity classes $\dP \subseteq \dNP \subseteq \dPH$ for dynamic problems:
These classes are described in details later.  For now they should be thought of as analogues of the classes P, NP and PH (polynomial hierarchy). 

\begin{theorem}($\dP$ vs. $\dNP$) \label{thm:intro:results_summary} Below, the phrase ``efficient algorithms'' refers to dynamic algorithms that are deterministic and require  polylogarithmic worst-case update time and  polynomial space to handle a polynomial number of updates in the bit-probe/cell-probe model. 
\begin{enumerate}%
	\item \label{item:NP-hard-problems} The dynamic orthogonal vector (OV) problem is ``$\dNP$-complete", and there are a number of dynamic problems that are ``$\dNP$-hard'' in the sense that if $\dP\neq \dNP$, then they admit {\em no} efficient algorithms.
	These problems include decision versions of Subgraph connectivity, ST-reachability, Chan's subset union,  and 3-vs-4 Diameter (see~\ref{table:BPP-NP},~\ref{table:NP hard} for more).%
	\item \label{item:PH-problems} If $\dP=\dNP$ then $\dP=\dPH$, meaning that all problems in $\dPH$ (which contains the class $\dNP$) admit efficient algorithms. These problems include decision versions of $k$-edge Connectivity, $(1+\epsilon)$-approximate Matching,\footnote{Technically speaking, $(1+\epsilon)$-approximate matching is a {\em promised} or {\em gap} problem in the sense that for some input instance all answers are correct. It is in $\dprmPH$ which is bigger than $\dPH$. We can make the same conclusion for promised problems: If $\dP=\dNP$, then all problems in $\dprmPH$ admit efficient algorithms.} $(1+\epsilon)$-approximate mincut, Planar nearest neighbors, Chan's subset union and 3-vs-4 Diameter (see~\ref{table:BPP-NP},~\ref{table:PH} for more). 
\end{enumerate}
\end{theorem}

Thus, proving or disproving $\dP\neq \dNP$ will both lead to interesting consequences: If $\dP\neq \dNP$, then many dynamic problems do not admit efficient algorithms. Otherwise, if $\dP = \dNP$, then many problems  admit efficient algorithms which are not known or even believed to exist.

\paragraph{Remark.} We can obtain  similar results in the word-RAM model, but we need a notion of ``efficient algorithms'' that is slightly non-standard in that a quasi-polynomial preprocessing time is allowed. (In contrast, all our results hold in the standard cell-probe setting.) We postpone discussing word-RAM results to later in the paper to avoid confusions. 

\medskip
As another showcase, our study implies a way to show {\em non-reducibility}, like below. 

\begin{theorem}\label{thm:intro:ruleout_reductions}
Assuming  $\dPH \not\subseteq \dMA \cap \dcoMA$, the $k$-edge connectivity problem cannot be $\dNP$-hard.  Consequently, there is no ``efficient reduction'' from the dynamic Orthogonal Vector (OV) problem to $k$-edge connectivity. 
\end{theorem}

From the discussion in~\ref{sec:intro}, recall that the $k$-edge connectivity problem is currently known to admit a polylogarithmic amortized update time algorithm for $k\leq 2$, and a $O(\sqrt{n}\polylog (n))$ update time algorithm for $k\in [3, O(\log n)].$
It is a very important open problem whether it admits polylogarithmic worst-case update time.~\ref{thm:intro:ruleout_reductions} rules out a way to prove lower bounds and suggest that an efficient algorithm might exist. 

A more important point beyond the $k$-edge connectivity problem is that one can prove a similar result for any dynamic problem $X$ by showing that  $X\in  \dMA \cap \dcoMA$ or, even better, $X\in  \dNP \cap \dcoNP$. See~\ref{sec:open} for some candidate problems for $X$.
This is easier than showing a dynamic algorithm for $X$ itself. Thus, this method is an example of the by-products of our study that we expect to be useful for developing algorithms and lower bounds for dynamic problems in  future. See~\ref{sec:contribution:other:results}  and~\ref{part:connectivity} 
for more details. 
As noted in~\ref{sec:intro}, we are not aware of any  existing technique 
that is capable of deriving a non-reducibility result of this kind.

The key challenge in deriving the above results is to come up with the right set of definitions for  various dynamic complexity classes. 
 We provide some of these definitions and discussions here, but defer more details to later in the paper.%

\subsection{Defining the Complexity Classes $\dP$ and $\dNP$}
\label{sec:contribution:classes}

\noindent {\bf Class $\dP$.}
We start with $\dP$, the class of dynamic problems that admit ``efficient'' algorithms in the cell-probe model.  
For any dynamic problem, define its {\em update size} to be the number of bits needed to describe each update.	Note that we have not yet defined what dynamic problems are formally. Such a definition is needed for a proper, rigorous  description of our complexity classes, and can be found in the full version of the paper. For an intuition, it suffices to keep in mind that most dynamic graph problems -- where each update is an edge deletion or insertion -- have logarithmic update size (since it takes $O(\log n)$ bits to specify an edge in an $n$-node graph).
\begin{defn}[$\dP$; brief]\label{def:intro:dP}
A dynamic problem with polylogarithmic update size is in $\dP$ if it admits a deterministic algorithm with polylogarithmic worst-case update time for handling a sequence of polynomially many updates.%
\end{defn}

Examples of problems in $\dP$ include connectivity on plane graphs and predecessor; for more, see~\ref{table:P}. 
Note that one can define $\dP$ more generally to include problems with larger update sizes. Our complexity results hold even with this more general definition. However, since our results are most interesting for problems with polylogarithmic update size, we focus on this case in this paper to avoid cumbersome notations.

\paragraph{Class $\dNP$ and nondeterminism with rewards.}
Next, we introduce our complexity class $\dNP$. Recall that in the static setting the class NP consists of the set of problems that admit efficiently verifiable proofs or, equivalently, that are solvable in polynomial time by a nondeterministic algorithm. 
Our notion of nondeterminism is captured by the proof-verification definition where, after receiving  a proof, the verifier does not only output YES/NO, but also a {\em reward}, which is supposed to be maximized  at every step. 

Before defining $\dNP$ more precisely, we remark that the notion of reward is a key for our $\dNP$-completeness proof. Having the constraint about rewards potentially makes $\dNP$ contains less problems.  Interestingly, all natural problems that we are aware of remains in $\dNP$ even with this constraint. This might not be a big surprise, when it is realized that in the static setting imposing a similar constraint about the reward does {\em not} make the class (static) NP smaller; see more discussions below.
We now define $\dNP$ more precisely.
\begin{defn}[$\dNP$; brief]\label{def:intro:dNP}
	A dynamic problem $\Pi$ with polylogarithmic update size is in $\dNP$ if there is a verifier %
	that can do the following over a sequence of  polynomially many updates:
	(i) after every update, the verifier takes the update and a polylogarithmic-size {\em proof} as an input, and
	(ii) after each update, the verifier outputs in polylogarithmic time a pair $(x, y)$, where $x\in \{YES,NO\}$%
	and $y$ is an integer (representing a reward) with the following properties.\footnote{Later in the paper, we use $x=1$ and $x=0$ to represent $x=YES$ and $x=NO$, respectively.}
	\begin{enumerate}%
		\item If the current input instance is an YES-instance and the verifier has so far always received a proof that maximizes the reward at every step, then the verifier outputs $x=YES$.
		\item If the current input instance is a NO-instance, then the verifier outputs $x=NO$  regardless of the sequence of proofs it has received so far. %
	\end{enumerate}
	\end{defn}

To digest the above definition, first consider the {\em static}  setting. One can redefine the class NP for static problems in a similar fashion to~\ref{def:intro:dNP} by removing the preprocessing part and letting  the only update be the whole input. Let us refer to this new (static) complexity class as ``{\em reward-NP}''.
To show that a static problem is in reward-NP, a verifier has to output some reward in addition to the YES/NO answer. Since usually a proof received in the static setting is a solution itself, a natural choice for reward is the {\em cost} of the solution (i.e., the proof). For example, a ``proof'' in the maximum clique problem is a big enough clique, and in this case an intuitive reward would be the size of the clique given as a proof. Observe that this is sufficient to show that max clique is in reward-NP. In fact, it turns out that {\em in the static setting the complexity classes NP and reward-NP are equal}. (Proof sketch: Let $\Pi$ be any problem in the original static NP and $V$ be a corresponding verifier. We extend $V$ to $V'$  which outputs $x=YES/NO$ as $V$ and outputs $y=1$ as a reward if $x=YES$ and $y=0$ otherwise. It is not hard to check that $V'$ satisfies conditions in~\ref{def:intro:dNP}.)

To further clarify~\ref{def:intro:dNP},  we now consider examples of some well-known dynamic problems that happen to be in $\dNP$.

\begin{example}[Subgraph Detection]\label{ex:subgraphDetect_inNP}
In the {\em dynamic subgraph detection} problem, an $n$-node and $k$-node graphs $G$ and $H$ are given at the preprocessing, for some $k=\polylog(n)$.  Each update is an edge insertion or deletion in $G$. We want an algorithm to output YES if and only if $G$  has $H$ as a subgraph. 

This problem is in $\dNP$ due to the following verifier: the verifier outputs $x=YES$ if and only if the proof (given after each update) is a mapping of the edges in $H$ to the edges in a subgraph of $G$ that is isomorphic to $H$. With output $x=YES$, the verifier gives  reward $y=1$. With output $x=NO$, the verifier gives reward $y=0$.  Observe that the proof  is of polylogarithmic size (since $k = \polylog (n))$, and the verifier can calculate its outputs $(x,y)$ in polylogarithmic time. Observe further that the properties stated in~\ref{def:intro:dNP} are satisfied: if the current input instance is a YES-instance, then the reward-maximizing proof is a mapping between $H$ and the subgraph of $G$ isomorphic to $H$, causing the verifier to output $x=YES$; otherwise, no proof will make the verifier output $x=YES$.  
\end{example}

The above example is in fact too simple to show the strength of our definition that allows $\dNP$ to include many natural problems (for one thing, $y$ is simply 0/1 depending on $x$). 
The next example demonstrates how the definition allows us to develop more sophisticated verifiers for other problems.

\begin{example}[Connectivity]\label{ex:intro:connectivity in NP}
In the {\em dynamic connectivity} problem, an $n$-node graph $G$ is given at the preprocessing.  Each update is an edge insertion or deletion in $G$. We want an algorithm to output YES if and only if $G$ is connected. 

This problem is in $\dNP$ due to the following verifier. After every update, the verifier   maintains a forest $F$ of $G$. A proof (given after each update) is an edge insertion to $F$ or an $\bot$ symbol indicating that there is no update to $F$.  It handles each update as follows.
\begin{itemize}%
	\item After an edge $e$ is inserted into $G$, the verifier checks if $e$ can be inserted into $F$ without creating a cycle. This can be done in $O(\polylog(n))$ time using a link/cut tree data structure \cite{SleatorT81}. It outputs reward $y=0$. (No proof  is needed in this case.)
	\item After an edge $e$ is deleted from $G$,  the verifier checks if $F$ contains $e$. If not, it outputs  reward $y=0$ (no proof  is needed in this case). If $e$ is in $F$, the verifier reads the proof (given after $e$ is deleted). If the proof is $\bot$ it outputs reward $y=0$. Otherwise, let the proof be an edge $e'$. The verifier checks if $F'=F\setminus\{e\}\cup \{e'\}$ is a forest; this can be done in $O(\polylog(n))$ time using a link/cut tree data structure \cite{SleatorT81}. If $F'$ is a forest, the verifier sets $F\gets F'$ and outputs reward $y=1$; otherwise, it outputs reward $y=-1$.  
\end{itemize}
After each update, the verifier outputs $x=YES$ if and only if $F$ is a spanning tree of $G$.

Observe that if the verifier gets a proof that maximizes the reward after every update, the forest $F$ will always be a spanning forest (since inserting an edge $e'$ to $F$ has higher reward than giving $\bot$ as a proof). Thus, the verifier will always output $x=YES$ for YES-instances in this case. It is not hard to see that the verifier never outputs $x=YES$ for NO-instances, no matter what proof it receives. 
\end{example}

In short, a proof for the connectivity problem is the maximal spanning forest. Since such proof is too big to specify and verify after every update, our definition allows such proof to be updated over input changes. (This is as opposed to specifying the densest subgraph from scratch every time as in~\ref{ex:subgraphDetect_inNP}.) Allowing this is crucial for most problems to be in $\dNP$, but create difficulties to prove $\dNP$-completeness. We remedy this by introducing rewards.

Note that if there is no reward in~\ref{def:intro:dNP}, then it is even easier to show that dynamic connectivity and other problems are in $\dNP$. Having an additional constraint about rewards potentially makes less problems verifiable.  Luckily, all natural problems that we are aware of that were verifiable without rewards remain verifiable with rewards.
Problems in $\dNP$ include decision/gap versions of $(1+\epsilon)$-approximate matching, planar nearest neighbor, and dynamic 3SUM; see~\ref{table:BPP-NP} for more. The concept of rewards (introduced while defining the class $\dNP$) will turn to be crucial when we attempt to show the existence of a complete problem in  $\dNP$. See~\ref{sec:contribution:np:completeness} and~\ref{sec:overview} for more details.

It is fairly easy to show that $\dP \subseteq \dNP$,  and we conjecture that $\dP \neq \dNP$.

\paragraph{Previous nondeterminism in the dynamic setting.}
The idea of nondeterministic dynamic algorithms is not completely new.
This was considered by Husfeldt and Rauhe \cite{HusfeldtR03} and their follow-ups \cite{PatrascuD06,PatrascuThesis,Yin10,LarsenNonDeterministic,WangY14}, and has played a key role in proving cell-probe lower bounds in some of these papers. 
As discussed in \cite{HusfeldtR03}, although it is straightforward to define a nondeterministic dynamic algorithm as the one that can make nondeterministic choices to process each update and query, there are different ways to handle how nondeterministic choices affect the states of algorithms which in turns affect how the algorithms handle future updates (called the ``side effects'' in \cite{HusfeldtR03}). For example, in \cite{HusfeldtR03} nondeterminism is allowed only for answering a query, which happens to occur only once at the very end. In \cite{PatrascuD06}, nondeterministic query answering may happen throughout, but an algorithm is allowed to write in the memory (thus change its state) only if all nondeterministic choices lead to the same memory state.

In this paper we define a different notion of nondeterminism and thus the class $\dNP$. It is more general than the previous definitions in that if a dynamic problem admits an efficient nondeterministic algorithm according to the previous definitions, it is in our $\dNP$.  
In a nutshell, the key differences are that (i) we allow nondeterministic steps while processing both updates and queries and (ii) different choices of nondeterminism  can affect the algorithm's states in different ways; however, we distinct different choices by giving them different rewards. 
These differences allow us to include more problems to our $\dNP$ (we do not know, for example, if dynamic connectivity admits nondeterministic algorithms according to previous definitions).

\subsection{$\dNP$-Completeness}
\label{sec:contribution:np:completeness}

Here, we sketch the idea behind our $\dNP$-completeness and hardness results. 
We begin by introducing a  problem is called {\em dynamic narrow DNF evaluation problem} (in short, $\dDNF$), as follows.

\begin{defn}[$\dDNF$; informally]\label{defn:intro:dnf}
Initially, we have to preprocess (i) an $m$-clause $n$-variable DNF formula\footnote{Recall that a DNF formula is in the form $C_{1}\vee\dots\vee C_{m}$, where each ``clause'' $C_i$ is a conjunction (AND) of literals.} where each clause contains $O(\polylog(m))$ literals, and (ii) an assignment of (boolean) values to the variables. Each update  changes the value of one variable. After each update, we have to answer whether the DNF formula is true or false. 
\end{defn}

It is fairly easy to see that $\dDNF \in \dNP$: After each update, if the DNF formula happens to be true, then the proof only needs to point towards one satisfied clause, and the verifier can quickly check if this clause is satisfied or not since it contains only $O(\polylog (m))$ literals. Surprisingly, it turns out that this is also a complete problem  in the classs $\dNP$.

\begin{theorem}[$\dNP$-completeness of $\dDNF$]\label{thm:intro:DNF}
The $\dDNF$ problem is $\dNP$-complete. This means that  $\dDNF\in \dNP$, and if $\dDNF\in \dP$, then $\dP=\dNP$.
\end{theorem}

To start with, recall the following intuition for proving NP-completeness in the static setting  (e.g. \cite[Section 6.1.2]{AroraBarak_book} for details): Since {\em Boolean circuits} can simulate polynomial-time Turing machine computation (i.e. $P \subseteq P \slash \mathsf{poly}$), we view the computation of the verifier $V$ for any problem $\Pi$ in NP as a circuit $C$. The input of $C$ is the proof that $V$ takes as an input. Then, determining whether there is an input (proof) that satisfies this circuit (known as  {\em CircuitSAT}) is NP-complete, since such information will allow the verifier to find a desired proof on its own.  Attempting to extend this intuition to the dynamic setting might encounter the following roadblocks.

\begin{enumerate}%
\item Boolean circuits cannot efficiently simulate algorithms in the RAM model without losing a linear factor in running time.%
Furthermore, an alternative such as circuits with ``indirect addressing'' gates seems useless, because this complex gate makes the model more complicated. 
This makes it more difficult to prove  $\dNP$-hardness.%
\item Since the verifier has to work through several updates in the dynamic setting, the YES/NO output from the verifier alone is insufficient to indicate proofs that can be useful for {\em future} updates. For example, suppose that in~\ref{ex:intro:connectivity in NP}   the connectivity verifier is allowed to output only $x\in \{YES, NO\}$, and we get rid of the concept of a reward.  
Consider a scenario where 
an edge $e$ (which is part of $F$) gets deleted from $G$, and $G$ was disconnected even before this deletion. In this case, the verifier can indicate no difference between having $e'$ (i.e. finding a reconnecting edge) and $\bot$ (i.e. doing nothing) as a proof (because it has to output $x=0$ in both cases). Having $e'$ as a proof, however, is more useful for the future, since it  helps maintain a spanning forest.
\end{enumerate}

It so happens that we can solve (ii) if the verifier additionally outputs an integer $y$ as a reward. Asking more from the verifier  makes less problems verifiable (thus smaller $\dNP$).  Luckily, all natural problems we are aware of that were verifiable without rewards remain verifiable with rewards!

To solve (i), we use the fact that in the cell-probe model a polylogarithmic-update-time algorithm can be modeled by a polylogarithmic-depth {\em decision assignment tree} \cite{Miltersen99cellprobe}, which naturally leads to a complete problem about a decision tree (we leave details here; see~\ref{sec:overview}  and~\ref{sec:DT} 
for more).
It turns out that we can reduce from this problem to $\dDNF$ (\ref{defn:intro:dnf}); the intuition being that each bit in the main memory corresponds to a boolean variable and each root-to-leaf path in the decision assignment tree can be thought of as a DNF clause. The only downside of this approach is that a  polylogarithmic-depth decision tree  has quasi-polynomial size. A straightforward reduction would cause quasi-polynomial space in the cell-probe model. By exploiting the special property of $\dDNF$ and the fact that the cell-probe model only counts the memory access, we can avoid this space blowup by ``hardwiring'' some space usage into the decision tree and reconstruct some memory when needed.

The fact that the $\dDNF$ problem is $\dNP$-complete (almost) immediately implies  that many well-known dynamic problems are $\dNP$-hard. To explain why this is the case, we first recall the definition of the \emph{dynamic sparse orthogonal vector} ($\dOV$) problem.

\begin{defn}[$\dOV$]\label{def:OV}\label{main:def:OV}
Initially, we have to preprocess a collection of $m$ vectors  $V = \{ v_1, \ldots, v_{m} \}$ where each $v_j \in \{0,1\}^n$, and another vector $u \in \{0,1\}^n$.  It is guaranteed that each $v_j \in \{0,1\}^n$ has at most $O(\polylog (m))$ many nonzero entries. Each update flips the value of one entry in the vector $u$. After each update, we have to answer if there is a vector $v \in V$ that is orthogonal to $u$ (i.e., if $u^T v = 0$).
\end{defn}

\medskip
The key observation is that the $\dOV$ problem is equivalent to the $\dDNF$ problem, in the sense that $\dOV \in \dP$ iff $\dDNF \in \dP$. The proof is relatively straightforward (the vectors $v_j$ and the individual entries of $u$ respectively correspond to the clauses and the variables in $\dDNF$), and we defer it to~\ref{sec:DNF equiv proof}. 
In~\cite{AbboudW14}, Abboud and Williams show SETH-hardness for all
of the problems in~\ref{main:table:NP hard}. In fact, they
actually show a reduction from $\dOV$ to these problems. Therefore, we immediately obtain the following result.
\begin{cor}
\label{main:cor:list rankNP hard}All problems in~\ref{main:table:NP hard}
are $\dNP$-hard.
\end{cor}

\subsection{Dynamic Polynomial Hierarchy}

By introducing the notion of {\em oracles}  (see~\ref{sec:model:dynamic}), 
it is not hard to extend the class $\dNP$ into a {\em polynomial-hierarchy} for dynamic problems, denoted by $\dPH$. Roughly,  $\dPH$ is the union of classes $\dSigma_i$ and $\dPi_i$, where (i) $\dSigma_1= \dNP$, $\dPi_1=\dcoNP$,  and (ii) we say that a dynamic problem is in class $\dSigma_i$ (resp. $\dPi_i$) if we can show that it is in $\dNP$  (resp. $\dcoNP$) assuming that there are efficient dynamic algorithms for problems in $\Sigma_{i-1}$.   The details appear in~\ref{sec:PH}.

\begin{example}[$k$- and (<$k$)-edge connectivity]\label{ex:intro:k-con in Pi2}
In the {\em dynamic $k$-edge connectivity} problem,  an $n$-node graph $G = (V, E)$ and a parameter $k = O(\polylog (n))$ is given at the time of preprocessing.  Each update is an edge insertion or deletion in $G$.  We want an algorithm to output YES if and only if $G$ has connectivity {\em at least $k$}, i.e. removing at most $k-1$ edges will {\em not} disconnect $G$. We claim that this problem is in $\dPi_2$. To avoid dealing with $\dcoNP$, we consider the complement of this problem called {\em dynamic (<$k$)-edge connectivity}, where $x=YES$ if and only if $G$ has connectivity less than $k$. We show that  (<$k$)-edge connectivity is in $\dSigma_2$.

We already argued in~\ref{ex:intro:connectivity in NP} that dynamic connectivity is in $\dNP=\dSigma_1$. Assuming that there exists an efficient (i.e. polylogarithmic-update-time) algorithm  $\cal A$ for dynamic connectivity, we will show that  (<$k$)-edge connectivity is in $\dNP$. 
Consider the following verifier $\cal V$.  After every update in $G$, the verifier $\cal V$ reads a proof that is supposed to be a set $S \subseteq E$ of at most $k-1$ edges. $\cal V$ then  sends the update to $\cal A$ and also tells $\cal A$ to delete the edges in $S$ from $G$. If $\cal A$ says that $G$ is not connected at this point, then the verifier $\cal V$ outputs $x=YES$ with reward $y=1$; otherwise, the verifier $\cal V$  outputs $x=NO$ with reward $y=0$. Finally, $\cal V$ tells $\cal A$ to add the edges in $S$ back in $G$. 

Observe that if $G$ has connectivity less than $k$ and the verifier always receives a proof that maximizes the reward, then the proof will be a set of edges disconnecting the graph and $\cal V$ will answer YES. Otherwise, no proof can make $\cal V$ answer YES. Thus the dynamic (<$k$)-edge connectivity problem is in $\dNP$ if $\cal A$ exists. In other words, the problem is in $\dSigma_2$.
\end{example}

By arguments similar to the above example, we can show that other problems such as Chan's subset union and small diameter are in $\dPH$; see~\ref{table:PH} and~\ref{sec:PH:examples} 
for more.

The theorem that plays an important role in our main conclusion (\ref{thm:intro:results_summary}) is the following. 

\begin{theorem}\label{thm:intro:PH collapse}
If $\dP=\dNP$, then $\dPH=\dP$. 
\end{theorem}

To get an idea how to proof the above theorem, observe that if $\dP=\dNP$, then $\cal A$ in Example~\ref{ex:intro:k-con in Pi2} exists and thus dynamic (<$k$)-edge connectivity are in $\dSigma_1=\dNP$ by the argument in~\ref{ex:intro:connectivity in NP}; consequently, it is in $\dP$! This type of argument can be extended to all other problems in $\dPH$. 

\subsection{Other Results and Remarks}
\label{sec:contribution:other:results}

In previous subsections, we have stated two complexity results, namely $\dNP$-completeness/hardness and the collapse of $\dPH$ when $\dP=\dNP$. With right definitions in place, it is not a surprise that more can be proved. For example, we  obtain the following results:

\begin{enumerate}%
	\item \label{thm:intro:NPcoNP} If $\dNP\subseteq \dcoNP$, then $\dPH=\dNP\cap \dcoNP$.
	\item \label{thm:intro:MAcoMA} If $\dNP\subseteq  \dMA \cap \dcoMA$, then $\dPH \subseteq \dMA \cap \dcoMA$.
\end{enumerate}
Here, $\dcoNP$, $\dMA$, and $\dcoMA$ are analogues of complexity classes coNP, AM, and coAM. %
 See~\ref{th:PH:collapse:next,lem:if NP in coMA} for more details.

While the coarse-grained complexity results in this paper are mostly resource-centric (in contrast to fine-grained complexity results that are usually centered around problems), we also show that this approach is helpful for understanding the complexity of specific problems as well, in the form of {\em non-reducibility}.
In particular, the following results  are  shown in~\ref{part:connectivity}:%

\begin{enumerate}%
	
	\item \label{item:intro:non_reducibility_1}Assuming $\dPH\neq \dNP\cap \dcoNP$, the two statements {\em cannot} hold at the same time.
	\begin{enumerate}%
		\item Connectivity is in $\dcoNP$. (This would be the case if it is in $\dP$.)
		\item One of the following problems is $\dNP$-hard: approximate minimum spanning forest (MSF), $d$-weight MSF,%
		bipartiteness, 
		and $k$-edge connectivity. 
	\end{enumerate}
	
	\item $k$-edge connectivity is in $\dMA \cap \dcoMA$. Consequently,  assuming  $\dPH \not\subseteq \dMA \cap \dcoMA$, $k$-edge connectivity cannot be $\dNP$-hard. 
\end{enumerate}

Note that both ${\sf PH}\neq {\sf NP}\cap {\sf coNP}$ and ${\sf PH} \not\subseteq {\sf AM} \cap {\sf coAM}$ are widely believed in the static setting since refuting them means collapsing ${\sf PH}$. While we can show that $\dPH$ would also collapse if $\dPH= \dNP\cap \dcoNP$, it remains open whether this is the case for $\dPH \subseteq \dMA \cap \dcoMA$; in particular is $\dPH \supseteq\dMA \cap \dcoMA$?

When a problem $Y$ cannot be $\dNP$-hard, there is no efficient reduction from an $\dNP$-hard problem $X$ to  $Y$, where an efficient reduction is roughly to a way to handle each update for problem $X$ by making polylogarithmic number of updates to an algorithm for $Y$ (such reduction would make $Y$ an $\dNP$-hard problem). Consequently, this rules out efficient  reductions from dynamic OV, since it is $\dNP$-complete. As a result, this rules out a common way  to prove lower bounds based on SETH, since previously this was mostly done via reductions from dynamic OV \cite{AbboudW14}. (A lower bound for dynamic diameter is among a very few exception \cite{AbboudW14}.)

\subsection{Relationship to Fine-Grained Complexity} 
\label{sec:contribution:fine:grained}

As noted earlier, it turns out that the dynamic OV problem is $\dNP$-complete. Since most previous reductions from SETH to dynamic problems (in the word RAM model) are in fact reductions from dynamic OV~\cite{AbboudW14}, and since any reduction in the word RAM model applies also in the (stronger) cell-probe model, we get many $\dNP$-hardness results for free. In contrast, our results above  imply that the following two statements are equivalent: (i)   ``problem $\Pi$ cannot be $\dNP$-hard'' and (ii) ``there is no {\em efficient} reduction from dynamic OV to $\Pi$'', where ``efficient reductions'' are reduction that only polynomially blow up the instance size (all reductions in \cite{AbboudW14} are efficient). In other words, we may not expect reductions from SETH that are similar to the previous ones for $k$-edge connectivity, bipartiteness, etc. 

Finally, we emphasize that the coarse-grained approach should be viewed as a complement of the fine-grained approach, as the above results exemplify. We do not expect to replace results from one approach by those from another.

\subsection{Complexity classes for dynamic problems in the word RAM model}
\label{sec:contribution:word:RAM}

As an aside, we managed to define complexity classes and completeness results for dynamic problems in the word RAM model as well. 
We refer to $\dP$ and  $\dNP$ as $\dPram$ and $\dNPram$ in the word-RAM model. One caveat is that for technical reasons we need to allow for quasipolynomial preprocessing time and space while defining the complexity classes $\dPram$ and $\dNPram$.  See~\ref{part:RAM} for more details.

\section{Related Work}
\label{sec:related:work}

There are several previous attempts to classify dynamic problems.
First, there is a line of works called ``dynamic complexity theory''
(see e.g. \cite{DattaKMSZ15,WeberS05,SchwentickZ16}) where the general
question asks whether a dynamic problem is in the class called $\mbox{DynFO}$.
Roughly speaking, a problem is in $\mbox{DynFO}$ if it admits a dynamic
algorithm expressible by a first-order logic. This means, in particular, that given
an update, such algorithm runs in $O(1)$ parallel time, but might
take arbitrary $\mbox{poly}(n)$ works when the input size is $n$.
A notion of reduction is defined and complete problems of $\mbox{DynFO}$
and related classes are proven in \cite{HesseI02,WeberS05}.
However, as the total work of algorithms from this field can be large
(or even larger than computing from scratch using sequential algorithms), they do
not give fast dynamic algorithms in our sequential setting. 
Therefore, this setting is somewhat irrelevant to our setting.

Second, a problem called the \emph{circuit evaluation} problem has been shown to be
complete in the following sense. First, it is in $\mathsf{P}$ (the
class of static problems). Second, if the dynamic version of circuit
evaluation problem, which is defined as $\dDNF$ where a DNF-formula
is replaced with an arbitrary circuit, admits a dynamic algorithm
with polylogarithmic update time, then for any static problem $L\in\mathsf{P}$,
a dynamic version of $L$ also admits a dynamic algorithm with polylogarithmic
update time. This idea is first sketched informally since 1987 by
Reif \cite{Reif87}. Miltersen et al. \cite{MiltersenSVT94} then
formalized this idea and showed that other $\mathsf{P}$-complete
problems listed in \cite{miyano1990list,greenlaw1991compendium} also
are complete in the above sense.\footnote{But they also show that this is not true for all $\mathsf{P}$-complete
	problems.} The drawback about this completeness result is that the dynamic
circuit evaluation problem is extremely difficult. 
Similar to the case for static problems that reductions from $\mathsf{EXP}$-complete problems to problems in $\mathsf{NP}$ are unlikely, 
reductions from the dynamic circuit evaluation problem to other natural dynamic problems studied in the field seem unlikely. 
Hence, this does not give a framework for
proving hardness for other dynamic problems. 

Our result can be viewed as a more fine-grained completeness
result than the above. As we show that a very special case of the
dynamic circuit evaluation problem which is $\dDNF$ is already a
complete problem. An important point is that $\dDNF$ is simple enough
that reductions to other natural dynamic problems are possible. 

Finally, Ramalingam and Reps \cite{RamalingamR96} classify dynamic problems
according to some measure,\footnote{They measure the complexity dynamic algorithms by comparing the update
	time with the size of \emph{change in input and output }instead of
	the size of input itself.} but did not give any reduction and completeness result.

\section{Future Directions}\label{sec:open}

One byproduct of our paper is a way to prove non-reducibility.
It is interesting to use this method to shed more light on the hardness of other dynamic problems. To do so, it suffices to show that such problem is in  $\dMA \cap \dcoMA$ (or, even better, in $\dNP\cap \dcoNP$). One particular problem is whether connectivity is in $\dNP\cap \dcoNP$. It is known to be in $\dMA \cap \dcoMA$ due to the randomized algorithm of Kapron~et~al \cite{KapronKM13}. It is also in $\dNP$ (see~\ref{ex:intro:connectivity in NP}). The main question is whether it is in  $\dcoNP$. (Techniques from \cite{NanongkaiSW17,Wulff-Nilsen17,NanongkaiS17} almost give this, with verification time $n^{o(1)}$  instead of polylogarithmic.) Having connectivity in  $\dNP\cap \dcoNP$ would be a strong evidence that it is in $\dP$, meaning that it admits a deterministic algorithm with polylogarithmic update time. Achieving such algorithm will be a major breakthrough. Another specific question is whether the promised version of the $(2-\epsilon)$ approximate matching problem is in $\dMA\cap \dcoMA$. This would rule out efficient reductions from dynamic OV to this problem. Whether this problem admits a randomized algorithm with polylogarithmic update time is a major open problem. 
Other problems that can be  studied in this direction include approximate minimum spanning forest (MSF), $d$-weight MSF,%
bipartiteness, 
dynamic set cover, dominating set, and $st$-cut.

It is also very interesting to rule out efficient reductions from the following variant of the OuMv conjecture: At the preprocessing, we are given a boolean $n\times n$ matrix $M$ and boolean $n$-dimensional row and column vectors $u$ and $v$. Each update changes one entry in either $u$ or $v$. We then have to output the value of $uMv$. Most lower bounds that are hard under the OMv conjecture \cite{OMV} are via efficient reductions from this problem. It is interesting to rule out such efficient reductions since SETH and OMv are two conjectures that imply most lower bounds for dynamic problems.

Now that we can prove completeness and relate some basic complexity classes of dynamic problems, one big direction to explore is whether more results from coarse-grained complexity for static problems can be reconstructed for dynamic problems. Below are a few samples.

\begin{enumerate}%
\item {\em Derandomization:} Making dynamic algorithms deterministic is an important issue. Derandomization efforts have so far focused on specific problems (e.g. \cite{NanongkaiS17,NanongkaiSW17,BernsteinC16,BernsteinC17,Bernstein17,BhattacharyaHI18,BhattacharyaCHN18-coloring,BhattacharyaHN16}). Studying this issue via the class $\dBPP$ might lead us to the more general understanding. 
For example, the {\em Sipser-Lautermann theorem} \cite{Sipser83a,Lautemann83} states that $BPP\subseteq \Sigma _{2}\cap \Pi _{2}$, Yao \cite{Yao82a} showed that the existence of some pseudorandom generators would imply that $P = BPP$, and Impagliazzo and Wigderson \cite{ImpagliazzoW97} suggested that $BPP=P$ (assuming that any problem in $E =DTIME(2^{O(n)})$ has circuit complexity $2^{\Omega(n)}$).  We do not know anything similar to these for dynamic problems.

\item {\em NP-Intermediate:} Many static problems (e.g. graph isomorphism and factoring) are considered good candidates for being NP-intermediate, i.e. being neither in P nor NP-complete. This paper leaves many natural problems in $\dNP$ unproven to be $\dNP$-complete. Are these problems in fact $\dNP$-intermediate? The first step towards this question might be proving an analogues of {\em Ladner's theorem} \cite{Ladner75}, i.e. that an $\dNP$-intermediate dynamic problem exists, assuming $\dP\neq \dNP$. 
It is also interesting to prove the analogue of the {\em time-hierarchy theorems}, i.e. that with more time, more dynamic problems can be solved. (Both theorems are proved by diagonalization in the static setting.)

\item This work and lower bounds from fine-grained complexity has focused mostly on {\em decision}  problems. There are also {\em search dynamic problems}, which always have valid solutions, and the challenge is how to maintain them. These problems include maximal matching, maximal independent set, minimal dominating set, coloring vertices with  $(\Delta+1)$ or more colors, and coloring edges with $(1+\epsilon)\Delta$ or more colors, where $\Delta$ is the maximum degree (e.g. \cite{BaswanaGS15,BhattacharyaCHN18-coloring,AssadiOSS18-MIS,SolomonW18-coloring,DuZ18,GuptaK18,OnakSSW18}).  These problems  do not seem to correspond to any decision problems. Can we define complexity classes for these problems and argue that some of them might not admit polylogarithmic update time? Analogues of TFNP and its subclasses (e.g. PPAD) might be helpful here.
\end{enumerate}

\noindent There are also other concepts that have not been discussed in this paper at all, such as interactive proofs, probabilistically checkable proofs (PCP), counting problems (e.g. Toda's theorem), relativization and other barriers. Finally, in this paper we did not discuss amortized update time. It is a major open problem whether similar results, especially an analogue of NP-hardness, can be proved for algorithms with amortized update time.

\begin{table}
\footnotesize
\begin{tabular}{|>{\raggedright}p{0.15\textwidth}|>{\raggedright}p{0.15\textwidth}|>{\raggedright}p{0.15\textwidth}|>{\raggedright}p{0.35\textwidth}|>{\centering}p{0.15\textwidth}|}
\hline 
\textbf{Dynamic Problems} & \textbf{Preprocess} & \textbf{Update} & \textbf{Queries} & \textbf{Ref.}\tabularnewline
\hline 
\hline 
\multicolumn{5}{|l|}{\textbf{Numbers}}\tabularnewline
\hline 
{Sum/max} & \multirow{2}{0.15\textwidth}{{a set $S$ of numbers}} & \multirow{2}{0.15\textwidth}{{insert/delete a number in $S$}} & {return $\sum_{x\in S}x$ or $\max_{x\in S}x$} & \multirow{2}{0.15\textwidth}{{Binary search trees}}\tabularnewline
\cline{1-1} \cline{4-4} 
{Predecessor} &  &  & {given $x$, return the maximum $y\in S$ where $y\le x$.} & \tabularnewline
\hline 
\multicolumn{5}{|l|}{\textbf{Geometry}}\tabularnewline
\hline 
{2-dimensional range counting} & \multirow{2}{0.15\textwidth}{{a set $S$ of points on a plane}} & {insert/delete a point in $S$} & {given $[x_{1},x_{2}]\times[y_{1},y_{2}]$, return $|S\cap([x_{1},x_{2}]\times[y_{1},y_{2}])|$} & {\cite[after Theorem 7.6.3]{Overmars}}\tabularnewline
\cline{1-1} \cline{3-5} 
{Incremental planar nearest neighbor} &  & {insert a point to $S$} & {given a point $q$, return $p\in S$ which is closest to
$q$} & {\cite[Theorem 7.3.4.1]{Overmars}}\tabularnewline
\hline 
{Vertical ray shooting} & {a set $S$ of segments on a plane} & {insert/delete a segment in $S$} & {given a point $q$, return the segment immediately above
$q$} & {\cite[Theorem 3.7]{ChanN15}}\tabularnewline
\hline 
\multicolumn{5}{|l|}{\textbf{Graphs}}\tabularnewline
\hline 
\multirow{3}{0.15\textwidth}{{Dynamic problems on forests}} & \multirow{3}{0.15\textwidth}{{a forest $F$}} & \multirow{2}{0.15\textwidth}{{insert/delete an edge in $F$ s.t. $F$ remains a forest}} & {given two nodes $u$ and $v$, decide if $u$ and $v$ are
connected in $F$} & \multirow{3}{0.15\textwidth}{{\cite{SleatorT81,HenzingerK99,AlstrupHLT05}}}\tabularnewline
\cline{4-4} 
 &  &  & {given a node $u$, return the size of the tree containing
$u$} & \tabularnewline
\cline{3-4} 
 &  & {many more kinds of updates} & {many more kinds of query} & \tabularnewline
\hline 
{Connectivity on plane graphs} & \multirow{2}{0.15\textwidth}{{a plane graph $G$ (i.e. a planar graph on a fixed embedding)}} & \multirow{2}{0.15\textwidth}{{insert/delete an edge in $G$ such that $G$ has no crossing
on the embedding }} & {given two nodes $u$ and $v$, decide if $u$ and $v$ are
connected in $G$} & {\cite{Frederickson85,EppsteinITTWY92}}\tabularnewline
\cline{1-1} \cline{4-5} 
{2-edge connectivity on plane graphs} &  &  & {given two nodes $u$ and $v$, decide if $u$ and $v$ are
2-edge connected in $G$} & {\cite{Frederickson97}}\tabularnewline
\hline 
{$(2+\epsilon)$-approx. size of maximum matching} & {a general graph $G$} & {insert/delete an edge in $G$ } & {decide whether the size of maximum matching is at most $k$
or at least $(2+\epsilon)k$ for some $k$ and constant $\epsilon>0$} & {\cite{BhattacharyaHN17}}\tabularnewline
\hline 
\end{tabular}{\footnotesize \par}

\caption{Problems in $\protect\dP$. Some problems are strictly promise problems, but our class can be extended easily to include them (see \Cref{sec:promise}). }
\label{table:P}
\end{table}

\begin{table}
\footnotesize
\begin{tabular}{|>{\raggedright}p{0.15\textwidth}|>{\raggedright}p{0.15\textwidth}|>{\raggedright}p{0.15\textwidth}|>{\raggedright}p{0.35\textwidth}|>{\centering}p{0.15\textwidth}|}
\hline 
\textbf{Dynamic Problems} & \textbf{Preprocess} & \textbf{Update} & \textbf{Queries} & \textbf{Ref.}\tabularnewline
\hline 
\hline 
{Connectivity } & \multirow{2}{0.15\textwidth}{{a general undirected unweighted graph $G$}} & \multirow{2}{0.15\textwidth}{{insert/delete an edge in $G$ }} & {given two nodes $u$ and $v$, decide if $u$ and $v$ are
connected in $G$} & {\ref{prop:conn in NP}}\tabularnewline
\cline{1-1} \cline{4-5} 
{$(1+\epsilon)$-approx. size of maximum matching} &  &  & {decide whether the size of maximum matching is at most $k$
or at least $(1+\epsilon)k$ for some $k$ and constant $\epsilon>0$} & \ref{prop:matching in NP}\tabularnewline
\hline 
{Subgraph detection} & {a general graph $G$ and $H$ where $|V(H)|=\mbox{polylog}(|V(G)|)$} & {insert/delete an edge in $G$ } & {decide whether $H$ is a subgraph of $G$} & \multirow{6}{0.15\textwidth}{{\ref{prop:small cert in NP}}}\tabularnewline
\cline{1-4} 
{$uMv$ (entry update)} & {$u,v\in\{0,1\}^{n}$ and $M\in\{0,1\}^{n\times n}$} & {update an entry of $u$ or $v$} & {decide whether $u^{T}Mv=1$ (multiplication over Boolean
semi-ring).} & \tabularnewline
\cline{1-4} 
{3SUM} & {a set $S$ of numbers} & {insert/delete a number in $S$} & {decide whether there is $a,b,c\in S$ where $a+b=c$} & \tabularnewline
\cline{1-4} 
{Planar nearest neighbor} & {a set $S$ of points on a plane} & {insert a point to $S$} & {given a point $q$, return $p\in S$ which is closest to
$q$} & \tabularnewline
\cline{1-4} 
Erikson's problem \cite{Patrascu10} & a matrix $M$ & choose a row or a column and increment all number of such row or column & given $k$, is the maximum entry in $M$ at least $k$? & \tabularnewline
\cline{1-4} 
Langerman's problem \cite{Patrascu10} & an array $A$ & given $(i,x)$, set $A[i]=x$ & is there a $k$ such that $\sum_{i=1}^{k}A[i]=0$? & \tabularnewline
\hline 
\end{tabular}{\footnotesize \par}

\caption{Problems in $\protect\dNP$ that are not known to be in $\protect\dP$. Some problems are strictly promise problems, but our class can be extended easily to include them (see \Cref{sec:promise}). }
\label{table:BPP-NP}
\end{table}

\begin{table}
\footnotesize
\begin{tabular}{|>{\raggedright}p{0.2\textwidth}|>{\raggedright}p{0.2\textwidth}|>{\raggedright}p{0.2\textwidth}|>{\raggedright}p{0.35\textwidth}|}
\hline 
\textbf{Dynamic Problems} & \textbf{Preprocess} & \textbf{Update} & \textbf{Queries}\tabularnewline
\hline 
\hline 
{Pagh's problem with emptiness query \cite{AbboudW14}} & {A collection $X$ of sets $X_{1},\dots,X_{k}\subseteq[n]$} & {given $i,j$, insert $X_{i}\cap X_{j}$ into $X$} & {given $i$, is $X_{i}=\emptyset$?}\tabularnewline
\hline 
{Chan's subset union problem \cite{AbboudW14}} & {A collection of sets $X_{1},\dots,X_{n}\subseteq[m]$. A
set $S\subseteq[n]$.} & {insert/deletion an element in $S$} & {is $\cup_{i\in S}X_{i}=[m]$?}\tabularnewline
\hline 
{Single source reachability Count (\#$s$-reach)} & {a directed graph $G$ and a node $s$} & {insert/delete an edge} & {count the nodes reachable from $s$.}\tabularnewline
\hline 
{2 Strong components (SC2)} & {a directed graph $G$ } & {insert/delete an edge} & {are there more than 2 strongly connected components?}\tabularnewline
\hline 
{$st$-max-flow} & {a capacitated directed graph $G$ and nodes $s$ and $t$} & {insert/delete an edge} & {the size of $s$-$t$ max flow.}\tabularnewline
\hline 
{Subgraph global connectivity} & {a fixed undirected graph $G$} & {turn on/off a node} & {is a graph induced by turned on nodes connected?}\tabularnewline
\hline 
{3 vs. 4 diameter} & {an undirected graph $G$} & {insert/delete an edge} & {is a diameter of $G$ $3$ or $4$?}\tabularnewline
\hline 
{$ST$-reachability} & {a directed graph $G$ and sets of node $S$ and $T$} & {insert/delete an edge} & {is there $s\in S$ and $t\in T$ where $s$ can reach $t$? }\tabularnewline
\hline 
\end{tabular}{\footnotesize \par}
\caption{Problems that are $\protect\dNP$-hard. The proof is discussed in \Cref{sub:list NP hard}.}
\label{main:table:NP hard}
\label{table:NP hard}
\end{table}

\begin{table}
\footnotesize	
	
	\begin{tabular}{|>{\centering}p{0.15\textwidth}|>{\centering}p{0.15\textwidth}|>{\centering}p{0.13\textwidth}|>{\centering}p{0.35\textwidth}|>{\centering}p{0.15\textwidth}|}
		\hline 
		\textbf{Dynamic Problems} & \textbf{Preprocess} & \textbf{Update } & \textbf{Queries} & \textbf{Ref.}\tabularnewline
		\hline 
		\hline 
		{Small dominating set} & {a graph $G$} & {insert/delete an edge} & {Is there a dominating set of size at most $k$?} & {~\ref{lem:small dominating set}} \tabularnewline
		\hline 
		{Small vertex cover} & {a graph $G$} & {insert/delete an edge} & {Is there a vertex cover of size at most $k$?} & {~\ref{lem:small vertex cover}} \tabularnewline
		\hline 
		{Small maximal independent set } & {a graph $G$} & {insert/delete an edge} & {Is there a maximal independent set of size at most $k$?} & {~\ref{lem:small maximal independent set}} \tabularnewline
		\hline 
		{Small maximal matching} & {a graph $G$} & {insert/delete an edge} & {Is there a maximal matching of size at most $k$?} & {~\ref{lem:small maximal matching}} \tabularnewline
		\hline 
		{Chan\textquoteright s Subset Union Problem} & {a collection of sets $X_{1},\dots,X_{n}$ from universe $[m]$,
			and a set $S\subseteq[n]$} & {insert/delete an index in $S$} & {is $\cup_{i\in S}X_{i}=[m]$?} & {~\ref{lem:subset union}} \tabularnewline
		\hline 
		{3 vs. 4 diameter} & {a graph $G$ } & {insert/delete an edge} & {Is the diameter of $G$ 3 or 4?} & {~\ref{lem:small diameter}} \tabularnewline
		\hline 
		{Euclidean $k$-center} & {a point set $X\subseteq\mathbb{R}^{d}$ and a threshold $T\in\mathbb{R}$} & {insert/delete a point} & {Is there a set $C\subseteq X$ where $|C| \leq k$ and $\max_{u\in X}\min_{v\in C}d(u,v)\le T$} & {~\ref{lem:small k center}} \tabularnewline
		\hline 
		{$k$-edge connectivity} & {a graph $G$} & {insert/delete an edge} & {Is $G$ $k$-edge connected?} & {~\ref{lem:k edge connectivity}} \tabularnewline
		\hline 
	\end{tabular}
	
	\caption{Problems in $\protect\dPH$ that are not known to be in $\protect\dNP$.
		The parameter $k$ in every problem must be at most $\mbox{polylog}(n)$
		where $n$ is the size of the instance.}
	
	\label{table:PH}
	
\end{table}

\section{An Overview of the $\dNP$-Completeness Proof}
\label{sec:overview}

In this section, we present an overview of one of our main technical contributions (the proof of~\ref{thm:intro:DNF}) at a finer level of granularity.   In order to explain the main technical insights we focus on a nonuniform model of computation called the \emph{bit-probe} model, which has been studied
since the 1970's~\cite{Fredman78, Miltersen99cellprobe}.

\subsection{Dynamic Complexity Classes $\dP$ and $\dNP$}
\label{main:sec:class}

We begin by reviewing (informally) the concepts of a dynamic problem and an algorithm in the bit-probe model. See~\ref{sec:model} for a formal description.
Consider any dynamic problem $\D_n$. Here, the subscript $n$ serves as a reminder that the bit-probe model is nonuniform and it also indicates that each  instance $I$ of this problem can be specified using $n$ bits. We will will mostly be concerned with dynamic {\em decision} problems, where the {\em answer} $\D_n(I) \in \{0,1\}$ to every instance $I$ can be specified using a single bit. We say that $I$ is an YES instance if $\D_n(I) = 1$, and a NO instance if $\D_n(I) = 0$. An algorithm $\A_n$ for this dynamic problem $\D_n$ has access to a memory $\mem_n$, and the total number of bits available in this memory is called the {\em space complexity} of $\A_n$.  The algorithm $\A_n$ works in steps $t = 0, 1, \ldots,$ in the following manner.

\medskip
\noindent {\em Preprocessing:} At step $t = 0$ (also called the preprocessing step), the algorithm gets a {\em starting instance} $I_0 \in \D_n$ as input. Upon receiving this input, it initializes the bits in its memory $\mem_n$ and then it {\em outputs} the {\em answer} $\D_n(I_0)$ to the {\em current instance} $I_0$.

\medskip
\noindent {\em Updates:} Subsequently, at each step $t \geq 1$, the algorithm gets an {\em instance-update} $(I_{t-1}, I_t)$ as input. The sole purpose of this instance-update is to change the current instance  from $I_{t-1}$ to $I_t$. Upon receiving this input, the algorithm probes (reads/writes) some bits in the memory $\mem_n$, and then outputs the answer $\D_n(I_t)$ to the  current instance $I_t \in \D_n$. The {\em update time} of  $\A_n$ is the maximum number of bit-probes it needs to make in $\mem_n$ while handling an instance-update.

\medskip
One way  to visualize the above description as follows. An adversary keeps constructing an {\em instance-sequence} $(I_0, I_1, \ldots, I_k, \ldots)$  one step at a time. At each step $t$, the algorithm $\A_n$ gets the corresponding instance-update $(I_{t-1}, I_t)$, and at this point it is only aware of the prefix $(I_0, \ldots, I_t)$. Specifically, the algorithm does not know the future instance-updates. After receiving the instance-update at each step $t$, the algorithm has to output the answer to the current instance $\D_n(I_t)$. This framework is flexible enough to capture dynamic problems that allow for both {\em update} and {\em query} operations, because we can easily model a query operation as an instance-update. This fact is illustrated in more details in~\ref{exa:query}. 
Furthermore,  w.l.o.g. we assume that an instance-update in a dynamic problem $\D_n$ can be specified using $O(\log n)$ bits. See the discussion preceding~\ref{assume:logn update size} for a more detailed explanation of this phenomenon.

For technical reasons, we will work under the following assumption. This assumption will be implicitly present in the definitions of the complexity classes $\dP$ and $\dNP$ below.

\begin{assumption}
 \label{assume:poly:main}
A dynamic algorithm $\A_n$ for a dynamic problem $\D_n$ has to handle at most $\poly(n)$ many instance-updates.
 \end{assumption}

We  now define  the complexity class $\dP$.

\begin{defn}
[Class $\dP$]\label{main:def:P} A dynamic decision problem $\D_n$ is in $\dP$ iff there is an algorithm $\A_n$
solving $\D_n$ which has update time $O(\polylog (n))$ and space-complexity $O(\poly (n))$.
\end{defn}

In order to define the  class $\dNP$, we first introduce the notion of a {\em verifier} in~\ref{main:def:verifier}. Subsequently, we introduce the class $\dNP$ in~\ref{main:def:NP}. We have already discussed the intuitions behind these concepts in~\ref{sec:intro} after the statement of~\ref{def:intro:dNP}.

\begin{defn}[Dynamic verifier]
\label{main:def:verifier}
We say that a dynamic algorithm $\V_n$ with space-complexity $O(\poly (n))$ is a {\em verifier} for a dynamic decision problem $\D_n$ iff it works as follows.

\smallskip
\noindent {\em Preprocessing:} At step $t = 0$, the algorithm $\V_n$ gets a {\em starting instance} $I_0 \in \D_n$ as input, and it outputs an ordered pair $(x_0, y_0)$ where $x_0 \in \{0,1\}$ and $y_0 \in \{0,1\}^{\polylog(n)}$.

\smallskip
\noindent {\em Updates:} Subsequently, at each step $t \geq 1$, the algorithm $\V_n$ gets an {\em instance-update} $(I_t, I_{t-1})$ and a {\em proof} $\pi_t \in \{0,1\}^{\polylog (n)}$ as input, and it outputs an ordered pair $(x_t, y_t)$ where $x_t \in \{0,1\}$ and $y_t \in \{0,1\}^{\polylog (n)}$. The algorithm $\V_n$ has $O(\polylog (n))$ update time, i.e., it makes at most $O(\polylog (n))$ bit-probes in the memory during each step $t$. Note that the output $(x_t, y_t)$ depends on  the instance-sequence $(I_0, \ldots, I_t)$  and the {\em proof-sequence} $(\pi_1, \ldots, \pi_t)$ seen so far. 
\end{defn}

\begin{defn}
[Class $\dNP$]\label{main:def:NP} A decision problem $\D_n$ is in $\dNP$ iff it admits a verifier $\V_n$ which satisfy the following properties.  Fix any instance-sequence $(I_0, \ldots, I_k)$. Suppose that the verifier  $\V_n$ gets $I_0$ as input at step $t = 0$ and the ordered pair $((I_{t-1},I_t), \pi_t)$ as input at every step $t \geq 1$. Then:
\begin{enumerate}
\item  For every proof-sequence $(\pi_1, \ldots, \pi_k)$,  we have  $x_t = 0$ for each $t \in \{0, \ldots, k\}$ where $\D_n(I_t) = 0$.
\item If the proof-sequence $(\pi_1, \ldots, \pi_k)$ is {\em reward-maximizing} (defined below), then we have $x_t = 1$ for each $t \in \{0, \ldots, k\}$ with $\D_n(I_t) = 1$, 
\end{enumerate}
The proof-sequence $(\pi_1, \ldots, \pi_k)$ is {\em reward-maximizing} iff the following holds.  At each step $t \geq 1$, given the past history $(I_0, \ldots, I_t)$ and $(\pi_1, \ldots, \pi_{t-1})$, the proof $\pi_t$ is chosen in such a way that maximizes the value of $y_t$. We say that such a proof $\pi_t$ is {\em reward-maximizing}.
\end{defn}

Just as in the static setting, we can easily prove that $\dP \subseteq \dNP$ and we conjecture that $\dP \neq \dNP$. The big question left open in this paper is to resolve this conjecture.

\begin{cor}
\label{main:cor:pnp}
We have $\dP \subseteq \dNP$.
\end{cor}

\subsection{A complete problem in $\dNP$}
\label{main:sec:complete}

One of the main results in this paper shows that a natural  problem called {\em dynamic narrow DNF evaluation} (denoted by $\dDNF$) is $\dNP$-complete. Intuitively, this means that (a) $\dDNF \in \dNP$, and (b) if $\dDNF \in \dP$ then $\dP = \dNP$.\footnote{To be more precise, condition (b) means that  every problem in $\dP$ is $\dP$-reducible to $\dDNF$. See~\ref{sec:reduction} for a formal definition.
} We now give an informal description of this problem.  

\medskip
\noindent {\bf Dynamic narrow DNF evaluation ($\dDNF$):} An instance $I$ of this problem consists of a triple $(\Z, \C, \phi)$, where $\Z = \{z_1, \ldots, z_N\}$ is a set of $N$ variables,  $\C = \{ C_1, \ldots, C_M\}$ is a set of $M$ DNF clauses, and  $\phi : \Z \rightarrow \{0,1\}$ is an assignment of values to the variables. Each clause $C_j$ is a conjunction (AND) of at most $\polylog (N)$ literals, where each literal is of the form $z_i$ or $\neg z_i$ for some variable  $z_i \in \Z$. This is an YES instance if   at least one clause $C \in \C$  is true under the assignment $\phi$, and this is a NO instance if every clause in $\C$ is false under the assignment $\phi$. Finally, an {\em instance-update}  changes the assignment $\phi$ by flipping the value of exactly one variable in $\Z$.

It is easy to see that the above problem is in $\dNP$. Specifically, if the current instance is an YES instance, then a proof $\pi_t$ simply points to a specific clause $C_j \in \C$ that is true under the current assignment $\phi$. The proof $\pi_t$ can be encoded using $O(\log M)$ bits. Furthermore, since each clause contains at most $\polylog (N)$ literals, the verifier can check that the clause $C_j$ specified by the proof $\pi_t$ is true under the assignment $\phi$ in $O(\polylog (N))$ time. On the other hand, no proof can fool the verifier if the current instance is a NO instance (where every clause is false). All these observations can be formalized in a manner consistent with~\ref{main:def:NP}. We will prove the following theorem.

\begin{theorem}
\label{main:thm:NP:complete}
The $\dDNF$ problem described above is $\dNP$-complete.
\end{theorem}

In order to prove~\ref{main:thm:NP:complete}, we consider an intermediate dynamic problem called $\fDNF$.

\medskip
\noindent {\bf $\fDNF$:} An instance $I$ of $\fDNF$ consists of a tuple $(\Z, \C, \phi, \prec)$. Here, the symbols $\Z, \C$ and $\phi$ denote exactly the same objects as in the $\dDNF$ problem described above. In addition, the symbol $\prec$ denotes a total order on the set of clauses $\C$. The answer to this instance $I$ is defined as follows. If every clause in $\C$ is false under the current assignment $\phi$, then the answer to $I$ is $0$. Otherwise, the answer to $I$ is the {\em first clause}  $C_j \in \C$ according to the total order $\prec$ that is true under $\phi$. It follows that $\fDNF$ is {\em not} a decision problem.  Finally, as before, an instance-update for the $\fDNF$  changes the assignment $\phi$ by flipping the value of exactly one variable in $\Z$.

\medskip
 We prove~\ref{main:thm:NP:complete} as follows. (1) We first show that  $\fDNF$ is $\dNP$-hard. Specifically, if there is an algorithm for $\fDNF$ with polylog update time and polynomial space complexity, then $\dP = \dNP$. We explain this in more details in~\ref{main:sec:NP:hard}. (2) Using a standard binary-search trick, we show that there exists an $O(\polylog (n))$ time {\em reduction} from $\fDNF$ to $\dDNF$. Specifically, this means that if $\dDNF \in \dP$, then we can use an algorithm for $\dDNF$ as a subroutine to design an algorithm for $\fDNF$ with polylog update time and polynomial space complexity.~\ref{main:thm:NP:complete} follows from (1) and (2), and the observation that $\dDNF \in \dNP$.

\subsubsection{$\dNP$-hardness of $\fDNF$}
\label{main:sec:NP:hard}

Consider any dynamic decision problem $\D_n \in \dNP$. Thus, there exists a verifier $\V_n$ for $\D_n$ with the properties mentioned in~\ref{main:def:NP}. Throughout~\ref{main:sec:NP:hard}, we assume that  there is an algorithm for $\fDNF$ with polynomial space complexity and polylog update time. Under this assumption, we will show that there exists an algorithm $\A_n$ for $\D_n$ that also has $O(\poly (n))$ space complexity and $O(\polylog (n))$ update time. This will imply the $\dNP$-hardness of $\fDNF$.

\medskip
\noindent {\bf The high-level strategy:} The algorithm $\A_n$ will use the following two {\em subroutines}: (1) The verifier $\V_n$ for $\D_n$ as specified in~\ref{main:def:verifier} and~\ref{main:def:NP}, and (2) A dynamic algorithm $\A^*$ that solves the $\fDNF$ problem with polylog update time and polynomial space complexity.   

To be more specific, consider any instance-sequence $(I_0, \ldots, I_k)$ for the problem $\D_n$. At step $t = 0$, after receiving the starting instance $I_0$, the algorithm $\A_n$ calls the subroutine $\V_n$ with the same input $I_0$. The subroutine $\V_n$ returns an ordered pair $(x_0, y_0)$. At this point, the algorithm $\A_n$ outputs the bit $x_0$. Subsequently, at each step $t \geq 1$, the algorithm $\A_n$ receives the instance-update $(I_{t-1}, I_t)$ as input. It then calls the subroutine $\A^*$ in such a manner which ensures that $\A^*$ returns a reward-maximizing proof $\pi_t$ for the verifier $\V_n$ (see~\ref{main:def:NP}). This is explained in more details below. The algorithm $\A_n$ then calls the verifier $\V_n$ with the input $((I_{t-1}, I_t), \pi_t)$, and the verifier returns an ordered pair $(x_t, y_t)$. At this point, the algorithm $\A_n$ outputs the bit $x_t$.

To summarize, the algorithm $\A_n$ uses $\A^*$ as a dynamic subroutine to construct a reward-maximizing proof-sequence $(\pi_1, \ldots, \pi_k)$ -- one step at a time. Furthermore, after each step $t \geq 1$, the algorithm $\A_n$ calls the verifier $\V_n$ with the input $((I_{t-1}, I_t), \pi_t)$. The verifier $\V_n$ returns $(x_t, y_t)$, and the algorithm $\A_n$ outputs $x_t$. Item (1) in~\ref{main:def:NP} implies that the algorithm $\A_n$ outputs $0$ on all the NO instances (where $\D_n(I_t) = 0$). Since the proof-sequence $(\pi_1, \ldots, \pi_k)$ is reward-maximizing, item (2) in~\ref{main:def:NP} implies that the algorithm $\A_n$ outputs $1$ on all the YES instances (where $\D_n(I_t) = 1$). So the algorithm $\A_n$ always outputs the correct answer and solves the problem $\D_n$. We now explain how the algorithm $\A_n$ calls the subroutine $\A^*$, and then analyze the space complexity and update time of $\A_n$. The key observation is that we can represent  the verifier $\V_n$ as a collection of decision trees, and each root-to-leaf path in each of these  trees can be modeled as a DNF clause.%

\medskip
\noindent {\bf The decision trees that define the verifier $\V_n$:}
Let $\mem_{\V_n}$ denote the memory of the verifier $\V_n$. We assume that during each step $t \geq 1$, the instance-update $(I_{t-1}, I_t)$ is written in a designated region $\mem_{\V_n}^{(0)} \subseteq \mem_{\V_n}$ of the memory, and the proof $\pi_t$ is written in another designated region $\mem_{\V_n}^{(1)} \subseteq \mem_{\V_n}$ of the memory. Each bit in $\mem_{\V_n}$ can be thought of as a boolean variable $z \in \{0,1\}$. We view the region $\mem_{\V_n} \setminus \mem_{\V_n}^{(1)}$ as a collection of boolean variables $\Z = \{ z_1, \ldots, z_N \}$ and the contents of $\mem_{\V_n} \setminus \mem_{\V_n}^{(1)}$ as an assignment $\phi : \Z \rightarrow \{0,1\}$. For example, if $\phi(z_j) = 1$ for some $z_j \in \Z$, then it means that the bit $z_j$ in $\mem_{\V_n} \setminus \mem_{\V_n}^{(1)}$ is currently set to $1$. Upon receiving an input $((I_{t-1}, I_t), \pi_t)$, the verifier $\V_n$ makes some probes in $\mem_{\V_n} \setminus \mem_{\V_n}^{(1)}$ according to some pre-defined procedure, and then outputs an answer $(x_t, y_t)$. This procedure can be modeled as a decision tree $T_{\pi_t}$. Each internal node (including the root) in this decision tree is either a "read" node  or a "write" node. Each read-node has two children and is labelled with a variable $z \in \Z$. Each write-node has one child and is labelled with an ordered pair $(z, \lambda)$, where $z \in \Z$ and $\lambda \in \{0,1\}$. Finally, each leaf-node of $T_{\pi_t}$ is labelled with an ordered pair $(x, y)$, where $x \in \{0,1\}$ and $y \in \{0,1\}^{\polylog (n)}$.  Upon receiving the input $((I_{t-1}, I_t), \pi_t)$, the verifier $\V_n$ {\em traverses} this decision tree $T_{\pi_t}$. Specifically, it starts at the root of $T_{\pi_t}$, and then inductively applies the following steps until it reaches a leaf-node. 
\begin{itemize}
\item Suppose that it is currently at a read-node of $T_{\pi_t}$  labelled with $z \in \Z$. If $\phi(z) = 0$ (resp. $\phi(z) = 1$), then it goes to the left (resp. right) child of the node. On the other hand, suppose that it is currently at a write-node of $T_{\pi_t}$ which is labelled with $(z, \lambda)$. Then it writes $\lambda$ in the memory-bit $z$ (by setting $\phi(z) = \lambda$) and then moves on to the only child of this node. 
\end{itemize}
Finally, when it reaches a leaf-node, the verifier $\V_n$ outputs the corresponding label $(x, y)$. This is the way the verifier operates when it is called with an input $((I_{t-1}, I_t), \pi_t)$. The depth of the decision tree (the maximum length of any root-to-leaf path) is at most $\polylog (n)$, since as per~\ref{main:def:verifier} the verifier makes at most $\polylog (n)$ many bit-probes in the memory while handling any input.  

Each possible proof $\pi$ for the verifier can be specified using $\polylog (n)$ bits. Hence, we get a collection of $O(2^{\polylog (n)})$ many decision trees $\T = \{ T_{\pi} \}$ -- one tree $T_{\pi}$ for each possible input $\pi$. This collection of decision trees $\T$ completely characterizes the verifier $\V_n$.

\medskip
\noindent {\bf DNF clauses corresponding to a decision tree $T_{\pi_t}$:} Suppose that the proof $\pi$ is given as part of the input to the verifier during some update step. Consider any root-to-leaf path $P$ in a decision tree $T_{\pi}$. We can naturally associate a DNF clause $C_P$  corresponding to this path $P$. To be more specific, suppose that the path $P$ traverses a read-node labelled with $z \in \Z$ and then goes to its left (resp. right) child. Then we have a literal $\neg z$ (resp. $z$) in the clause $C_P$ that corresponds to this read-node.\footnote{W.l.o.g. we can assume that no two read nodes on the same path are labelled with the same variable.} The clause $C_P$ is the conjunction (AND) of these literals, and $C_P$ is true iff the verifier $\V_n$ traverses the path $P$ when $\pi$ is the proof given to it as part of the input. Let $\C = \{ C_P : P \text{ is a root-to-leaf path in some tree } T_{\pi} \in \T \}$ be the collection of all these DNF clauses. 

\medskip
\noindent {\bf Defining a total order $\prec$ over $\C$:} We now define a total order $\prec$ over $\C$ which satisfies the following property: Consider any two root-to-leaf paths $P$ and $P'$ in the collection of decision trees $\T$. Let $(x, y)$ and $(x', y')$ respectively denote the labels associated with the leaf nodes of the paths $P$ and $P'$. If $C_{P} \prec C_{P'}$, then $y \geq y'$.  Thus, the paths with higher $y$ values appear earlier in $\prec$. 

\medskip
\noindent {\bf Finding a reward-maximizing proof:} Suppose that $(I_0, \ldots, I_{t-1})$ is the instance-sequence of $\D_n$ received by $\A_n$ till now. By induction, suppose that  $\A_n$ has managed to construct a reward-maximizing proof-sequence $(\pi_1, \ldots, \pi_{t-1})$ till this point, and has fed this as input to the verifier $\V_n$ (which is used as a subroutine). At the present moment, suppose that $\A_n$ receives an instance-update $(I_{t-1}, I_t)$ as input. Our goal now is  to find a reward-maximizing proof $\pi_t$ at the current step $t$.

Consider the tuple $(\Z, \C, \phi, \prec)$ where $\Z = \mem_{\V_n} \setminus \mem_{\V_n}^{(1)}$ is the set of variables, $\C = \{ C_P : P \text{ is a root-to-leaf path in some decision tree } T_{\pi}\}$ is the set of DNF clauses,  the assignment $\phi : \Z \rightarrow \{0,1\}$ reflects the current contents of the memory-bits in $\mem_{\V_n} \setminus \mem_{\V_n}^{(1)}$, and $\prec$ is the total ordering over $\C$ described above.  Let $C_{P'} \in \C$ be the answer to this $\fDNF$ instance $(\Z, \C, \phi, \prec)$, and suppose that the path $P'$ belongs to the decision tree $T_{\pi'}$ corresponding to the proof $\pi'$. A moment's thought will reveal that $\pi_t = \pi'$ is the desired reward-maximizing proof at step $t$ we were looking for, because of the following reason. Let $(x', y')$ be the label associated with the leaf-node in $P'$. By definition, if the verifier gets the ordered pair $((I_{t-1}, I_t), \pi')$ as input at this point, then  it will traverse the path $P'$ in the decision tree $T_{\pi'}$ and return the ordered pair $(x', y')$. Furthermore, the path $P'$ comes first according to the total ordering $\prec$, among all the paths that can be traversed by the verifier at this point. Hence, the path $P'$ is chosen in such a way that maximizes $y'$, and accordingly, we conclude that $y_t = y'$ is a reward-maximizing proof at step $t$. 

\medskip
\noindent {\bf Wrapping up: Handling an instance-update $(I_{t-1}, I_t)$.} To summarize, when the algorithm $\A_n$ receives an instance-update $(I_{t-1}, I_t)$, it works as follows. It first writes down in the instance-update $(I_{t-1}, I_t)$ in $\mem_{\V_n}^{(0)}$ and accordingly updates the assignment $\phi : \Z \rightarrow \{0,1\}$. It then calls the subroutine $\A^*$ on the $\fDNF$ instance $(\Z, \C, \phi, \prec)$. The subroutine $\A^*$ returns a reward-maximizing proof $\pi_t$. The algorithm $\A_n$ then calls the verifier $\V_n$ as a subroutine with the ordered pair $((I_{t-1}, I_t), \pi_t)$ as input. The verifier updates at most $\polylog (n)$ many bits in $\mem_{\V_n}$  and returns an ordered pair $(x_t, y_t)$. The algorithm $\A_n$ now  updates the assignment $\phi : \Z \rightarrow \{0,1\}$ to ensure that it is synchronized with the current contents of  $\mem_{\V_n}$. This requires $O(\polylog (n))$ many calls to the subroutine $\A^*$ for the $\fDNF$ instance. Finally,  $\A_n$ outputs the bit $x_t \in \{0,1\}$.

\medskip
\noindent {\bf Bounding the update time of $\A_n$:} 
Notice that after each instance-update $(I_{t-1}, I_t)$, the algorithm $\A_n$ makes one call to the verifier $\V_n$ and at most $\polylog (n)$ many calls to $\A^*$. By~\ref{main:def:verifier}, the call to  $\V_n$ requires $O(\polylog (n))$ time. Furthermore, we have assumed that  $\A^*$ has polylog update time. Hence, each call to $\A^*$ takes $O(\polylog (N, M)) = O(\polylog (2^{\polylog (n)})) = O(\polylog (n))$ time. Since the algorithm $\A_n$ makes at most $\polylog (n)$ many calls to $\A^*$, the total time spent in these calls is still $O(\polylog (n))$. Thus, we conclude that  $\A_n$ has $O(\polylog (n))$ update time.

\renewcommand{\S}{\mathcal{S}}

\medskip
\noindent {\bf Bounding the space complexity of $\A_n$:} The space complexity of $\A_n$ is dominated by the space complexities of the subroutines $\V_n$ and $\A^*$. As per~\ref{main:def:verifier},   the verifier $\V_n$ has space complexity $O(\poly (n))$.

We next bound the memory space used by the subroutine $\A^*$. Note that in the $\fDNF$ instance, we have a DNF clause $C_P \in \C$ for every root-to-leaf path $P$ of every decision tree $T_{\pi}$. Since a proof $\pi$ consists of $\polylog (n)$ bits, there are at most $O(2^{\polylog (n)})$ many decision trees of the form $T_{\pi}$. Furthermore,  since every root-to-leaf path is of length at most $\polylog (n)$, each decision tree $T_{\pi}$ has at most $O(2^{\polylog (n)})$ many root-to-leaf paths. These two observations together imply that the set of clauses $\C$ is of size at most $O\left(2^{\polylog (n)} \cdot 2^{\polylog (n)}\right) = O(2^{\polylog (n)})$. Furthermore, as per~\ref{main:def:verifier} there are at most $O(\poly (n))$ many bits in the memory $\mem_{\V_n}$, which means that there are at most $O(\poly (n))$ many variables in $\Z$. Thus, the $\fDNF$ instance $(\Z, \C, \phi, \prec)$ is defined over a set of $N = \poly (n)$ variables and a set of $M = 2^{\polylog (n)}$ clauses (where each clause consists of at most $\polylog (n)$ many literals).  We have assumed that $\A^*$ has quasipolynomial space complexity. Thus, the total space needed by the subroutine $\A^*$ is $O(2^{\polylog (N, M)}) = O(2^{\polylog (n)})$. 

Unfortunately, the bound of $2^{\polylog (n)}$ is too large for us. Instead, we will like to have a space complexity of $O(\poly (n))$. Towards this end, we introduce a new subroutine $\S^*_n$ that acts as an {\em interface} between the subroutine $\A^*$ and the memory $\mem_{\A^*}$ used by $\A^*$ (see~\ref{sec:cor:np-hardness} for details). 
Specifically, as we observed in the preceding paragraph, the memory $\mem_{\A^*}$ consists of $2^{\polylog (n)}$ many bits and we cannot afford to store all these bits during the execution of the algorithm $\A_n$. The subroutine $\S^*_n$ has the nice properties that (a) it has space complexity  $O(\poly (n))$ and (b) it can still  return the content of a given bit in $\mem_{\A^*}$ in $O(\polylog (n))$ time. In other words, the subroutine $\S^*_n$ stores the contents of $\mem_{\A^*}$ in an {\em implicit manner}, and whenever the subroutine $\A^*$ wants to read/write a given bit in $\mem_{\A^*}$, it does that by calling the subroutine $\S^*_n$. This ensures that the overall space complexity of $\A^*$ remains $O(\poly (n))$. However, the subroutine $\S^*_n$ will be able perform its designated task with $\polylog (n)$ update time and $\poly (n)$ space complexity only if the algorithm $\A_n$ is required to handle at most $\poly (n)$ many instance-updates after the preprocessing step. This is why we need~\ref{assume:poly:main} while defining the complexity classes $\dP$ and $\dNP$.

\medskip
To summarize, we have shown that the algorithm $\A_n$  has polylog update time and polynomial space complexity. This implies that the $\fDNF$ problem is $\dNP$-hard.

\newpage
\part{Basic Complexity Classes (in the Bit-probe Model)}\label{part:formalization}

This part is organized as follows. In~\ref{sec:model}, we formally define the notion of a dynamic problem and a dynamic algorithm in the bit-probe model. In~\ref{sec:class}, we define the dynamic complexity classes $\dP$, $\dNP$ and $\dcoNP$. We also present multiple examples of dynamic problems belonging to these classes. See~\ref{sec:problem in rankNP} for more details. In~\ref{sec:reduction}, we formally define the notion of an ``efficient'' reduction between dynamic problems, and define the notion of a complete problem within any given dynamic complexity class. Finally, in~\ref{sec:PH} we introduce the dynamic polynomial hierarchy (denoted by $\dPH$), and show that $\dPH$ collapses to $\dP$ if $\dP = \dNP$. We refer the reader to~\ref{sec:PH:examples} for multiple examples of dynamic problems that happen to belong to $\dPH$ (but are not known to be in $\dNP$ or $\dcoNP$).

\section{The Models}

\label{sec:model}

We work with the bit-probe model of computation. In~\ref{sec:problem}, we formally define the notion of a dynamic problem in this model.  In~\ref{sec:model:dynamic}, we present a formal description of a bit-probe algorithm for solving a dynamic problem. 

\subsection{Dynamic Problems}

\label{sec:problem}

\paragraph{Static Problems.}
We begin by recalling the standard definition of a computational problem in the static setting. A \emph{problem}
is a function $\P:\{0,1\}^{*}\rightarrow\{0,1\}^{*}$
which maps each \emph{instance} $I\in\{0,1\}^{*}$ to an \emph{answer
}$\P(I)\in\{0,1\}^{*}$. We say that $\P$ is a \emph{decision
problem} iff the range of $\P$ is $\{0,1\}$. If $\P$ is a
decision problem, then we respectively refer to $\P^{-1}(0)$ and
$\P^{-1}(1)$ as the set of \emph{no instances} and \emph{yes instances}
of $\P$. 
\begin{example}
Let $\P:\{0,1\}^{*}\rightarrow\{0,1\}$ be the problem that,
given a planar graph $G$, decide whether $G$ is Hamiltonian. Then
$\P^{-1}(0)$ and $\P^{-1}(1)$  are the sets of
non-Hamiltonian planar graphs and Hamiltonian planar graphs,  respectively. 
\end{example}
For any integer $n\ge1$, let $\P_{n}:\{0,1\}^{n}\rightarrow\{0,1\}^{*}$
be obtained by restricting the domain of $\P$ to the set of all bit-strings
of length $n$. We say that $\P_{n}$ an $n$-slice of $\P$, and
we write $\P=\{\P_{n}\}_{n}$. We refer to each bit-string $I\in\{0,1\}^{n}$
as an \emph{instance} of $\P_{n}$.

\paragraph{Dynamic Problems.}
We define a \emph{dynamic problem} to be a graph structure imposed on instances. Formally, a dynamic
problem is an ordered pair $\D=(\P,\G)$, where $\P$ is a static
problem \emph{and} $\G=\{\G_{n}\}_{n}$ is a family of directed graphs
such that the node-set of each $\G_{n}$ is equal to the set of all
instances of $\P_{n}$. Thus, for each integer $n\geq1$, the directed
graph $\G_{n}=(\U_{n},\E_{n})$ has a node-set $\U_{n}=\{0,1\}^{n}$.\footnote{We use $\U_{n}$, instead of $\V_{n}$, to denote the set of nodes
of $\G_{n}$, since later on $\V$ will be used frequently for a ``verifier''.} We refer to the ordered pair $\D_{n}=(\P_{n},\G_{n})$ as the \emph{$n$-slice}
of $\D$, and we write $\D=\{\D_{n}\}_{n}$. Each $I\in\U_{n}$ is
called an \emph{instance} of $\D_{n}$. Each $(I,I')\in\E_{n}$ is
called an \emph{instance-update} of $\D_{n}$. We refer to the graph
$\G_{n}$ as the \emph{update-graph} of $\D_{n}$. We also call $\G$
the \emph{family of update-graphs} of $\D$ or simply the update-graphs
of $\D$. We will show the below definition often:
\begin{defn}
[Instance-sequence]We say that a tuple $(I_{0},\dots,I_{k})$ is an \emph{instance-sequence
of $\D$} iff $(I_{0},\dots,I_{k})$ is a directed path in the graph  $\G_{n}$ for some $n$.
\end{defn}
For each instance $I$ of $\D$, we write $\D(I)=\P(I)$ as the \emph{answer}
of $I$. We say that $\D=(\P,\G)$ is a \emph{dynamic decision} problem
iff $\P$ is a decision problem. From now, we usually use $\D$ to
denote some dynamic problem, and we just call $\D$ a problem for
brevity.

For each integer $n\geq1$, consider the function $u:\E_{n}\rightarrow\{0,1\}^{*}$
that maps each instance-update $(I,I')\in\E_{n}$ to the bit-string
which represents the positions of bits where $I'$ differ from $I$.
We call $u(I,I')$ the \emph{standard encoding} of $(I,I')$ and write
$I'=I+u(I,I')$. The \emph{length} of this encoding is denoted by
$|u(I,I')|$. More generally, for any instance-sequence $(I_{0},\ldots,I_{k})$
of $\D$, we write $I_{k}=I_{0}+u(I_{0},I_{1})+\cdots+u(I_{k-1},I_{k})$.
Note that it is quite possible for two different instance-update $(I_{0},I_{1})\in\E_{n}$
and $(I_{2},I_{3})\in\E_{n}$ to have the same standard encoding,
so that we get $u(I_{0},I_{1})=u(I_{2},I_{3})$.

Let $\lambda_{\D}:\mathbb{N\rightarrow\mathbb{N}}$ be an integer
valued function such that $\lambda_{\D}(n)=\max_{(I,I')\in\E_{n}}|u(I,I')|$
is equal to the maximum length over all standard encoding of instance-updates
in $\G_{n}$, for each positive integer $n\in\mathbb{N}$. We refer
to $\lambda_{\D}(\cdot)$ as the \emph{instance-update-size} of $\D$. 
\begin{fact}
\label{fact:logn update size} We have $\lambda_{\D}(n)\ge\log n$
if there is some instance-update between instances of $\D$ of size
$n$.\end{fact}
\begin{proof}
The standard encoding $u$ needs at least $\log n$ bits to specify
a single bit-position where two instances of $\D_{n}$ differ from
one another. 
\end{proof}

Recall that in the static setting,  two decision problems $\P$ and $\mathcal{Q}$ are {\em complements} of each other iff the set of YES instances of $\P$ are the set of NO instances of $\mathcal{Q}$, and vice versa.  We now define when two dynamic decision problems are complements of each other.

\begin{defn}
\label{def:complement}
We say that a dynamic decision problem $\D = (\P, \G)$ is the {\em complement} of another dynamic decision problem $\D' = (\mathcal{Q}, \G')$ (and vice versa) iff the following two conditions hold:
\begin{itemize}
\item The corresponding static problems $\P$ and $\mathcal{Q}$ are complements of each other. Specifically, we have $\P^{-1}(1) = \mathcal{Q}^{-1}(0)$ and $\P^{-1}(0) = \mathcal{Q}^{-1}(1)$.
\item For each $n \geq 1$, we have $\G_n = \G'_n$. In other words, the dynamic problems $\D$ and $\D'$ have the same update-graphs for each $n \geq 1$.
\end{itemize}
\end{defn}

\noindent Our formalization captures many dynamic problems -- even the ones
that allow for query operations (in addition to update operations). 
\begin{example}
[Dynamic problems with queries]\label{exa:query}Consider the dynamic
connectivity problem. In this problem, we are given an undirected
graph $G$ with $N$ nodes which is updated via a sequence of edge
insertions/deletions. At any time, given a query $(u,v)$, we have
to decide whether the nodes $u$ and $v$ are connected in the current
graph $G$. We can capture this problem using our formalization.

Set $n=N^{2}+2\log N$, and define the update-graph $\G_{n}=(\U_{n},\E_{n})$
as follows. Each instance $I\in\U_{n}=\{0,1\}^{n}$ represents a triple
$(G,u,v)$ where $G$ is an $N$-node graph and $u,v\in[N]$ are two
nodes in $G$. There is an instance-update $(I,I')\in\E_{n}$ iff
either \{$I=(G,u,v)$ and $I'=(G,u',v')$\} or \{$I=(G,u,v)$ and
$I'=(G',u,v)$ and $G,G'$ differs in exactly one edge\}. Intuitively,
the former case corresponds to a query operation, whereas the latter
case corresponds to the insertion/deletion of an edge in $G$. Since
an $N$-node graph can be represented as a string of $N^{2}$ bits
using an adjacency matrix, a triple $(G,u,v)$ can be represented
as a string of $N^{2}+2\log N=n$ bits. Let $\P_{n}:\{0,1\}^{n}\rightarrow\{0,1\}$
be such that $(G,u,v)\in\P_{n}^{-1}(1)$ is an yes instance if $u$
and $v$ are connected in $G$, and $(G,u,v)\in\P_{n}^{-1}(0)$ is
a no instance otherwise. 
Let $\P=\{\P_{n}\}_{n}$ and $\G=\{\G_{n}\}_{n}$. Then the ordered
pair $\D=(\P,\G)$ captures the dynamic connectivity problem. It is easy to see that $\D$ has an instance-update-size of $\lambda_{\D}(n)=\Theta(\log n)$.
\end{example}

\begin{example}
[Partially dynamic problems]The decremental connectivity problem
is the same as the dynamic connectivity problem, except that the update
sequence consists only of edge deletions. Our formalization captures
this problem in a similar manner as in \ref{exa:query}. The only
difference is this. For each $n\in\mathbb{N}$, there exists an instance-update
$(I,I')\in\E_{n}$ iff either \{$I=(G,u,v)$ and $I'=(G,u',v')$\}
or \{$I=(G,u,v)$ and $I'=(G',u,v)$ and $G'$ is obtained from $G$
by deleting an edge\}. %
\end{example}

In the two examples described above, we observed that $\lambda_{\D}(n) = \Theta(\log n)$. This happens to be the case with almost all the dynamic problems considered in the literature. For example, in a dynamic graph problem an update typically consists of insertion or deletion of an edge in an input graph, and it can be specified using $O(\log N)$ bits if the input graph contains $N$ nodes. Accordingly,  we will make the following assumption throughout the rest of the paper.

\begin{assumption}
\label{assume:logn update size} Every dynamic problem $\D$ has $\lambda_{\D}(n) = \Theta(\log n)$. 
\end{assumption}

\subsection{Dynamic Algorithms in the Bit-probe Model}
\label{sec:model:dynamic}

One of the key ideas in this paper is to work with a nonuniform model
of computation called the \emph{bit-probe} model, which has been studied
since the 1970's by Fredman \cite{Fredman78} (see also a survey by
Miltersen \cite{Miltersen99cellprobe}). This allows us to view a
dynamic algorithm as a clean combinatorial object, which turns out
to be very useful in deriving our main results (defining complexity
classes and showing the completeness and hardness of specific problems
with respect to these classes).

\subsubsection{Formal Description}

\label{sub:formal alg}

An \emph{algorithm-family} $\A=\{\A_{n}\}_{n\ge1}$ is a collection
of algorithms. For each $n$, an \emph{algorithm} $\A_{n}$ \emph{operates
on} an array of bits $\mem\in\{0,1\}^{*}$ called the \emph{memory}.
The memory contains two designated sub-arrays called the \emph{input
memory} $\mem^{\inp}$ and the \emph{output memory} $\mem^{\outp}$.
$\A_{n}$ works in \emph{steps} $t=0,1,\ldots$. At the beginning
of any step $t$, an \emph{input $\inp(t)\in\{0,1\}^{*}$at step $t$}
is written down in $\mem^{\inp}$, and then $\A_{n}$ is \emph{called}.
Once $\A_{n}$ is called, $\A_{n}$ reads and write $\mem$ in a certain
way described below. Then $\A_{n}$ \emph{returns} the call. The bit-string
stored in $\mem^{\out}$ is the \emph{output at step $t$}. Let $\inp(0\rightarrow t)=(\inp(0),\dots,\inp(t))$
denote the \emph{input transcript} up to step $t$. We denote the
output of $\A_{n}$ at step $t$ by $\A_{n}(\inp(0\rightarrow t))$
as it can depend on the whole sequence of inputs it received so far.

After $\A_{n}$ is called in each step, how $\A_{n}$ probes (i.e.
reads or writes) the memory $\mem$ is determined by 1) a \emph{preprocessing
function} $\init_{n}:\{0,1\}^{*}\rightarrow\{0,1\}^{*}$ and 2) a
\emph{decision tree} $T_{n}$ (to be defined soon). At step $t=0$
(also called the \emph{preprocessing step}), $\A_{n}$ initializes
the memory by setting $\mem\gets\init_{n}(\inp(0))$. We also call
$\inp(0)$ an \emph{initial input}. At each step $t\ge1$ (also called
an\emph{ update step}), $\A_{n}$ uses the decision tree $T_{n}$
to \emph{operate on} $\mem$. 

A decision tree is a rooted tree with three types of nodes: \emph{read
nodes}, \emph{write nodes}, and \emph{end nodes}. Each read node $u$
has two children and is labeled with an index $i_{u}$. Each write
node has one child and is labeled with a pair $(i_{u},b_{u})$ where
$i_{u}$ is an index and $b_{u}\in\{0,1\}$. End nodes are simply
leaves of the tree. For any index $i$, let $\mem[i]$ be the $i$-th
bit of $\mem$. To say that $T_{n}$ \emph{operates on} $\mem$, we
means the following: 
\begin{itemize}
\item Start from the root of $T_n$. If the current node
$u$ is a read node, then proceed to the left-child if $\mem[i_{u}]=0$,
otherwise proceed to the right-child. If $u$ is a write node, then
set $\mem[i_{u}]\gets b_{u}$ and proceed to $u$'s only child. Else,
$u$ is a leaf (an end node), then stop.
\end{itemize}

The root-to-leaf path followed by the decision tree  while operating on its memory is called the {\em execution path}. Clearly, the execution path depends on the contents of the memory bits. Also, note that the number of probes made by the algorithm during a call
at an update step is equal to the length of the execution path traversed during the call. Thus, the \emph{update
time of the algorithm} $\A_{n}$ is defined by the depth (the maximum
length of any root to leaf path) of $T_{n}$. Similarly, the {\em space complexity} of the algorithm $\A_n$ is defined to be the  number of bits available in its memory $\mem$.

We denote the \emph{update
time of the algorithm-family} $\A$ by a function $\time_{\A}(n)$
where $\time_{\A}(n)$ is the update time of $\A_{n}$. Similarly, the {\em space complexity of the algorithm-family} $\A$ is denoted by a function $\text{Space}_{\A}(n)$, where $\text{Space}_{\A}(n)$ is the space complexity of $\A_n$.

From now on, whenever we have to distinguish between multiple different
algorithms, we will add the subscript $\A_{n}$ to the notations introduced
above (e.g., $\init_{\A_{n}}$, $T_{\A_{n}}$, $\mem_{\A_{n}}$, $\inp_{\A_{n}}(0)$).

\subsubsection{High-level Terminology}

It is usually too cumbersome to specify an algorithm at the level
of its preprocessing function and decision tree. Hence, throughout this paper,
we usually only describe how $\A_{n}$ reads and writes the memory
at each step, which determines its preprocessing function $\init_{\A_{n}}$
and decision tree $T_{\A_{n}}$.

\paragraph{Solving problems.}
A  problem $\D$ is \emph{solved }by an algorithm-family
$\A$ if, for any $n$, we have:
\begin{enumerate}
\item In the preprocessing step, $\A_{n}$ is given an initial instance
$I_{0}$ of size $n$ (i.e. $\inp(0)=I_{0}$), and it outputs $\A_{n}(\inp(0\rightarrow 0))=\D(I_{0})$.
\item In each update step $t\in \{1, \ldots, \poly (n)\}$, $\A_{n}$ is given an instance-update
$(I_{t-1},I_{t})$ as input (i.e. $\inp(t)=(I_{t-1},I_{t})$), and it outputs $\A_{n}(\inp(0\rightarrow t))=\D(I_{t})$.  
\end{enumerate}
We also say that the algorithm $\A_{n}$ \emph{solves }an $n$-slice
$\D_{n}$ of the problem $\D$. For each step $t$, we say $I_{t}$
is an \emph{instance maintained by $\A_{n}$ at step $t$}. Note that $t \in \{1, \ldots, \poly (n) \}$ in item (2) above. This is emphasized in the assumption below.

\begin{assumption}
\label{assume:poly}
When we say that an algorithm $\A_n$ {\em solves} an $n$-slice $\D_n$ of a problem $\D_n$, we mean that the algorithm gives the correct output for polynomially  many update-steps $t = 1, \ldots, \poly (n)$. 
\end{assumption}

\paragraph{Subroutines.}

Let $\A$ and $\B$ be algorithm-families, and consider two algorithms
$\A_{n}\in\A$ and $\B_{m}\in\B$. We say that the algorithm $\A_{n}$
\emph{uses $\B_{m}$ as a subroutine} iff the following holds. The
memory $\mem_{\A_{n}}$ of the algorithm $\A_{n}$ contains a designated
sub-array $\mem_{\B_{m}}$ for the subroutine $\B_{m}$ to operate
on. As in \ref{sub:formal alg}, $\mem_{\B_{m}}$ has two designated
sub-arrays for input $\mem_{\B_{m}}^{\inp}$ and for $\mem_{\B_{m}}^{\outp}$.
$\A_{n}$ might read and write anywhere in $\mem_{\B_{m}}$. At each
step $t_{\A_{n}}$ of $\A_{n}$, $\A_{n}$ can \emph{call} $\B_{m}$
several times. 

The term ``call'' is consistent with how it is used in \ref{sub:formal alg}.
Let $t_{\B_{m}}$ denote \emph{the step of $\B_{m}$} which is set
to zero initially. When $\A_{n}$ calls $\B_{m}$ with an input $x\in\{0,1\}^{*}$
then the following holds. First, $\A_{n}$ writes an \emph{input $\inp_{\B_{m}}(t_{\B_{m}})=x$}
\emph{at step $t_{\B_{m}}$ of $\B_{m}$ }in $\mem_{\B_{m}}^{\inp}$.
Then $\B_{m}$ reads and write $\mem_{\B_{m}}$ in according to its
preprocessing function $\init_{\B_{m}}$ and decision tree $T_{\B_{m}}$.
Then $\B_{m}$ returns the call with the output $\B_{m}(\inp_{\B_{m}}(0\rightarrow t_{\B_{m}}))$
on $\mem_{\B_{m}}^{\out}$. Finally, the step $t_{\B_{m}}$ of $\B_{m}$
get incremented: $t_{\B_{m}}\gets t_{\B_{m}}+1$.

For each call, the update time of $\B_{m}$ contributes to
the update time of $\A_{n}$. In low-level, we can see that the preprocessing
function $\init_{\A_{n}}$ is defined by ``composing'' $\init_{\B_{m}}$
with some other functions, and the decision tree $T_{\A_{n}}$ is
a decision tree having $T_{\B_{m}}$ as sub-trees in several places.

\paragraph{Oracles.}

Suppose that $\O$ is an algorithm-family which solves some  problem
$\D$. We say that the algorithm $\A_{n}$ \emph{uses $\O_{m}$ as
a oracle} if $\A_{n}$ uses $\O_{m}$ as a subroutines just like above,
except that there are the two differences.
\begin{enumerate}
\item \textbf{(Black-box access):} $\A_{n}$ has very limited access to
$\mem_{\O_{m}}$. $\A_{n}$ can call $\O_{m}$ as before, but must
write only in $\mem_{\O_{m}}^{\inp}$ and can read only from $\mem_{\O_{m}}^{\out}$.
More specifically, suppose that $\A_{n}$ call $\O_{m}$ when the
step of $\O_{m}$ is $t_{\O_{m}}=0$. Then, $\A_{n}$ must write $\inp_{\O_{m}}(0)=I'_{0}$
in $\mem_{\O_{m}}^{\inp}$ where $I'_{0}$ is some instance of the
problem $\D$ and will be called an instance maintained by $\O_{m}$
from then. If the step of $\O_{m}$ is $t_{\O_{m}}\ge1$, then $\A_{n}$
must write $\inp_{\O_{m}}(t_{\O_{m}})=u(I'_{t_{\O_{m}}-1},I'_{t_{\O_{m}}})$
in $\mem_{\O_{m}}^{\inp}$ where $(I'_{t_{\O_{m}}-1},I'_{t_{\O_{m}}})$
is some instance-update of the problem $\D$. After each call, $\A_{n}$
can read the output $\O_{m}(\inp_{\O_{m}}(0\rightarrow t_{\O_{m}})=\D(I'_{t_{\O_{m}}}$)
which is the answer of the instance $I'_{t_{\O_{m}}}$.
\item \textbf{(Free call):} The update time of $\O_{m}$ does \emph{not}
contribute to the update time of $\A_{n}$. We model this high-level
description as follows. We already observed that the decision tree
$T_{\A_{n}}$ is a decision tree which has $T_{\O_{m}}$ as sub-trees
in several places. For each occurrence $T'$ of $T_{\O_{m}}$ in $T_{\A_{n}}$,
we assign the weight of edges between two nodes of $T'$ to be zero.
The update time of $\A_{n}$ is the \emph{weighted }depth of $T_{\A_{n}}$,
i.e. the maximum \emph{weighted} length of any root to leaf path.
\item \textbf{(Space complexity):} The memory $\mem_{\O_m}$ is {\em not} part of the memory $\mem_{\A_n}$. In other words, the space complexity of the oracle does not contribute to the space complexity of  $\A$.
\end{enumerate}

\paragraph{Oracle-families and Blow-up size.}

Let $\A$ be an an algorithm-family. Let $\qsize_{\A}:\mathbb{N}\rightarrow\mathbb{N}$
be a function. We say that $\A$ \emph{uses an oracle-family }$\O$
with \emph{blow-up size $\qsize_{\A}$ }if, for each $n$, $\A_{n}$
uses $\O_{m}$ as an oracle where $m\le\qsize_{\A}(n)$ (or even when
$\A_{n}$ uses many oracles $\O_{m_{1}},\dots,\O_{m_{k}}$ where $m_{i}\le\qsize(n)$
for all $i$). %

\section{Dynamic Complexity Classes $\dP$ and $\dNP$}

\label{sec:class}

We start with an informal description of the complexity class $\dP$, which is a natural analogue of the class P in the dynamic setting. First, we recall that in almost all the dynamic problems known in the literature, an instance-update  $(I_{t-1}, I_t)$  can be specified using $O(\log n)$ bits (see~\ref{assume:logn update size}).  Hence, intuitively, a dynamic {\em decision} problem $\D_n$ should be in the class $\dP$ if it admits an algorithm $\A_n$ whose update time is polynomial in $\log n$ (the number of bits needed to represent an instance-update). Thus, it is natural to say that a dynamic decision problem is in $\dP$ if it admits a dynamic algorithm with $O(\polylog (n))$ update time. In addition, for technical reasons that will become apparent later on, we need to allow the algorithm to have quasipolynomial space complexity. This is summarized in the definition below.

\begin{defn}
[Class $\dP$]\label{def:P} A dynamic decision problem $\D$ is in $\dP$ iff there is an algorithm-family $\A$ for
solving $\D$ with update time $\time_{\A}(n) = O(\polylog (n))$ and space-complexity $\text{Space}_{\A}(n) = O(\poly (n))$.
\end{defn}

Next, in order to define the  class $\dNP$, we first introduce the notion of a {\em verifier} in~\ref{def:verifier}. Note that this  is almost analogous to the definition of a verifier in the static setting, {\em except} the fact that at each step $t$ the verifier $\V_n$ outputs an ordered pair $(x_t, y_t)$ where $x_t \in \{0,1\}$ and $y_t \in \{0,1\}^{\polylog (n)}$ (instead of outputting a single bit). Intuitively, the bit $x_t \in \{0,1\}$ corresponds to the standard single bit output of a verifier in the static setting, whereas $y_t$ -- when thought of as a $\polylog (n)$ bit integer -- captures the {\em reward} obtained by the verifier.

\begin{defn} [Verifier-family]
\label{def:verifier} 
An algorithm-family $\V$ is a {\em verifier-family} for a dynamic decision problem $\D$ iff  the following holds for each $n \geq 1$.
\begin{itemize}
\item {\em Preprocessing:} At step $t = 0$, the algorithm $\V_n$ gets a {\em starting instance} $I_0$ of $\D_n$ as input, and it outputs an ordered pair $(x_0, y_0)$ where $x_0 \in \{0,1\}$ and $y_0 \in \{0,1\}^{\polylog(n)}$.
\item {\em Updates:} Subsequently, at each step $t \geq 1$, the algorithm $\V_n$ gets an {\em instance-update} $(I_t, I_{t-1})$ of $\D_n$ and a {\em proof} $\pi_t \in \{0,1\}^{\polylog (n)}$ as input, and it outputs an ordered pair $(x_t, y_t)$ where $x_t \in \{0,1\}$ and $y_t \in \{0,1\}^{\polylog (n)}$.  Note that the output $(x_t, y_t)$ depends on  the instance-sequence $(I_0, \ldots, I_t)$  and the {\em proof-sequence} $(\pi_1, \ldots, \pi_t)$ seen so far. 
\end{itemize}
\end{defn}

\noindent We  now  define the complexity class $\dNP$. Intuitively, a dynamic decision problem $\D$ is in $\dNP$ iff it admits a verifier-family $\V$ with polylogarithmic update time and polynomial space complexity that satisfy the following two properties for every $n \geq 1$. (1) The verifier $\V_n$  always outputs $x_t = 0$ on the NO instances of $\D_n$, regardless of the proof-sequence given to it as part of the input.  (2) The verifier $\V_n$  always outputs $x_t = 1$ on the YES instances of $\D_n$,  provided the proof-sequence given to it is {\em reward-maximizing}, in the sense that at each step $t \geq 1$ the proof $\pi_t$ is chosen in such a way that maximizes  the {\em reward} $y_t$ (when we think of $y_t$ as a $\polylog (n)$ bit integer).

\begin{defn}
[Class $\dNP$]\label{def:NP} A decision problem $\D$ is in $\dNP$ iff it admits a verifier-family $\V$ with update-time $\time_{\V}(n) = O(\polylog (n))$ and space-complexity $\text{Space}_{\V}(n) = O(\poly (n))$ which satisfy the following properties for each $n \geq 1$.  Fix any instance-sequence $(I_0, \ldots, I_k)$ of $\D_n$. Suppose that  $\V_n$ gets $I_0$ as input at step $t = 0$, and  $((I_{t-1},I_t), \pi_t)$ as input at every step $t \geq 1$. Furthermore, the verifier $\V_n$ outputs $(x_t, y_t)$ at each step $t$. Then:
\begin{enumerate}
\item  For every proof-sequence $(\pi_1, \ldots, \pi_k)$,  we have  $x_t = 0$ for each $t \in \{0, \ldots, k\}$ where $\D_n(I_t) = 0$.
\item If the proof-sequence $(\pi_1, \ldots, \pi_k)$ is {\em reward-maximizing} (defined below), then we have $x_t = 1$ for each $t \in \{0, \ldots, k\}$ with $\D_n(I_t) = 1$, 
\end{enumerate}
The proof-sequence $(\pi_1, \ldots, \pi_k)$ is {\em reward-maximizing} iff at each step $t \geq 1$, given the past history $(I_0, \ldots, I_t)$ and $(\pi_1, \ldots, \pi_{t-1})$, the proof $\pi_t$ is chosen in such a way that maximizes  $y_t$ (when we think of $y_t$ as a $\polylog (n)$ bit integer). We say that such a proof $\pi_t$ is {\em reward-maximizing}.
\end{defn}

We can now define the dynamic complexity class $\dcoNP$ in a natural manner. Intuitively, we get the class $\dcoNP$ if we switch the phrase ``$\D_n(t) = 0$'' with ``$\D_n(t) = 1$'', and the phrase ``$x_t = 0$'' with ``$x_t = 1$'' in~\ref{def:NP}. In other words, a decision problem $\D$ is in $\dcoNP$ iff its complement decision problem $\D'$ (see~\ref{def:complement}) is in $\dNP$.   A more formal definition is given below.

\begin{defn}
[Class $\dcoNP$]
\label{def:coNP}
 A decision problem $\D$ is in $\dcoNP$ iff it admits a verifier-family $\V'$ with update-time $\time_{\V'}(n) = O(\polylog (n))$ and space-complexity $\text{Space}_{\V'}(n) = O(\poly (n))$ which satisfy the following properties for each $n \geq 1$.  Fix any instance-sequence $(I_0, \ldots, I_k)$ of $\D_n$. Suppose that  $\V'_n$ gets $I_0$ as input at step $t = 0$, and  $((I_{t-1},I_t), \pi'_t)$ as input at every step $t \geq 1$.  Furthermore, the verifier $\V'_n$ outputs $(x'_t, y'_t)$ at each step $t$. Then:
\begin{enumerate}
\item  For every proof-sequence $(\pi'_1, \ldots, \pi'_k)$,  we have  $x'_t = 1$ for each $t \in \{0, \ldots, k\}$ where $\D_n(I_t) = 1$.
\item If the proof-sequence $(\pi'_1, \ldots, \pi'_k)$ is {\em reward-maximizing} (defined below), then we have $x'_t = 0$ for each $t \in \{0, \ldots, k\}$ with $\D_n(I_t) = 0$, 
\end{enumerate}
The proof-sequence $(\pi'_1, \ldots, \pi'_k)$ is {\em reward-maximizing} iff at each step $t \geq 1$, given the past history $(I_0, \ldots, I_t)$ and $(\pi'_1, \ldots, \pi'_{t-1})$, the proof $\pi'_t$ is chosen in such a way that maximizes  $y'_t$ (when we think of $y'_t$ as a $\polylog (n)$ bit integer). We say that such a proof $\pi'_t$ is  {\em reward-maximizing}.
\end{defn}

Just as in the static setting, we can easily prove that $\dP \subseteq \dNP \cap \dcoNP$ and we conjecture that $\dP \neq \dNP$. The big question left open in this paper is to resolve this conjecture.

\begin{cor}
\label{cor:pnp}
We have $\dP \subseteq \dNP \cap \dcoNP$.
\end{cor}

\begin{proof}(sketch)
Consider any dynamic decision problem $\D$ that belongs to the class $\dP$. Then it admits an algorithm-family $\A$ with polynomial space complexity and polylogarithmic update time. For each $n \geq 1$, we can easily modify $\A_n$ to get a verifier $\V_n$ for $\D_n$ which satisfies the conditions stated in~\ref{def:NP}. Specifically, upon receiving an input $((I_{t-1}, I_t), \pi_t)$ at step $t \geq 1$, the verifier $\V_n$  ignores the proof $\pi_t$ and simulates the execution of $\A_n$ on the instance-update $(I_{t-1}, I_t)$. The verifier then outputs the ordered pair $(x_t, y_t)$, where $x_t$ is equal to the output of $\A_n$ and $y_t$ is any arbitrary $\polylog (n)$ bit string. This implies that the problem $\D$ also belongs to the class $\dNP$, and hence we get $\dP \subseteq \dNP$.

Using a similar argument, we can also show that $\dP \subseteq \dcoNP$.
\end{proof}

\subsection{Examples of some dynamic problems in $\dP$ and $\dNP$}

In this section, we classify many dynamic problems that were previously
studied in various contexts into our complexity classes. We will only give a high-level description of each problem\footnote{For each $n$, it should be clear how to describe each problem as the tuple $(\P_{n},\G_{n})$ from \ref{sec:problem}. Also, if a problem
can handle both updates and queries, then we can formalize both as
instance-updates in $\G_{n}$ as shown in \ref{exa:query}. }. We also list a problem which is not decision problem but a correct
answer is a number with logarithmic bits (e.g. what is the number
$a\in[n]$?) as well. This is because there is a corresponding decision
problem (e.g. is $a>k$?). It is easy to see that if a dynamic algorithm
for one problem implies another algorithm for the corresponding problem
with essentially the same update time (up to a logarithmic factor),
and vice versa.  

For $\dP$, we give a list of some problems in \ref{table:P} which
is not at all comprehensive. The only goal is to show that there are
problems from various contexts that are solvable by fast deterministic
dynamic algorithms. For each problem $\D$ in \ref{table:P}, when an instance
can be represented using $n$ bits, it holds that the instance-update
size is $\lambda_{\D}(n)=\Theta(\log n)$. This corroborates~\ref{assume:logn update size}.

Problems in $\dNP$ that are not known to be in $\dP$
are listed in \ref{table:BPP-NP}.  The complexity class $\dNP$ is  huge, in the sense that it contains many problems which are not
known to be in $\dP$. It is easy to show that dynamic connectivity
on general graphs is $\dNP$ by giving a spanning forest as a proof
(see \ref{prop:conn in NP}). For any constant $\epsilon>0$, dynamic
$(1+\epsilon)$-approximate maximum matching is also in $\dNP$ by
giving a short augmenting path of length $O(1/\epsilon)$ as a proof
(see \ref{prop:matching in NP}). There is also a general way to show
that a problem is in $\dNP$: for any problem $\D$ whose the yes-instance
has a ``small certificate'', then $\D\in\dNP$. This includes many
problems such as dynamic subgraph detection, dynamic $uMv$, dynamic
3SUM,  dynamic planar nearest neighbor\footnote{The best known algorithm for planar nearest neighbor is by Chan \cite{Chan10}
which has polylogarithmic \emph{amortized} update time. The algorithm
is randomized, but it is later derandomized using the result by Chan
and Tsakalidis \cite{ChanT16}.}, Erickson's problem, and Langerman's problem (see \ref{prop:small cert in NP}).

\section{Reductions, Hardness and Completeness}

\label{sec:reduction}

We first define the concept of a $\dP$-reduction between two dynamic problems. This notion
is analogous to Turing-reductions for static problems. 
\begin{defn}
[$\dP$-reduction]\label{def:P-reduction}A dynamic problem $\D=(\P,\G)$ is\emph{ $\dP$-reducible}
to another dynamic problem $\D'$ iff there is an algorithm-family $\A$ that solves $\D$ using an oracle-family
$\O$ for the problem $\D'$, and  has update-time $\time_{\A}(n) = O(\polylog (n))$, space complexity $\text{Space}_{\A}(n) = O(\poly (n))$ and blow-up size $\qsize_{\A}(n) = O(\poly (n))$. We write $\D\le\D'$ and refer to the algorithm-family $\A$ as a \emph{ $\dP$-reduction } from $\D$ to $\D'$.
\end{defn}
The above definition is almost the same as showing that
$\D\in\dP$, except that $\A$ can use an oracle-family
for $\D'$. 

We now show some basic properties of $\dP$-reductions. In \ref{thm:easiness transfering} below, item (1) implies that if
$\D$ is reducible to an ``easy'' problem, then $\D$ is ``easy''
as well. Item (2), on the other hand, shows that the reduction is \emph{transitive}. 

The idea behind the proof of  item (1) in \ref{thm:easiness transfering} is straightforward
and standard: Given that $\D\le\D'$, there is an algorithm-family
$\R$ solving $\D$ efficiently using an oracle-family $\O$ for $\D'$.
Now, if $\D'$ can be solved efficiently by some algorithm $\A'$,
then every time $\R$ need to call $\O$, we instead call $\A'$ as
a subroutine, and hence obtain an algorithm $\A$ for solving $\D$
without calling an oracle and we are done. The proof of items (2) is just an extensions of the same argument.

\begin{prop}
\label{thm:easiness transfering}Suppose that $\D\le\D'$. Then we have:
\begin{enumerate}
\item If $\D'\in\dP$, then $\D\in\dP$. 
\item If $\D'\le\D''$, then $\D\le\D''$.
\end{enumerate}
\end{prop}
\begin{proof}
Let $\R$ be a $\dP$-reduction from $\D$ to $\D'$ as per~\ref{def:P-reduction}.

\medskip
\noindent
\textbf{(1):} Suppose that $\A'$ is a $\dP$-algorithm-family for
$\D'$. We will show a $\dP$-algorithm-family $\A$ for $\D$. For
each $n$, $\A_{n}$ just simulates $\R_{n}$ step by step, except that whenever
$\R_{n}$ calls an oracle $\O_{m}$ with  $m = \qsize_{\R}(n)$, the algorithm
$\A_{n}$ calls $\A'_{m}$ instead. Thus, the update time of
$\A_{n}$ is given by:
\begin{align*}
\time_{\A}(n) \leq \time_{\R}(n)\times\time_{\A'}(m) & =\polylog (n) \times\polylog (m)\\
 & =\polylog (n) \times\polylog (\qsize_{\R}(n))\\
 & = \polylog (n) \times \polylog (\poly (n)) = \polylog (n).
\end{align*}
The space complexity of $\A_n$ is given by:
\begin{align*}
\text{Space}_{\A}(n) & = \text{Space}_{\R}(n) + \text{Space}_{\A'}(m) \\
& = \poly (n) + \poly (m) \\
& = \poly (n) + \poly (\qsize_{\R}(n)) = \poly (n) + \poly (\poly (n)) = \poly (n).
\end{align*}
As both $\A'_{m}$ and $\O_{m}$ solve $\D'_{m}$,
it must be the case that $\A_{n}$ solves $\D_{n}$. Hence, we have:
$\D\in\dP$.

\medskip
\noindent
\textbf{(2): } The proof is similar in spirit to the argument used to prove (1). Since $\D' \le \D''$, let $\A'$ be an algorithm-family that solves $\D'$ using an oracle-family
$\O'$ for $\D''$. We have $\time_{\A'}(m)= O(\polylog (m))$ and $\qsize_{\A'}(m)=O(\poly (m))$.

Similar to the proof for item (1), we construct an algorithm-family $\A$ for the problem $\D$, which uses the oracle-family $\O'$ for $\D''$ with following parameters.
For each $n$, let $m = \qsize_{\R}(n)$ and $\ell = \qsize_{\A'}(m)$.
The algorithm $\A_{n}$ uses $\O'_{\ell}$ as an oracle to solve $\D_{n}$
with update time $O(\polylog (n))$. Let $\qsize_{\A}$ be
the blow-up size of $\A$ to $\O'$. We infer that:
\begin{align*}
\qsize_{\A}(n) & = \ell\\
 & = \qsize_{\A'}(m) \\
 & = \poly (m) \\
 & = \poly (\qsize_{\R}(n))\\
 & = \poly (\poly (n)) \\
 & =  \poly (n).
\end{align*}
This concludes that $\A$ is a $\dP$-reduction from $\D$ to $\D''$
and hence $\D\le\D''$.
\end{proof}

Next, we define the  notions of  $\dNP$-hardness and completeness.

\begin{defn}
\label{def:completeness}
[Hardness and Completeness]Let  $\D$
be a dynamic problem. 
\begin{enumerate}
\item We say that $\D$ is $\dNP$-hard  iff the following condition holds: If $\D \in \dP$ then $\D' \in \dP$ for every problem $\D' \in \dNP$.
\item We say that $\D$ is $\dNP$-complete  iff $\D\in \dNP$ and
$\D$ is $\dNP$-hard.
\end{enumerate}
\end{defn}

With our notions of hardness, the following is true. 
\begin{cor}
\label{cor:NP not P}Assuming that $\dNP\neq\dP$, if a problem
$\D$ is $\dNP$-hard, then $\D\notin\dP$. 
\end{cor}

\begin{cor}
\label{cor:reductions}
Consider any two dynamic problems $\D$ and $\D'$ such that (a) $\D \leq \D'$ and (b) $\D$ is $\dNP$-hard. Then $\D'$ is also $\dNP$-hard.
\end{cor}

\begin{proof}
Suppose that $\D' \in \dP$. Then~\ref{thm:easiness transfering} implies that $\D \in \dP$. Since $\D$ is $\dNP$-hard, this leads us to the conclusion that if $\D' \in \dP$ then $\dP = \dNP$. So $\D'$ is $\dNP$-hard.
\end{proof}

\section{The Dynamic Polynomial Hierarchy}
\label{sec:PH}

\newcommand{\sig}{\Sigma^{dy}}
\newcommand{\pid}{\Pi^{dy}}

In this section, we define a hierarchy of dynamic complexity classes that is analogous to the polynomial hierarchy in the static setting. We begin by introducing a useful notation. For any dynamic complexity classe $\C_1$ and $\C_2$, we define the class $\left(\C_1\right)^{\C_2}$ as follows. Intuitively, a dynamic problem $\D_1$ belongs to the class 
$\left(\C_1\right)^{\C_2}$ iff there is an algorithm-family $\A_1$ for $\D_1$ that is allowed have access to an oracle-family $\O_1$ for a problem $\D_2 \in \C_2$. We illustrate this by considering the following example.

\begin{example}
\label{ex:np:oracle}
\label{ex:p:oracle}
Consider any dynamic complexity class $\C$ and a dynamic problem $\D$.
\begin{itemize}
\item We say that the problem $\D$ belongs to the class $\left( \dP \right)^{\C}$ iff there is an algorithm-family $\A$ for $\D$ that uses an oracle-family $\O$ for some problem $\D' \in \C$ with $\qsize_{\A}(n) = O(\poly (n))$,  and satisfies the conditions stated in~\ref{def:P}.
\item We say that the problem $\D$ belongs to the class $\left( \dNP \right)^{\C}$  iff there is a verifier family $\V$ for $\D$ that uses an oracle-family $\O$ for some problem $\D' \in \C$ with $\qsize_{\V}(n) = O(\poly (n))$,  and satisfies the conditions stated in~\ref{def:NP}. 
\item Similarly, we say that the problem $\D$ belongs to the class $\left( \dcoNP \right)^{\C}$ iff there is a verifier family $\V'$ for $\D$ that uses an oracle-family $\O$ for some problem $\D' \in \C$ with $\qsize_{\V'}(n) = O(\poly (n))$,  and satisfies the conditions stated in~\ref{def:coNP}. 
\end{itemize}
\end{example}

We are now ready to introduce the dynamic polynomial hierarchy.
\begin{defn}[Dynamic polynomial hierarchy]
\label{def:ph}
We first inductively define the complexity classes $\sig_i$ and $\pid_i$ in the following manner.
\begin{itemize}
\item For $i = 1$, we have $\sig_1 = \dNP$ and $\pid_1 = \dcoNP$.
\item  For $i >  1$, we have $\sig_i = \left( \dNP \right)^{\sig_{i-1}}$ and $\pid_i = \left( \dcoNP \right)^{\sig_{i-1}}$.
\end{itemize}
Finally, we define $\dPH = \bigcup_{i \geq 1} \left( \sig_i \cup \pid_i \right)$.
\end{defn}

We refer the reader to~\ref{sec:PH:examples} for a list of dynamic problems that belong to $\dPH$. As in the static setting,  the successive levels of $\dPH$ are contained within each other.
\begin{lem}
\label{lem:ph:contain}
For each $i \geq 1$, we have $\sig_i \subseteq \sig_{i+1}$ and $\pid_i \subseteq \pid_{i+1}$.
\end{lem}

\begin{proof}
We use induction on $i$. The base case is trivial. For $i = 1$, we have $\sig_1 = \dNP \subseteq \left(\dNP\right)^{\dNP} = \sig_{2}$ and similarly $\pid_1 \subseteq \pid_2$. Suppose that the lemma holds for some $i$. We now observe that $\sig_{i+1} = \left( \dNP \right)^{\sig_i} \subseteq \left( \dNP \right)^{\sig_{i+1}} = \sig_{i+2}$. In this derivation, the second step holds because of our induction hypothesis that $\sig_i \subseteq \sig_{i+1}$. Similarly, we can show that our induction hypothesis implies that $\pid_{i+1} \subseteq \pid_{i+2}$. This completes the proof.
\end{proof}

 Similar to the static setting, we can show that if $\dP = \dNP$, then $\dPH$ collapses to $\dP$.

\begin{thm}
\label{th:ph:collapse:1}
If $\dP = \dNP$, then $\dPH = \dP$.
\end{thm}

\begin{proof}
Throughout the proof, we assume that $\dP =  \dNP = \sig_1$.  We use induction on $i$. For the base case, we already have $\sig_1 = \dP$. Now, suppose that $\sig_i = \dP$ for some $i \geq 1$. Then we get: 
$$\sig_{i+1} = \left(\dNP\right)^{\sig_{i}} = \left(\dNP\right)^{\dP} = \left( \dP \right)^{\dP} = \dP.$$
Thus, we derive that $\sig_i = \dP$ for all $i \geq 1$. Since each problem in $\pid_{i}$ is a complement of some problem in $\sig_i$, we also infer that $\pid_i = \dP$ for all $i \geq 1$. This concludes the proof.
\end{proof}

\subsection{Further result regarding  the collapse of $\dPH$}
\label{sec:PH:promise}

In this section, we prove that $\dPH$ collapses to the second level if $\dNP \subseteq \dcoNP$. Towards this end, we first state the following important lemma whose proof appears in~\ref{proof:lem:NP:coNP}.

\begin{lem}
\label{lem:NP:coNP}
We have $(\dNP)^{\dNP \cap \dcoNP} = \dNP$.
\end{lem}

\begin{thm}
\label{th:PH:collapse:next}
If $\dNP \subseteq \dcoNP$, then $\dPH = \dcoNP \cap \dNP$.
\end{thm}

\begin{proof}
Throughout the proof, we assume that $\dNP \subseteq \dcoNP$. This implies that  $\dNP = \dNP \cap \dcoNP$. We now claim: 
\begin{itemize}
\item $\sig_i = \dNP$ for all $i \geq 1$. 
\end{itemize}
To prove this claim, we use induction on $i$. The base case is clearly true, since we have $\sig_1 = \dNP$ by definition. By induction hypothesis, suppose that $\sig_i = \dNP$ for some $i \geq 1$. But this implies that $\sig_{i+1} = (\dNP)^{\sig_i} =  (\dNP)^{\dNP} = (\dNP)^{\dNP \cap \dcoNP} = \dNP$. We thus conclude that:
\begin{equation} 
\label{eq:PH:collapse:100}
\sig_i = \dNP \text{ for all } i \geq 1.
\end{equation} 
Recall that every  problem in $\pi_i$ is a complement of some  problem in $\sig_i$. Hence,~\ref{eq:PH:collapse:100} implies that:
\begin{equation}
\label{eq:PH:collapse:101}
\pid_i = \dcoNP \text{ for all } i \geq 1.
\end{equation}
The theorem follows from~\ref{eq:PH:collapse:100},~\ref{eq:PH:collapse:101} and the observation that $\dPH$ is closed under complements.
\end{proof}

We conclude this section with one more lemma that will be useful later on. Its proof is analogous to the proof of~\ref{lem:NP:coNP} and is therefore omitted.

\begin{lem}
\label{lem:NP:coNP:next}
We have $(\dNP \cap \dcoNP)^{\dNP \cap \dcoNP} = \dNP \cap \dcoNP$.
\end{lem}

\subsubsection{Proof of~\ref{lem:NP:coNP}}
\label{proof:lem:NP:coNP}

Since it is clearly the case that $\dNP \subseteq (\dNP)^{\dNP \cap \dcoNP}$, to complete the proof we only need to show that $(\dNP)^{\dNP \cap \dcoNP} = \dNP$. Consider any decision problem $\D^* \in (\dNP)^{\dNP \cap \dcoNP}$. We will show that $\D^* \in \dNP$. We begin by setting up some notations that will be used throughout the proof. By definition, the problem $\D^*$ admits a verifier-family $\V^*$ with the following properties. 
\begin{enumerate}
\item The verifier-family $\V^*$ uses an oracle-family $\O$ for a decision problem $\D \in \dNP \cap \dcoNP$ with $\qsize_{\V^*}(m) = O(\poly (m))$. Let $n(m) = \qsize_{\V^*}(m)$. Thus,  for each $m \geq 1$, the verifier $\V^*_m$ uses the oracle $\O_{n(m)}$. To ease notation, we simply write $n$ instead of $n(m)$. 
\item The verifier-family $\V^*$ has  $\time_{\V^*}(m) = O(\polylog (m))$ and  $\text{Space}_{\V^*}(m) = O(\poly (m))$. 
\item Fix any $m \geq 1$, and consider any instance-sequence $(I^*_0, \ldots, I^*_k)$ of $\D^*_m$. Suppose that $\V^*_m$ gets $I^*_0$ as input at step $t = 0$, and the ordered pair $((I^*_{t-1}, I^*_t), \pi^*_t)$ as input at each step $t \geq 1$. Furthermore, let $(x^*_t, y^*_t)$ denote the output of the verifier $\V^*_n$ at each step $t \geq 0$. Then:
\begin{itemize}
\item For every proof-sequence $(\pi^*_1, \ldots, \pi^*_k)$, we have $x^*_t = 0$ for each $t \in \{0,\ldots, k\}$ with $\D^*_m(t) = 0$.
\item If the proof-sequence is {\em reward-maximizing}, then we have $x^*_t = 1$ for each $t \in \{0,\ldots, k\}$ with $\D^*_m(t) = 1$. 
\end{itemize}
\item Since $\D \in \dNP \cap \dcoNP$, the decision problem $\D$ admits two verifier-families $\V$ and $\V'$ that respectively satisfy~\ref{def:NP} and~\ref{def:coNP}. While referring to the verifier-families $\V$ and $\V'$, we use the same notations that were introduced in~\ref{def:NP} and~\ref{def:coNP}.
\end{enumerate}

\medskip
\noindent We will now construct a verifier-family $\hat{\V}$ for the problem $\D^*$ that does not make any call to an oracle. Fix any $m \geq 1$. At a high level, instead of using the oracle $\O_n$, the verifier $\hat{\V}_m$ uses $\V_n$ and $\V'_n$ as subroutines in order to simulate the behavior of the verifier $\V^*_m$. Each time $\V^*_m$ makes a call to the oracle $O_n$, the verifier $\hat{\V}_m$ makes two calls to the verifiers $\V_n$ and $\V'_n$ for the problem $\D_n$. The verifier $\hat{\V}_m$ also checks that the answers returned by $\V_n$ and $\V'_n$ are {\em consistent} with each other.

\paragraph{Constructing the verifier $\hat{\V}_m$:} To be more specific, after receiving an instance $I^*_0$ of $\D^*_m$ as input in the preprocessing step, the verifier $\hat{\V}_m$ simulates the behavior of $\V^*_m$ on the same input $I^*_0$ and returns the same answer as $\V^*_m$. Now, suppose that the verifier $\hat{\V}_m$ has received an instance-sequence $(I^*_0, \ldots, I^*_{t^*-1})$ of $\D^*_m$ and a proof-sequence $(\hat{\pi}_1, \ldots, \hat{\pi}_{t^*-1})$ as input till this point. Furthermore, suppose that the verifier $\hat{\V}_m$ has successfully been able to simulate the behavior of $\V^*_m$ on the same instance-sequence $(I^*_0, \ldots, I^*_{t^*-1})$ and a (different) proof-sequence $(\pi^*_1, \ldots, \pi^*_{t^*-1})$ till this point. Now, at step $t^*$, the verifier $\V^*_m$ gets an ordered pair $((I^*_{t^*-1}, I^*_{t^*}), \pi^*_{t^*})$ as input and the verifier $\hat{\V}_m$ gets an ordered pair $((I^*_{t^*-1}, I^*_{t^*}), \hat{\pi}_{t^*})$ as input. Note that the instance-update part of the input at step $t^*$ remains the same for $\V^*_m$ and $\hat{\V}_m$. In contrast, the proofs $\pi^*_{t^*}$ and $\hat{\pi}_{t^*}$  differ from each other. The proof $\hat{\pi}_{t^*}$ is supposed to consist of $\pi^*_{t^*}$ followed by a sequence of proofs $\{\pi_t, \pi'_t\}$ (for $\V_n$ and $\V'_n$ respectively) corresponding to all the calls to the oracle $\O_n$ made by $\V^*_m$.

We now proceed with the description of the verifier $\hat{\V}_m$. After receiving the input $((I^*_{t^*-1}, I^*_{t^*}), \hat{\pi}_{t^*})$ at step $t^*$, the verifier $\hat{\V}_m$ starts simulating the verifier $\V^*_m$ (also at step $t^*$) on input $((I^*_{t^*-1}, I^*_{t^*}), \pi^*_t)$.
\begin{itemize}
\item During this simulation, whenever $\V^*_m$ calls the oracle $\O_n$ for $\D_n$ at a step (say) $t$ with input $(I_{t-1}, I_t)$, the verifier $\hat{V}_m$ calls $\V_n$ and $\V'_n$ as subroutines respectively with inputs $((I_{t-1}, I_t), \pi_t)$ and $((I_{t-1}, I_t), \pi'_t)$. We emphasize that $\pi_t$ and $\pi'_t$ are specified within the proof $\hat{\pi}_{t^*}$ received by $\hat{\V}_m$. Let $(x_t, y_t)$ and $(x'_t, y'_t)$ respectively denote the answers returned by $\V_n$ and $\V'_n$ at the end of this call. If $x_t \neq x'_t$, then we say that the verifier $\hat{\V}_m$ enters into "invalid'' mode. Specifically, this means that at every future update-step (including step $t^*$), the verifier $\hat{\V}_m$ outputs the answer $(0, 0)$. Otherwise, if $x_t = x'_t$  then we claim that:
\begin{itemize}
\item $x_t$ (or, equivalently, $x'_t$) is equal to the answer returned by the call to  $\O_n$ made by $\V^*_m$.
\end{itemize} 
The claim holds because either the call to the oracle $\O_n$ returns a $0$ (in which case item (1) in~\ref{def:NP} implies that $x_t = 0$), or the call to the oracle $\O_n$ returns a $1$ (in which case item (1) in~\ref{def:coNP} implies that $x'_t = 1$). This claim ensures that if $x_t = x'_t$, then the verifier $\hat{\V}_m$ can continue with its simulation of the behavior of $\V^*_m$. And this is precisely what the verifier $\hat{\V}_m$ does in this case.
\end{itemize}
Suppose that the verifier $\hat{\V}_m$ manages to complete the simulation of $\V^*_m$ at step $t^*$ without ever entering into {\sc Invalid} mode. Let $(x^*_{t^*}, y^*_{t^*})$ denote the answer returned by $\V^*_m$ at step $t^*$. Then the verifier $\hat{\V}_m$ returns the answer $(\hat{x}_{t^*}, \hat{y}_{t^*})$ at step $t^*$, where $\hat{x}_{t^*} = x^*_{t^*}$ and $\hat{y}_{t^*} = 1 y^*_{t^*}$. Here, the ``$1$" in front of $y^*_{t^*}$ represents the fact that the verifier $\hat{\V}_m$ did not ever enter the {\sc Invalid} mode. At this point, the verifier $\hat{\V}_m$ is ready to handle the next update at step $t^*+1$.

\paragraph{Analysis of correctness:} Suppose that the verifier $\hat{\V}_m$ received the instance-sequence $(I^*_0, \ldots, I^*_{t^*})$ and the proof-sequence $(\hat{\pi}_1, \ldots, \hat{\pi}_{t^*})$ as input till this point. Let $(\pi^*_1, \ldots, \pi^*_{t^*})$ denote the corresponding proof-sequence received by the verifier $\V^*_m$ till this point. For each $k \in \{0, \ldots, t^*\}$, let $(\hat{x}_k, \hat{y}_k)$ and $(x^*_k, y^*_k)$ respectively denote the outputs of the verifiers $\hat{\V}_m$ and $\V^*_m$ at step $k$.

\medskip 
In order to prove~\ref{lem:NP:coNP}, we need to show that the verifier $\hat{\V}_m$ for the problem $\D^*_m$ satisfies the properties outlined in~\ref{def:verifier} and~\ref{def:NP}. These properties are shown in~\ref{cl:reward:1},~\ref{cl:reward:2} and~\ref{cl:reward:3}. 

\begin{claim}
\label{cl:reward:1}
If the proof-sequence $(\hat{\pi}_1, \ldots, \hat{\pi}_{t^*})$ is reward-maximizing for the verifier $\hat{\V}_m$ w.r.t. the input-sequence $(I^*_0, \ldots, I^*_{t^*})$, then we have $\hat{x}_k = \D^*_m(k)$ at each step $k \in \{0, \ldots, k\}$. Thus, the verifier $\hat{\V}_m$ always produces the correct output when it works with a reward-maximizing proof-sequence.
\end{claim}

\begin{proof}
Throughout the proof, we assume that the proof-sequence $(\hat{\pi}_1, \ldots, \hat{\pi}_{t^*})$ is reward-maximizing for $\hat{\V}_m$ w.r.t.  the input-sequence $(I^*_0, \ldots, I^*_{t^*})$. This implies that the verifier $\hat{\V}_m$ never enters the {\sc Invalid} state during steps $k \in \{1, \ldots, t^*\}$, for the following reason.
\begin{itemize}
\item Consider any step $k \in \{1, \ldots, t^*\}$. If the verifier $\hat{\V}_m$ enters into the {\sc Invalid} state during this step, then it outputs  $(\hat{x}_k, \hat{y}_k)$ where $\hat{x}_k = \hat{y}_k = 0$. Otherwise, the verifier $\hat{\V}_m$ outputs $(\hat{x}_k, \hat{y}_k)$ where $\hat{y}_k$ starts with a ``$1$''. Hence, the value of $\hat{y}_k$ is maximized when the verifier $\hat{\V}_m$ does {\em not} enter into the {\sc Invalid} state. Since the entire sequence $(\pi^*_1, \ldots, \pi^*_{t^*})$ is reward-maximizing, it follows that the verifier $\hat{\V}_m$ does not enter into the {\sc Invalid} state during steps $1, \ldots, t^*$.
\end{itemize}
Next, we note that the corresponding proof-sequence $(\pi^*_1, \ldots, \pi^*_{t^*})$ is also reward-maximizing for the verifier $\V^*_m$ w.r.t. the same instance-sequence $(I^*_0, \ldots, I^*_{t^*})$, for the following reason.
\begin{itemize}
\item We have already shown that the verifier $\hat{\V}_m$ never enters the {\sc Invalid} state during steps $1, \ldots, t^*$. This implies that $\hat{y}_k = 1 y^*_k$ for each $k \in \{1, \ldots, t^*\}$, where $(\hat{x}_k, \hat{y}_k)$ and $(x^*_k, y^*_k)$ are respectively the outputs of the verifiers $\hat{\V}_m$ and $\V^*_m$ at step $k$. Thus,  maximizing the value of $\hat{y}_k$ is equivalent to maximizing the value of $y^*_k$. Since the proof-sequence $(\hat{\pi}_1, \ldots, \hat{\pi}_{t^*})$ is reward-maximizing for the verifier $\hat{\V}_m$ as per our assumption, it necessarily follows that the corresponding proof-sequence $(\pi^*_1, \ldots, \pi^*_{t^*})$ is also reward-maximizing for the verifier $\V^*_m$. 
\end{itemize}
Since the proof-sequence $(\pi^*_1, \ldots, \pi^*_{t^*})$ is reward-maximizing for the verifier $\V^*_m$ w.r.t. the instance-sequence is $(I^*_0, \ldots, I^*_{t^*})$, we infer that $x^*_k = \D^*_m(I^*_k)$ for each $k \in \{0, \ldots, t^*\}$. Furthermore, since the verifier $\hat{\V}_m$ never enters the {\sc Invalid} state during steps $1, \ldots, t^*$, we have $\hat{x}_k = x^*_k = \D^*_m(I^*_k)$ for each $k \in \{0, \ldots, t^*\}$. In other words, the verifier $\hat{\V}_m$ always outputs the correct answer when it receives a reward-maximizing proof-sequence.
\end{proof}

\begin{claim}
\label{cl:reward:2}
Fix any instance-sequence $(I^*_0, \ldots, I^*_{t^*})$. For every proof-sequence $(\hat{\pi}_1, \ldots, \hat{\pi}_{t^*})$, the verifier $\hat{\V}_m$ outputs $\hat{x}_k = 0$ at each step $k \in \{0, \ldots, t^*\}$ where $\D^*_m(I^*_k) = 0$.
\end{claim}

\begin{proof}
If the verifier $\hat{\V}_m$ ever enters the {\sc Invalid} state at some step $k^*$, then it keeps returning the answer $\hat{x}_k = 0$ at every step $k \geq k^*$. So the claim trivially holds in this case. Thus, throughout the rest of the proof, w.l.o.g. we assume that the verifier $\hat{\V}_m$ never enters the {\sc Invalid} state. However, if  this is the case, then we have $\hat{x}_k = x^*_k$ at each step $k \in \{0, \ldots, t^*\}$, where $(x^*_k, y^*_k)$ is the output of the verifier $\V^*_m$ at step $k$ when it receives the same instance-sequence $(I^*_0, \ldots, I^*_{t^*})$ and the proof-sequence $(\pi^*_1, \ldots, \pi^*_k)$ corresponding to $(\hat{\pi}_1, \ldots, \hat{\pi}_{t^*})$. Now, from the definition of $\V^*_m$ it follows that $\hat{x}_k = x^*_k = 0$ on all the instances $I^*_k$ where $\D^*_m(I^*_k) = 0$.
\end{proof}

\begin{claim}
\label{cl:reward:3}
The verifier $\hat{\V}_m$ has space-complexity $\text{Space}_{\hat{\V}}(m) = O(2^{\polylog (m)})$ and update-time $\time_{\hat{\V}}(m) = O(\polylog (m))$.
\end{claim}

\begin{proof}
The space-complexity of the verifier $\hat{\V}_m$ is dominated by the space-complexities of the subroutines $\V_n, \V'_n$, and that of the verifier $\V^*_m$. Thus, we have:
\begin{eqnarray*}
\text{Space}_{\hat{\V}}(m) & = & \text{Space}_{\V^*}(m) + \text{Space}_{\V}(n) + \text{Space}_{\V'}(n)  \\
& = & O\left(\poly (m) + \poly (n) + \poly (n) \right) \\
& = & O\left( \poly (m) \right)
\end{eqnarray*}
The last equality holds since the verifier $\V^*_m$ uses the oracle $\O_n$ for the problem $\D_n$, and $n = \qsize_{\V^*}(m) = O\left(\poly (m)\right)$.

Moving on,  the verifier $\V^*_m$ has update time $O(\polylog (m))$ and it uses the oracle $\O_n$ for the problem $\D_n$. In the verifier $\hat{\V}_m$, each call to the oracle $\O_n$ is replaced by two calls to the verifiers $\V_n$ and $\V'_n$. Since each of the verifiers $\V_n$ and $\V'_n$ has update time $O(\polylog (n))$, we get:
\begin{eqnarray*}
\time_{\hat{\V}}(m) & = & O(\polylog (m)) \cdot O(\polylog (n)) = O(\polylog (m)).
\end{eqnarray*}
Again, the last equality holds since $n = O\left(\poly (m) \right)$.
\end{proof}

\ref{lem:NP:coNP} follows from~\ref{cl:reward:1},~\ref{cl:reward:2} and~\ref{cl:reward:3}.

\newpage{}

\part{$\dNP$-completeness (in the Bit-probe Model)}\label{part:completeness}

This part is organized as follows. In~\ref{sec:DT}, we define a dynamic problem called ``First Shallow Decision Tree" ($\fDT$ for short) and show that this problem is $\dNP$-hard. In~\ref{sec:DNF}, we define another problem called \emph{dynamic narrow DNF evaluation problem} (or $\dDNF$ for short). We show that $\fDT$ is $\dP$-reducible   to $\dDNF$. This means that $\dDNF$ is $\dNP$-complete. We conclude this part  by explaining (in~\ref{sub:list NP hard}) why the $\dNP$-completeness of $\dDNF$ almost immediately implies that many natural dynamic problems are $\dNP$-hard.

\section{First Shallow Decision Tree: an intermediate problem}

\label{sec:DT}

This section is organized as follows. In \ref{sub:DT def}, we define
a dynamic problem called First Shallow Decision Tree ($\fDT$). In
\ref{sub:DT hard}, we show that this problem is $\dNP$-hard.

\subsection{Definition of $\protect\fDT$}
\label{sub:DT def}

We start with a definition of the First Shallow Decision Tree problem
in the static setting. We denote this static problem by $\P'$, and emphasize that this is {\em not} a decision problem. An instance $I \in \P'$ is
an ordered pair $(\mem_{I},\T_{I})$ such that:
\begin{itemize}
\item  $\mem_{I}$ is an array of bits and $\T_{I}$ is a collection of decision trees. 
\item Each leaf node $v$ in each decision tree $T \in \T_I$ is labelled with a $\polylog(|\T_I|)$ bit integer $r(v) \in \{0,1\}^{\polylog (|\T_I|)}$.  We refer to $r(v)$ as the {\em rank} of $v$. This rank $r(v)$ is independent of the contents of the array $\mem_I$.
\item Each decision tree $T \in\T_{I}$
\emph{operates on} the same memory $\mem_{I}$\footnote{Recall the definition of ``operating on'' from \ref{sub:formal alg}.}.
\item All the decision trees in $\T_{I}$ are \emph{shallow},
in the sense that the \emph{depth} (the maximum length of a root to
leaf path) of each tree $T\in\T_{I}$ is at most $O(\text{poly}\log|\T_{I}|)$.
\end{itemize}
For every decision tree $T \in \T_I$, we define $v^*_{T}$ to be the leaf node of the {\em execution-path} of $T$ when it operates on $\mem_I$.\footnote{Recall the definition of ``execution path" from~\ref{sub:formal alg}.} Note that $v^*_T$ depends on the contents of the array $\mem_I$, since the latter determines the execution-path of $T$ when it operates on $\mem_I$. The answer $\P'(I)$ to an instance $I\in\P'$ points to the decision tree $T \in \T_I$ which maximizes the rank $r(v^*_T)$. Thus, the answer $\P'(I)$ can be encoded using $O(\log |\T_I|)$ bits. We emphasize that the answer $\P'(I)$ depends on the contents of the array $\mem_I$ (which determines the leaf node $v^*_T$ for every tree $T \in \T_I$).

Intuitively, in the dynamic version of the problem denoted by $\fDT$,
the instance $I$ keeps changing via a sequence of updates where each
update flips one bit in the memory $\mem_{I}$, and we have
to keep track of the answer $\P'(I)$ at the current instance $I$.
Below, we give a formal description.

For each integer $n\geq1$, let $\P'_{n}$ denote the $n$-slice of
the $\fDT$ problem, which consists of all instances $I\in\P'$ that
are encoded using $n$ bits. In the dynamic setting, we impose the following
 graph structure $\G'_{n}=(\U'_{n},\E'_{n})$ on $\P'_{n}$ with
node-set $\U'_{n}=\{0,1\}^{n}$: For any
two instances $I,I'\in\P'_{n}$, there is an instance-update  $(I,I')\in\E'_{n}$
iff $\T_{I}=\T_{I'}$, and $\mem_{I}$ and
$\mem_{I'}$ differ in exactly one bit. 

We denote the problem $\fDT$ by $\D'=(\P',\G')$, and the $n$-slice
of this problem by $\D'_{n}=(\P'_{n},\G'_{n})$. We now derive a simple 
upper bound on the instance-update-size of $\D'$, which is consistent with~\ref{assume:logn update size}.
\begin{cor}
\label{cor:lambda:dt} for every integer $n\geq1$, we have $\lambda_{\D'}(n)\leq\log n$. \end{cor}
\begin{proof}
There is an instance-update $(I,I')$ of $\D'$ iff $\mem_{I}$ and
$\mem_{I'}$ differ in exactly one bit. Hence, the standard encoding
of $(I,I')$ can be specified using at most $\log n$ bits, and we
get $\lambda_{\D'}(n)\leq\log n$. 
\end{proof}

\subsection{$\protect\dNP$-hardness of $\protect\fDT$ }

\label{sub:DT hard}

We will show that one can efficiently solve any problem
in $\dNP$ using an oracle for $\fDT$. Specifically, we will
prove the following theorem. 
\begin{thm}
\label{th:np-hardness} Every problem $\D\in\dNP$ admits
an algorithm-family $\A$ that solves $\D$ with update time $\time_{\A}(n)=O(\polylog(n))$ and space complexity $\text{Space}_{\A}(n) = O(\poly (n))$, and uses an oracle-family $\O'$ for $\fDT$ with $\qsize_{\A}(n)=O(2^{\polylog(n))})$. 
\end{thm}

As a corollary of the above theorem, we will derive that $\fDT$ is $\dNP$-hard.

\renewcommand{\S}{\mathcal{S}}

\begin{cor}
\label{cor:np-hardness} $\fDT$ is $\dNP$-hard.
\end{cor}

The proofs of~\ref{th:np-hardness} and~\ref{cor:np-hardness} appears in~\ref{sec:th:np-hardness} and~\ref{sec:cor:np-hardness}.

\subsection{Proof of~\ref{th:np-hardness}}
\label{sec:th:np-hardness}

Throughout the proof, we use the notations and concepts introduced in~\ref{sec:class}. We fix a problem $\D\in\dNP$ and construct an algorithm-family
$\A$ for \ref{th:np-hardness}. In more details, we define a function
$m:\mathbb{N}^{+}\rightarrow\mathbb{N}^{+}$ of the form $m(n)=O(2^{\polylog(n)})$,
and show that for every $n\in\mathbb{N}^{+}$, there exists an algorithm
$\A_{n}$ that solves $\D_{n}$ with $\time_{\A}(n)=O(\polylog(n))$ and $\text{Space}_{\A}(n) = O(\poly (n))$,
using the oracle $\O'_{m(n)}$ for $\fDT$. To ease notation, henceforth
we  write $m$ instead of $m(n)$.

\subsubsection{High-level Strategy}

Suppose that the algorithm $\A_{n}$ is given the instance-sequence
$(I_{0},\ldots,I_{k})$ of $\D_{n}$. Since $\D\in\dNP$, there
exists a verifier $\V_{n}$ for $\D_{n}$ as per \ref{def:NP}.
At a  high-level, the algorithm $\A_{n}$ uses the verifier $\V_{n}$ as a subroutine and works
as follows.

\paragraph{Preprocessing step.}

At step $t=0$, $\A_{n}$ does the following:
\begin{enumerate}
\item \label{enu:init V}Call the verifier subroutine $\V_{n}$ with input
$\inp_{\V_{n}}(t_{\V_{n}})=I_{0}$ at step $t_{\V_{n}}=0$. Then $\V_{n}$
returns an ordered pair $(x_0, y_0)$.
\item \label{enu:init O}Call the oracle $\O'_{m}$ in a certain way. 
\item Output $\A_{n}(\inp_{\A_{n}}(0))= x_0$. 
\end{enumerate}

\paragraph{Update step.}

Subsequently, at each step $t>0$, $\A_{n}$ does the following:
\begin{enumerate}
\item \label{enu:update 1}Call the oracle $\O'_{m}$ (several times) in
a certain way, and use the outputs of these calls to come up with
a \emph{reward maximizing proof} $\pi_{t}$ as per \ref{def:NP}.
\item \label{enu:update 2}Call the verifier subroutine $\V_{n}$ with input
$\inp_{\V_{n}}(t_{\V_{n}})=\left((I_{t-1},I_{t}),\pi_{t}\right)$
at step $t_{\V_{n}}=t$. The verifier $\V_n$ returns an ordered pair $(x_t, y_t)$.
\item \label{enu:update 3}Call the oracle $\O'_{m}$ (several times) in
a certain way.
\item Output $\A_{n}(\inp_{\A_{n}}(0\rightarrow t))= x_t$. 
\end{enumerate}
As the verifier $\V_{n}$ is given a reward-maximizing proof-sequence as inputs, \ref{def:NP} implies 
that the algorithm $\A_{n}$ solves $\D_{n}$: 
\begin{lem}
\label{lem:correct} 
For every $t\in\{0,\ldots,k\}$, we have
$\A_{n}(\inp_{\A_{n}}(0\rightarrow t))=\D(I_{t})$.\end{lem}
\begin{proof}
As $\A_{n}$ calls $\V_{n}$ exactly once at each step, we always have $t_{\V_{n}}=t$. Recall~\ref{def:NP}, observe that if $\V_{n}$ is always
given a reward-maximizing proof at each step, i.e. $(\pi_{1},\dots,\pi_{k})$
is a reward-maximizing proof-sequence w.r.t. $(I_{0},\dots,I_{k})$, then
$x_t=\D(I_{t_{\V n}})$
for all $t_{\V n}\in\{0,\ldots,k\}$.
Since the algorithm $\A_{n}$ outputs $x_t$  at
every step $t$, the lemma holds.
\end{proof}
It now remains to specify \ref{enu:init O} in the preprocessing
step, and \ref{enu:update 1} and \ref{enu:update 3} in the update
step. These steps are key to constructing the reward-maximizing proof-sequence
for $\V_{n}$. We specify these steps respectively  in the subsections below.

\subsubsection{Initializing the Oracle }

\label{sub:hardness init oracle}

We specify \ref{enu:init O} in the preprocessing step of $\A_{n}$
as follows: $\A_{n}$ calls $\O'_{m}$ by giving the initial input
$I'$ where $I'=(\mem_{I'}, \T_{I'})$ is an instance of
$\fDT$ as defined below.
\begin{itemize}
\item We set $\mem_{I'}=\mem_{\V_{n}}(0)$ where $\mem_{\V_{n}}(0)$ is
just the memory state of $\V_{n}$ after \ref{enu:init V} in the
preprocessing step of $\A_{n}$ (i.e. after $\V_{n}$ returns). 
\item We set $\T_{I'}=\T_{\V_{n}}$, where the collection $\T_{\V_{n}}$
of decision-trees is defined as follows. Recall that the input of
$\V_{n}$ at step $t$ is of the form $((I_{t-1},I_{t}),\pi_{t})$
consisting of the \emph{instance-update} $(I_{t-1}, I_t)$ and the \emph{proof} $\pi_t$,
respectively. Let $T_{\V_{n}}$ be the decision tree of $\V_{n}$.
For each possible proof $\pi\in\{0,1\}^{\polylog(n)}$,
let $T_{\pi}$ be the decision tree obtained from $T_{\V_{n}}$ by
``fixing'' the ``proof-part" of the input be to $\pi$. More
specifically, consider every read node $u\in$$T_{\V_{n}}$ whose
index points to some $i$-th bit $\pi[i]$ of $\pi$ (in the proof-part
of the input).\footnote{Recall the formal description of a decision tree from~\ref{sub:formal alg}.} Let $p(u),r(u),l(u)$ be the parent, right child, and
left child of $u$, respectively. To construct $T_{\pi}$, we remove
$u$ and if $\pi[i]=0$, then add an edge $(p(u),l(u))$, and remove
the subtree rooted at $r(u)$. Else, if $\pi[i]=1$, then add $(p(u),r(u))$
and remove the subtree rooted at $l(u)$. We set $\T_{\V_{n}}=\{T_{\pi}\mid\pi\in\{0,1\}^{\polylog(n)}\}$.
\item We now define the ranks $r(v)$. Consider any proof $\pi \in \{0,1\}^{\polylog (n)}$ and the corresponding decision tree $T_{\pi} \in \T_{\V_n}$.~\ref{def:NP} guarantees that when given any instance-update $(I_{t-1}, I_t)$ and the proof $\pi$ as input, the verifier $\V_n$  outputs some ordered pair $(x, y)$ where $x \in \{0,1\}$ and $y \in \{0,1\}^{\polylog(n)}$. This has the following important implication. 
\begin{itemize}
\item Consider any leaf node $v$ in the decision tree $T_{\pi}$. We can associate  an ordered pair $(x_v, y_v)$ with this leaf-node $v$, where  $x_v \in \{0,1\}$ and $y_v \in \{0,1\}^{\polylog (n)}$, such that whenever the  decision-tree $T_{\pi}$ follows the root-to-leaf execution path ending at the node $v$, it writes $(x_v, y_v)$ in the output memory. Note that the ordered pair $(x_v, y_v)$ does not depend on the contents of the memory the decision tree $T_{\pi}$ operates on.\footnote{The contents of the memory determines which root-to-leaf path  becomes the {\em execution-path} followed by $T_{\pi}$, and every root-to-leaf path  is associated with an ordered pair $(x_v, y_v)$ where $v$ is the end leaf-node of the concerned path.} We define the rank of a leaf-node $v$ in $T_{\pi}$ to be: $r(v) = y_v$.
\end{itemize}
\end{itemize}
In order to prove that $I'$ is indeed an instance of $\fDT$, it remains to show the following:
\begin{lem}
\label{lm:shallow} All the decision trees $T_{\pi}\in\T_{I'}$ are
shallow.\end{lem}
\begin{proof}
Note that there are $2^{\polylog (n)}$
many decision trees in the collection $\T_{I'}$, one for each bit
string $\pi\in\{0,1\}^{\polylog (n)}$. Furthermore, since the verifier $\V_n$ has $O(\polylog (n))$ update time, each decision tree
$T_{\pi}\in\T_{I'}$ has depth $O(\polylog (n))$. Hence,
the depth of each tree $T_{\pi}\in\T_{I'}$ is at most $O(\poly\log|\T_{I'}|)$,
which implies that all the trees in $\T_{I'}$ are shallow according
to our definition. 
\end{proof}

\subsubsection{Constructing the Reward-Maximizing Proof}

\label{sub:hardness update}

We maintain the following two invariants in the beginning of every step $t > 0$.
\begin{enumerate}
\item The sequence of proofs $(\pi_{1},\dots,\pi_{t-1})$ that the verifier $\V_{n}$ received so far is the
\emph{reward-maximizing proof-sequence} w.r.t. $(I_{0},\dots,I_{t-1})$.
\item The $\fDT$ instance $I'=(\mem_{I'},\T_{I'})$ maintained
by $\O'_{m}$ is such that $\mem_{I'}=\mem_{\V_{n}}(t-1)$ where $\mem_{\V_{n}}(t-1)$
is the memory state of $\V_{n}$ after finishing the step $t-1$.
\end{enumerate}
We now describe \ref{enu:update 1} and \ref{enu:update 3} in the
update step of $\A_{n}$. Specifically, we show how to call $\O'_{m}$
several times to keep the invariant. After receiving the input $\inp_{\A_{n}}(t)=(I_{t-1},I_{t})$ at
step $t$, the precise description of \ref{enu:update 1} in the update
step of $\A_{n}$ is as follows:
\begin{itemize}
\item Write $(I_{t-1},I_{t})$ in the input part $\mem_{\V_{n}}^{\inp}$
of $\mem_{\V_{n}}$.
\item Make a sequence of calls to $\O'_{m}$ to update $\mem_{I'}$ so that
$\mem_{I'}=\mem_{\V_{n}}$. 
\end{itemize}
At this point, let $T_{\pi_{t}} \in \T_{I'} = \T_{\V_n}$ be the output of the oracle $\O'_{m}$.
We claim the following:
\begin{claim}
The proof $\pi_{t}$ is the {\em reward-maximizing proof} at step $t$.
\end{claim}
\begin{proof}
From~\ref{sub:DT def} and~\ref{sub:hardness init oracle}, it follows that the answer from
$\O'_{m}$ is $\P'(I')=T_{\pi_{t}}$, where  $T_{\pi_{t}}$ is the decision tree  $T \in \T_{\V_{n}}$ which maximizes the rank $r(v^*_T)$.

By the invariant, it holds that the sequence of proofs $(\pi_{1},\dots,\pi_{t-1})$ that the verifier
$\V_{n}$ received so far is the reward-maximizing proof-sequence w.r.t.
$(I_{0},\dots,I_{t-1})$. Moreover, recall that $(I_{t-1},I_{t})$ is just written
the input part $\mem_{\V_{n}}^{\inp}$ of $\mem_{\V_{n}}$. This implies
that $\pi_{t}$ is the proof which maximizes $y_t$ (the second part of the output $(x_t, y_t)$ of the verifier $\V_n$ when it is given $((I_{t-1},I_{t}),\pi_{t})$ as input in step $t$). Hence, we conclude that $\pi_{t}$ is the desired reward-maximizing proof at step $t$.
\end{proof}
Finally,~\ref{enu:update 3} in the update step of $\A_{n}$ does the following:
\begin{itemize}
\item Make a sequence of calls to $\O'_{m}$ to update $\mem_{I'}$ so that
$\mem_{I'}=\mem_{\V_{n}}$. 
\end{itemize}
Clearly, this ensures that  the invariant is maintained.

\subsubsection{Analyzing Update Time, Space Complexity and Blow-up Size of $\A_n$}

We are now ready to prove \ref{th:np-hardness}. The theorem holds  since  $\A_{n}$ solves $\D_{n}$
(\ref{lem:correct}) and has small blow-up size (\ref{lm:qsize}) and
update time and space complexity (\ref{lm:update:time}).

\begin{lem}
\label{lm:qsize} It takes $m=O(2^{\polylog (n)})$ bits
to encode the $\fDT$ instance $I'$. In other words, the algorithm
$\A_{n}$ has blow-up size $\qsize_{\A}(n)=O\left(2^{\polylog(n)}\right)$. \end{lem}
\begin{proof}
From the proof of \ref{lm:shallow}, we deduced that there are $O\left(2^{\polylog(n)}\right)$
decision trees in the collection $\T_{I'}$. Furthermore, each tree
$T_{\pi}\in\T_{I'}$ has depth at most $O(\polylog(n))$,
which implies that each tree $T_{\pi}\in\T_{I'}$ contains at most
$O\left(2^{\polylog(n)}\right)$ nodes. Thus, there are
at most $O\left(2^{\polylog(n)}\right)\times O\left(2^{\polylog(n)}\right)=O\left(2^{\polylog(n)}\right)$
many nodes over the collection of trees $\T_{I'}$. This also implies
that the memory $\mem_{I'}$ contains at most $O\left(2^{\polylog(n)}\right)$
bits, for the number of bits in $\mem_{I'}$ can w.l.o.g. be assumed to be upper bounded by the number of nodes
in $\T_{I'}$ (otherwise, there will be some bits in $\mem_{I'}$
that are not accessible to any tree $T_{\pi}\in\T_{I'}$). We therefore
conclude that the running instance $I'=(\mem_{I'},\T_{I'})$
that the algorithm $\A_{n}$ asks the oracle $\O'_{m}$ to maintain
an answer to can be encoded using $m=O\left(2^{\polylog(n)}\right)$
bits. 
\end{proof}
\begin{lem}
\label{lm:update:time} The algorithm $\A_{n}$ has  update time
 $O(\polylog(n))$ and  space complexity $O(\poly (n))$.
\end{lem}
\begin{proof}
We analyze the time of $\A_{n}$ in the update steps. In \ref{enu:update 1},
the algorithm $\A_{n}$ writes $(I_{t-1},I_{t})$ in the input-part $\mem_{\V_{n}}^{\inp}$. Since $(I_{t-1},I_{t})$ is specified using $O(\polylog (n))$ bits (see~\ref{assume:logn update size}), the algorithm $\A_{n}$ makes
at most $O(\polylog (n))$ many calls to $\O'$ to update $\mem_{I'}=\mem_{\V_{n}}$.
In \ref{enu:update 2}, the algorithm $\A_{n}$ calls $\V_{n}$ which takes $O(\polylog (n))$ time.
Hence, the call to $\V_n$ changes at most $O(\polylog (n))$ bits in $\mem_{\V_n}$. Accordingly, in \ref{enu:update 3}, the algorithm $\A_{n}$ makes at most $O(\polylog(n))$
many calls to update $\mem_{I'}=\mem_{\V_{n}}$. Thus,  in total the algorithm $\A_n$ takes
$O(\polylog(n))$ time.

Finally, the space complexity of $\A_n$ is given by: $\text{Space}_{\A}(n) = O\left( \text{Space}_{\V}(n)\right) = O\left(\poly (n) \right)$. The last equality follows from~\ref{def:NP}.
\end{proof}

\subsection{Proof of~\ref{cor:np-hardness}}
\label{sec:cor:np-hardness}

Consider any $\D \in \dNP$. Throughout this section, we assume that $\fDT \in \dP$. Under this assumption, we will prove: $\D \in \dP$. This will imply that $\fDT$ is $\dNP$-hard.

As per~\ref{th:np-hardness}, there is an algorithm-family $\A$ that solves $\D$ using an oracle $\O'$ for $\fDT$, with update time $\time_{\A}(n) = O(\polylog (n))$, space complexity $\text{Space}_{\A}(n) = O(\poly (n))$, and blow-up size $\qsize_{\A}(n) = O(2^{\polylog (n)})$. Throughout the rest of the proof, fix any integer $n \geq 1$ and let $m = m(n) = \qsize_{\A}(n) = O(2^{\polylog (n)})$. Thus, the algorithm $\A_n$ solves $\D_n$ by using the oracle $\O'_m$ for $\fDT$. As described in~\ref{sub:hardness init oracle}, the oracel $\O'_m$ keeps track of the output of the $\fDT$ problem on the instance $(\mem_{I'}, \T_{I'})$. 

Since $\fDT \in \dP$, it follows that there is an algorithm-family $\A'$ that solves $\fDT$ with pologarithmic update time and polynomial space complexity. Using this fact, we will now design an algorithm $\A^*_n$ for $\D_n$ that has $O(\polylog (n))$ update time and $O(\poly (n))$ space complexity, and does {\em not} use any oracle. This will imply that $\D \in \dP$, thereby concluding the proof of the corollary.

As a first attempt, let us try to design $\A^*_n$ as follows. The algorithm $\A^*_n$ mimics the behavior of the algorithm $\A_n$ (as specified in~\ref{sec:th:np-hardness}). The only difference is that instead of using the oracle $\O'_m$, the algorithm $\A^*_n$ uses $\A'_m$ as a subroutine. To be more specific, whenever the algorithm $\A_n$ calls the oracle $\O'_m$ (see~\ref{sec:th:np-hardness}), the algorithm $\A^*_n$ calls the subroutine $\A'_m$. Clearly, if we design $\A^*_n$ in this manner, then it will always give the same output as $\A_n$.  Unfortunately, however, we have $m = 2^{\polylog (n)}$, and thus if we design $\A^*_n$ in this manner then the space complexity of $\A^*_n$ will be dominated by $\text{Space}_{\A'}(m) = O(\poly (m)) = O(2^{\polylog (n)})$. To address this concern, we ensure 
that  $\A^*_n$  uses the subroutine $\A'_m$ in a {\em white box} manner. In particular:
\begin{itemize}
\item The subroutine $\A'_m$ does not have direct access to its memory $\mem_{\A'_m}$. Indeed, since $m = 2^{\polylog (n)}$ and $\text{Space}_{\A'}(m) = O(\poly (m))$, there are $\poly (m) = 2^{\polylog (n)}$ bits in $\mem_{\A'_m}$. As we want  $\A^*_n$ to have $\poly (n)$ space complexity, we cannot afford to store all the bits of $\mem_{\A'_m}$.
\item During a call to the subroutine $\A'_m$, whenever $\A'_m$ wants to read the content of (say) the $i^{th}$ bit in $\mem_{\A'_m}$,  it passes the value of $i$ to the algorithm $\A^*_n$, and the algorithm $\A^*_n$ returns the content of the $i^{th}$ bit in $\mem_{\A'_m}$ to  $\A'_m$ by calling a different subroutine $\S^*_n$ (to be described below). Similarly, whenever $\A'_m$ wants to write (say) the $i^{th}$ bit of $\mem_{\A'_m}$ with some bit $b \in \{0,1\}$, then it passes the ordered pair $(i, b)$ to the algorithm $\A^*_n$, and the algorithm $\A^*_n$ in turn calls the subroutine $\S^*_n$ with $(i, b)$ as input to handle this operation. 

To summarize, the subroutine $\S^*_n$ acts as an {\em interface} between $\A'_m$ and its memory $\mem_{\A'_m}$. The crucial point is this: Although $\mem_{\A'_m}$ is of size $\poly (m) = 2^{\polylog (n)}$, the memory $\mem_{\S^*_n}$ of the subroutine $\S^*_n$ itself will be of size $\poly (n)$. Hence, the subroutine $\S^*_n$ is able to store only a tiny fraction of the memory bits in $\mem_{\A'_m}$. In spite of this severe restriction, we will show that we can design such a subroutine $\S^*_n$ with $\polylog (n)$ update time to act as an interface between $\A'_m$ and $\mem_{\A'_m}$, {\em as long as $\A^*_n$ has to deal with at most $\poly (n)$ update-steps for $\D_n$.}
\end{itemize}
It now remains to describe the subroutine $\S^*_n$. Towards this end, we first need to define the notion of a {\em canonical instance} for the $\fDT$ problem {\em with respect to the} problem $\D$.

\medskip
\noindent {\bf Canonical instance $\mathbf{I}'$:} From~\ref{sub:hardness init oracle}, recall that the oracle $\O'_m$ (and hence the subroutine $\A'_m$) deals with input instances of the form $(\mem_{I'}, \T_{I'})$. Note that $\mem_{I'}$ contains only $\poly (n)$ many bits. This is because the memory $\mem_{\V_n}$ of the verifier $\V_n$ is of $\poly (n)$ size and $\mem_{I'}$ reflects the state of $\mem_{\V_n}$. The total size of the decision trees in $\T_{I'}$, however, is $2^{\polylog (n)}$, and this is the reason why we have $m = O(2^{\polylog (n)})$. The canonical instance of $\fDT$ with respect to $\D_n$  is defined as $\mathbf{I}' = (\mem_{I'}, \T_{I'})$, {\em where all the bits in $\mem_{I'}$ are set to $0$}. 

Let $\mem_{\A'_m}(\mathbf{I}')$ denote the state of the memory $\mem_{\A'_m}$ when  $\A'_m$ receives the canonical instance $\mathbf{I}'$ as input in the {\em preprocessing step}. Crucially,  note that the contents of $\mem_{\A'_m}(\mathbf{I}')$ {\em is completely determined by} the $n$-slice $\D_n$ of the problem $\D$. Since there are $\poly (m) = 2^{\polylog (n)}$ bits in $\mem_{\A'_m}$, it is easy to design a subroutine $\Z^*_n$ in the bit-probe model that takes as input an index $i \in \{0,1\}^{\polylog (n)}$ and returns the contents of the $i^{th}$ bit of $\mem_{\A'_m}(\mathbf{I}')$. Specifically, the decision tree of $\Z^*_n$ will have depth $\polylog (n)$, and each leaf in this decision tree will correspond to a unique index $i \in \{0,1\}^{\polylog (n)}$. The leaf corresponding to $i \in \{0,1\}^{\polylog (n)}$ will contain the $i^{th}$ bit of $\mem_{\A'_m}(\mathbf{I}')$. Since contents of $\mem_{\A'_m}(\mathbf{I}')$ do not change (it is determined by $\mathbf{I}'$), the subroutine $\Z^*_n$ will  need to access its own memory $\mem_{\Z^*_n}$ only when reading the input $i \in \{0,1\}^{\polylog (n)}$ and producing the desired output $b \in \{0,1\}$. In other words, the subroutine $\Z^*_n$ returns the content of any given bit of $\mem_{\A'_m}(\mathbf{I}')$ in $O(\polylog (n))$ time, and it only uses $O(\polylog (n))$ space.

\medskip
\noindent {\bf The set $Y$:} The subroutine $\S^*_n$ will also store a collection of ordered pairs of the form $(i, b)$, where $i \in \{0,1\}^{\polylog (n)}$ and $b \in \{0,1\}$, in a set $Y$. The set $Y$ will be maintained as a balanced search tree.  The subroutine $\S^*_n$ will use $Y$, along with the subroutine $\Z^*_n$ describe above, to return the content of a given bit of $\mem_{\A'_m}$. Specifically, suppose that the subroutine $\S^*_n$ is asked for the content of the $i^{th}$ bit in $\mem_{\A'_m}$. It will first check if the set $Y$ contains an ordered pair of the form $(i, b)$. If yes, then the subroutine $\S^*_n$ will return $b$ as  output. If no, then the subroutine $\S^*_n$ will return the $i^{th}$ bit of $\mem_{\A'_m}(\mathbf{I}')$ as output (after making a single call to $\Z^*_n$ with input $i$).

\medskip
We are now ready to state the subroutine $\S^*_n$ in details. See~\ref{sec:new:subroutineS}.

\subsubsection{The subroutine $\S^*_n$}
\label{sec:new:subroutineS}
The job of the subroutine $\S^*_n$ is to act as an interface between $\A'_m$ and $\mem_{\A'_m}$. We assume that $\A'_m$ gets $\mathbf{I}'$ as input in the preprocessing step. (In~\ref{sub:sec:wrapup} we will get rid of this assumption.)

\medskip
\noindent {\bf Initialization:}  $\A'_m$ gets $\mathbf{I}'$ as input in the preprocessing step. At this point, we have 
$\mem_{\A'_m} = \mem_{\A'_m}(\mathbf{I}')$, and the subroutine $\S^*_n$ sets $Y = \emptyset$.

\medskip
\noindent {\bf Writing a bit in $\mem_{\A'_m}$:} Suppose that $\A'_m$ wants to write $b \in \{0,1\}$ in the $i^{th}$ bit of $\mem_{\A'_m}$. Accordingly, we call the  subroutine $\S^*_n$ with the ordered pair $(i, b)$ as input. The subroutine $\S^*_n$ first checks if there is any ordered pair of the form $(i, b')$ in the set $Y$, and if yes, then it deletes that ordered pair $(i, b')$ from $Y$. Next, it inserts the ordered pair $(i, b)$ into $Y$.

\medskip
\noindent {\bf Reading a bit from $\mem_{\A'_m}$:} Suppose that $\A'_m$ wants to read the content of the $i^{th}$ bit in $\mem_{\A'_m}$. Accordingly, we call the subroutine $\S^*_n$ with $i$ as input. If this bit was modified by $\A'_m$ after the initialization step, then there is an ordered pair of the form $(i, b)$ in the set $Y$ where $b \in \{0,1\}$ denotes the current content of the $i^{th}$ bit of $\mem_{\A'_m}$. Accordingly, the subroutine $\S^*_n$ first searches for an ordered pair of the form $(i, b), b \in \{0,1\},$ in the set $Y$. If it finds such an ordered pair $(i, b)$ in $Y$, then it returns $b$ as the output. Otherwise, if it fails to find such an ordered pair in $Y$, then $\A'_m$ has not modified the $i^{th}$ bit in $\mem_{\A'_m}$ since the initialization step. Accordingly, in this event $\S^*_n$ returns the content of the $i^{th}$ bit of $\mem_{\A'_m}(\mathbf{I}')$ after calling  the subroutine $\Z^*_n$ with input $i$.

\medskip
\noindent {\bf Update time and space complexity of $\S^*_n$:} The key observation is this. Since the initialization step described above, suppose that   $\A'_m$ has made at most $\poly (n)$ many (read/write) {\em probes} to its memory $\mem_{\A'_m}$. To implement each such probe one call was made to the subroutine $\S^*_n$, for it acts as an interface between $\A'_m$ and its memory $\mem_{\A'_m}$. Then the set $Y$ is of size at most $\poly (n)$, because each call to $\S^*_n$ can add at most one ordered pair to $Y$, and initially we had $Y = \emptyset$.  Now, the space of complexity of the subroutine $\S^*_n$ is dominated by the space needed to store the set $Y$, and its update time is at most $O(\polylog (n)) + O(\log |Y|)$. The $O(\polylog (n))$ term comes from the fact that a call to $\Z^*_n$ takes $O(\polylog (n))$ time, whereas the $O(\log |Y|)$ term comes from the fact that the set $Y$ is stored as a balanced search tree, and hence searching for an element in $Y$ takes $O(\log |Y|)$ time. Thus, as long as we ensure that $\A'_m$ makes $\poly (n)$ many probes to its memory $\mem_{\A'_m}$, it will hold that the subroutine $\S^*_n$ can act as an interface between $\A'_m$ and its memory $\mem_{\A'_m}$ with update time $= O(\polylog (n)) + O(\log |Y|) = O(\polylog (n))$ and space complexity $= O(|Y|) = O(\poly (n))$.

\subsubsection{How the subroutine $\S^*_n$ is used by the algorithm $\A^*_n$}
\label{sub:sec:wrapup}

The algorithm $\A^*_n$ mimics the behavior of the algorithm $\A_n$ from~\ref{sec:th:np-hardness}, with two differences. 
\begin{enumerate}
\item \label{item: cannon 1} Whenever $\A_n$ calls the oracle $\O'_m$,  $\A^*_n$ calls the subroutine $\A'_m$ in a {\em white box} manner.
\item \label{item: cannon 2} The subroutine $\S^*_n$ acts as an interface between $\A'_m$ and $\mem_{\A'_m}$.
\end{enumerate}
From the discussion in~\ref{sec:new:subroutineS}, it becomes clear that as long as  $\A'_m$ makes at most $\poly (n)$ many probes to $\mem_{\A'_m}$, the space complexity of  $\S^*_n$ will be bounded by $O(\poly (n))$. We now show how to enforce this condition.

\medskip
\noindent {\bf Preprocessing step for $\A^*_n$:} Recall the discussion in~\ref{sub:hardness init oracle}. Here, $\A^*_n$ wants to initialize $\A'_m$ with the input $I' = (\mem_{I'}, \T_{I'})$, where $\mem_{I'} = \mem_{\V_n}(0)$. This is done as follows.

As in~\ref{sec:new:subroutineS}, $\A^*_n$ starts  by giving $\mathbf{I}'$ as input to  $\A'_m$. Recall that $\mathbf{I}' = (\mem_{I'}, \T_{I'})$, where each bit in $\mem_{I'}$ is set to $0$. Furthermore, recall that $\mem_{I'}$ consists of $\poly (n)$ many bits, since $\mem_{I'}$ is supposed to reflect the state of $\mem_{\V_n}$, which in turn has at most $\poly (n)$ bits. Accordingly, $\A^*_n$ now asks the subroutine $\A'_m$ to handle $\poly (n)$ many instance-updates on $I' = (\mem_{I'}, \T_{I'})$, where each instance-update changes one bit in $\mem_{I'}$ in such a way that: At the end of these instance-updates, the subroutine $\A'_m$ ends up with the input $(\mem_{\V_n}(0), \T_{I'})$.

This is how $\A^*_n$ initializes the subroutine $\A'_m$. (For comparison, recall how $\A_n$ initializes the oracle $\O'_m$ in~\ref{sub:hardness init oracle}.) Note that until this point, $\A'_m$ clearly makes at most $\poly (n)$ many probes to $\mem_{\A'_m}$, and hence the space complexity of $\S^*_n$ is at most $O(\poly (n))$.

\medskip
\noindent {\bf Update-steps for $\A^*_n$:} The algorithm $\A^*_n$ mimics the behavior of $\A_n$ during an update-step at time $t > 0$ (see~\ref{sub:hardness update}). The only  differences between $\A^*_n$ and $\A_n$ have been emphasized in~\ref{item: cannon 1} and~\ref{item: cannon 2} above. It follows from the proof of~\ref{lm:update:time} that $\A^*_n$ makes $O(\polylog (n))$ many calls to $\A'_m$ during each update-step. Since $\A'_{m}$ has update time $O(\polylog (m)) = O(\polylog (2^{\polylog (n)})) = O(\polylog (n))$, we conclude that $\A'_m$ makes at most $O(\polylog (n)) \cdot O(\polylog (n)) = O(\polylog (n))$ many probes to $\A'_m$ during each update-step of $\A^*_n$.

\medskip
\noindent From the discussion above, we reach the following conclusion. As long as $\A^*_n$ has to deal with at most $\poly(n)$ many update-steps, $\A'_m$ makes at most $\poly(n) \cdot \polylog (n) = \poly(n)$ many probes to $\A'_m$ in total, and so the space complexity and the update time of the subroutine $\S^*_n$ remains at most $\poly (n)$ and $\polylog (n)$ respectively.

\medskip
\noindent {\bf Bounding the space-complexity and the update time of $\A^*_n$:} Consider a sequence of at most $\poly (n)$ many instance-updates given to $\A^*_n$ as input (see~\ref{assume:poly}). The space complexity of $\A^*_n$ is bounded by the sum of the space complexities of the verifier $\V_n$ and the subroutine $\S^*_n$. By definition, we have $\text{Space}_{\V}(n) = O(\poly (n))$ since $\D \in \dNP$. From the discussion above, it follows that the space complexity of $\S^*_n$ is also at most $O(\poly (n))$. Thus, we conclude that the overall space complexity of $\A^*_n$ is also $O(\poly (n))$.
 
To bound the update time of $\A^*_n$, recall the proof of~\ref{lm:update:time}. During each update-step at $t > 0$, the algorithm $\A^*_n$ spends at most $O(\polylog (n))$ time (excluding the calls to $\A'_m$) and makes $O(\polylog (n))$ many calls to $\A'_m$. Since $\A'_m$ has update time $O(\polylog (m)) = O(\polylog (2^{\polylog (n)})) = O(\polylog (n))$, each call to $\A'_m$ requires $O(\polylog (n))$ time (excluding the calls to $\S^*_n$) and requires a further $O(\polylog (n))$ calls to the subroutine $\S^*_n$. Finally, under~\ref{assume:poly} we have already shown in~\ref{sub:sec:wrapup} that each call to $\S^*_n$ takes $O(\polylog (n))$ time. Thus, we conclude that the update time of $\A^*_n$ is at most  $O(\polylog (n))$.

\medskip
To summarize, the algorithm $\A^*_n$ correctly solves $\D_n$ on any sequence of instance-updates of length at most $\poly (n)$. It has $\polylog (n)$ update time and $\poly (n)$ space complexity. Thus, we have $\D \in \dP$ (see~\ref{assume:poly} and~\ref{def:P}). 

In other words, if $\fDT \in \dP$ then every problem $\D \in \dNP$ belongs to the class $\dP$. Hence, we derive that $\fDT$ is $\dNP$-hard.

\section{$\protect\dNP$-complete/hard Problems}

\label{sec:DNF}

In this section, we consider a problem called the \emph{dynamic narrow
DNF evaluation problem} (or $\dDNF$ for short). 
In~\ref{sub:DNF:NPhard}, we show that this problem is $\dNP$-complete, which, in turn, implies $\dNP$-hardness
of many other problems (see  \ref{cor:list rankNP hard} in \ref{sub:list NP hard}).

\medskip
\noindent {\bf DNF formula:}
A $w$-width DNF formula $F$ with $n$ variables and $m$ clauses is defined as follows. Let $\X=\{x_{1},\dots,x_{n}\}$
be the set of variables and $\mathcal{C} = \{C_1, \ldots, C_m\}$ be the set of clauses. We have $F=C_{1}\vee\dots\vee C_{m}$, where each clause
$C_{j}$ is a conjunction (AND) of at most $w$ literals (i.e. $x_{i}$
or $\neg x_{i}$ where $x_{i}\in\X$). Let $\phi:\X\rightarrow\{0,1\}$
be an assignment of variables. Let $C_j(\phi)\in\{0,1\}$ be the
value of a clause $C_j$ after assigning the value of each $x_{i}$ with $\phi(x_{i})$.
If $C_{j}(\phi)=1$, then we say that $C_{j}$ is \emph{satisfied }by $\phi$. Similarly, let $F(\phi) = C_1(\phi) \vee \cdots \vee C_m(\phi)$ be the value of $F$ under the assignment $\phi$. We say that $F$ is satisfied by $\phi$ if $F(\phi) = 1$.
\begin{defn}
\label{def:range DNF}In the dynamic narrow DNF evaluation problem
($\dDNF$), we are first given 
\begin{itemize}
\item a $w$-width DNF formula $F$ over $n$ variables $\X=\{x_{1},\dots,x_{n}\}$ 
and $m$ clauses $\mathcal{C} = \{C_1, \ldots, C_m\}$, where $w=\polylog(m)$, and
\item an assignment $\phi:\X\rightarrow\{0,1\}$. 
\end{itemize}
An {\em update}  is denoted by $(i,b)\in[n]\times\{0,1\}$, which modifies the assignment $\phi$ by setting  $\phi(x_{i})=b$. After each update, we ask if $F(\phi)=1$.
\end{defn}

Note that, for any instance $(F,\phi)$ of $\dDNF$, we can assume
w.l.o.g. that $n\le mw=\tilde{O}(m)$ because we can ignore all variables
that are not in any clause. So we have:
\begin{prop}
\label{prop:size DNF} An instance $(F,\phi)$ where $F$ has $m$ clauses can
be represented using $\tilde{O}(m)$ bits.  Also, there is a trivial dynamic algorithm for the $\dDNF$ problem 
with $\tilde{O}(m)$ update time. 
\end{prop}

\subsection{$\protect\dNP$-completeness of $\protect\dDNF$ }

Our main result in this section is summarized below.

\label{sub:DNF:NPhard}
\begin{thm}
\label{thm:DNF complete}The $\dDNF$ problem is $\dNP$-complete.
\end{thm}
It is easy to see that $\dDNF\in\dNP$. Let $(F,\phi)$ be an
instance of $\dDNF$ where $F$ has $n$ variables and $m$ clauses, which is represented using $m' = \tilde{O}(m)$ bits.
After each update, the proof is simply an index $j$ where $C_{j}(\phi)=1$.
Then, a verifier $\V_{m'}$ accepts (i.e. outputs $1$) if $C_{j}(\phi)=1$,
otherwise $\V_{m'}$ rejects (i.e. outputs 0).  In order to check the value of $C_j(\phi)$, the verifier 
 $\V_{m'}$ only needs to read the (at most $w$) literals of $C_{j}$ which takes $O(w) = \polylog(m)=\polylog(m')$ time. Hence, the update time of the verifier $\V_{m'}$ is $O(\polylog (m'))$. This shows
that  $\dDNF\in\dNP$. 

We devote the rest of this section towards showing that the $\dDNF$ problem is $\dNP$-hard. We do this
in steps. We first define an intermediate problem called $\fDNF$.
Then we give a $\dP$-reduction from $\fDT$ to $\fDNF$ (see~\ref{lem:reduction:1}) and then from $\fDNF$ to $\dDNF$ (see~\ref{lem:first to exist}). This implies that there is a $\dP$-reduction from $\fDT$ to $\dDNF$ (see~\ref{thm:easiness transfering}). Since $\fDT$ is $\dNP$-hard (see~\ref{th:np-hardness}), we derive that $\dDNF$ is also $\dNP$-hard (see~\ref{cor:reductions}). We have already established that $\dDNF$ is in $\dNP$, and hence we conclude that $\dDNF$ is $\dNP$-complete.

\medskip
\noindent {\bf The $\fDNF$ problem:}
The definition of $\fDNF$ is the same as $\dDNF$ except the following. There is a total order $\prec$ defined on the set of clauses $\mathcal{C}$. After each update, instead, we must return  the first clause $C_j$ (according to the total order $\prec$) which has $C_j(\phi) = 1$, provided such a clause exists,  and $0$ otherwise (which indicates that $F(\phi)=0$). 
\begin{lem}
\label{lem:reduction:1}
$\fDT$ is $\dP$-reducible to $\fDNF$.\end{lem}
\begin{proof}
Let $(\mem,\T)$ be an instance of $\fDT$. We construct an
instance $(F,\phi)$ of $\fDNF$  as follows. For each decision tree
$T\in\T$, we construct a DNF formula  $F_{T}$ using a standard mapping from
decision trees to DNF instances (see e.g. \cite[Proposition 4.5]{ODonnell}).
More precisely, for every root-to-leaf path $P$ in $T$, the DNF formula $F_{T}$ contains a conjunctive clause $C_{P}$. The set of literals in $C_{P}$
correspond to the set of read nodes in $P$, as follows. (Recall the formal description of a decision tree from~\ref{sub:formal alg}.)
\begin{itemize} 
\item Suppose that $u$ is a read node in $P$ with index $i_{u}$, and $v$ is $u$'s
child and $v\in P$. If $v$ is a left child of $u$, then $C_{P}$
contains $\neg x_{i_{u}}$. Otherwise, if $v$ is a right child of
$u$, then $C_{P}$ contains $x_{i_{u}}$.
\end{itemize} 
Finally, we associate a {\em rank} $\tau(C_P) = r(v')$  with this clause that is equal to the rank $r(v')$ of the leaf-node $v'$ in the path $P$. We set $F=\bigvee_{T\in\T}F_{T}$.  The total order $\prec$ on the set of clauses is defined in such a way which ensures that for any two clauses $C_P$ and $C_{P'}$, we have $C_P \prec C_{P'}$ iff $\tau(C_P) \geq \tau(C_{P'})$.

Let $\X=\{x_{1},\dots,x_{n}\}$ be the variables corresponding to each
bit of $\mem$. That is, let $\phi$ be such that $\phi(x_{i})=\mem[i]$.
Given the update to $\mem$ in the problem $\fDT$, we update $\phi$ accordingly. Suppose
that $C_P$ is the first clause among all clauses in $F$ (according to the total order $\prec$) where $C_{P}(\phi)=1$. Furthermore, suppose that the path $P$ corresponds to the decision tree $T$. We can find the path $P$ by calling the oracle for $\fDNF$. Clearly, $P$ is the  {\em execution path} of  $T$ when it operates on $\mem$, and  $v^*_T$ is the leaf-node of the path $P$. Now, the total order $\prec$ is defined in such a way which ensures that $T$
is the  decision tree that maximizes $r(v^*_T)$. Thus, we derive that $T$ is the current output for the problem $\fDT$.
\end{proof}

\begin{lem}
\label{lem:first to exist}$\fDNF$ is $\dP$-reducible to $\dDNF$.\end{lem}
\begin{proof}
Let $(F,\phi)$ be an instance of $\fDNF$ where $F$
has $n$ variables $\X=\{x_{1},\dots,x_{n}\}$ and $m$ clauses $\mathcal{C} = \{C_1, \ldots, C_m\}$. W.l.o.g., we assume that the total order $\prec$ defined over $\mathcal{C}$ is such that $C_1 \prec \cdots \prec C_m$. In other words, after each update we have to return the {\em first satisfied clause} (the one with the minimum index), provided such a clause exists. We
construct an instance of $\dDNF$ $(F',\phi')$ as follows. Let $F'$
have $n+2\log m$ variables $\X'=(x_{1},\dots,x_{n},s_{1}^{0},\dots,s_{\log m}^{0},s_{1}^{1},\dots,s_{\log m}^{1})$.
We call the variables $s_{i}^{b}$ the \emph{search variables}. They
are for ``searching for the first satisfied clause''. 

$F'$ has $m$ clauses. For each $j\in[m]$, we write the binary expansion
of $j=j_{1}j_{2}\dots j_{\log m}$. For each clause $C_{j}$ in $F$,
we construct 
\[
C'_{j}=C_{j}\wedge\bigwedge_{i=1}^{\log m}s_{i}^{j_{i}}.
\]

Let $\phi'$ be such that $\phi'(x_{i})=\phi(x_{i})$ for all $i\in[m]$
and $\phi'(s_{i}^{b})=1$ for all $i\in[\log m]$ and $b\in\{0,1\}$.
Given an update of $\phi$, we update $\phi'$ accordingly so that
$\phi'$ and $\phi$ agree on $\X$. As we set all search variables
to $1$, at this point we have $F(\phi)=F'(\phi')$. So if $F(\phi)=0$, then we notice
this by looking at $F'(\phi')$. But if $F(\phi)=1$, then we need to find
the first index $j$ where $C_{j}(\phi)=1$. To do this, we apply the following
binary search trick.

\medskip
\noindent
We repeat the following steps from $i=1$ until $\log m$. 
\begin{enumerate}
\item Set $\phi'(s_{i}^{1})=0$ (i.e. all $C_{j}$'s where $j_{1}=1$ are
``killed''). 
\item If $F'(\phi')=0$, then there is no $j$ where $C_{j}(\phi)=1$ and
$j_{1}=0$. Set $\phi'(s_{i}^{0})=0$ and $\phi'(s_{i}^{1})=1$ (i.e.
all $C_{j}$'s where $j_{1}=0$ are ``killed''). 
\item Else, $F'(\phi')=1$, then there is some $j$ where $C_{j}(\phi)=1$ and
$j_{1}=0$. Set $\phi'(s_{i}^{0})=1$ and $\phi'(s_{i}^{1})=0$ (i.e.
all $C_{j}$'s where $j_{1}=1$ are ``killed'').
\end{enumerate}
As in the above binary search, we always have a ``preference'' for
the ``first half''. So we will obtain the first $j$ where $C_{j}(\phi)=1$.
\end{proof}

\subsection{Reformulations of $\protect\dDNF$ }

\label{sub:DNF:equiv}

In this section, we define three problems ($\dOV$, $\dIndep$ and
$\dAW$) which are just different formulations of the same problem
$\dDNF$. However, these different views are useful for showing reductions
between $\dDNF$ and other problems. 

We define a problem called the \emph{dynamic all-white problem ($\dAW$):}
\begin{defn}
In the dynamic all-white problem ($\dAW$), we are first given 
\begin{itemize}
\item a bipartite graph $G=(L,R,E)$ where $|L|=n$, $|R|=m$ and $\deg(u)=\polylog(m)$
for each $u\in R$, and 
\item each node $u\in L$ is colored black or white
\end{itemize}

Then, the color of each node in $L$ can be updated. After each update,
we ask there is a node in $R$ whose neighbors are all white.

\end{defn}
Next, for a matrix $V=(v_{1},\dots,v_{m})\in\{0,1\}^{n\times m}$,
$v_{j}$ denote the $j$-th column of $V$ and $nnz(v_{j})$ be the
number of non-zero in $v_{j}$. The following problem is called the
\emph{dynamic sparse orthogonal vector problem ($\dOV$)}:
\begin{defn}
\label{def:OV}In the dynamic sparse orthogonal vector problem (\emph{$\dOV$}),
we are first given 
\begin{itemize}
\item a matrix $V\in\{0,1\}^{n\times m}$ where $nnz(v_{j})=\polylog(m)$
for each $j\in[m]$, and 
\item a vector $u\in\{0,1\}^{n}$
\end{itemize}

Then, each entry of $u$ updated. After each update, we ask if there
is a column $v$ in $V$ orthogonal to $u$, i.e. $u^{T}v=0$.

\end{defn}
Next, given a hypergraph $H=(V,E)$. We say that a set $S$ in independent
in $H$ if there is no $e\in E$ such that $e\subseteq S$. We define
a problem called the \emph{independent set query problem ($\dIndep$).}
\begin{defn}
In the \emph{independent set query problem }($\dIndep$), we are first
given 
\begin{itemize}
\item a hypergraph $H=(V,E)$ where $|V|=n$, $|E|=m$ and $|e|=\polylog(m)$
for each edge $e\in E$, and 
\item a set of node $S\subseteq V$
\end{itemize}

Then, the set $S$ can be updated by inserting or deleting a node
to/from $S$. After each update, we ask if $S$ is independent in
$H$.

\end{defn}
All three problems are the same problem with different representation,
so it holds that:
\begin{prop}
\label{prop:DNF equiv}An algorithm with the update and query time
are at most $O(u(m))$ for any one of the following problems implies
algorithms with the same update and query time for all other problems:
\begin{enumerate}
\item $\dDNF$ on an formula $F$ with $m$ clauses,
\item $\dAW$ on an graph $G=(L,R,E)$ where $|R|=m$ clauses,
\item $\dIndep$ on a hypergraph $H$ with $m$ edges, and
\item $\dOV$ on a matrix $V\in\{0,1\}^{n\times m}$.
\end{enumerate}
\end{prop}
As the proof is very straightforward, we defer it to \ref{sec:DNF equiv proof}.

\subsection{$\protect\dNP$-hardness of Some Known Dynamic Problems}

\label{sub:list NP hard}

In \cite{AbboudW14}, Abboud and Williams show SETH-hardness for all
of the problems in \ref{table:NP hard}. In their reduction, they
actually show a $\dP$-reduction from $\dOV$ to these problems.\footnote{For 3 vs. 4 diameter and $ST$-reach
		problems, they actually show a stronger reduction from a dynamic version
		of the $3$-OV problem, and not $\dOV$. But $\dOV$ is trivially
		reducible to this dynamic version of 3-OV.} Therefore, we immediately obtain the following:
\begin{cor}
\label{cor:list rankNP hard}All problems in \ref{table:NP hard}
are $\dNP$-hard.
\end{cor}

\newpage{}

\newpage{}

\part{Further Complexity Classes}\label{part:further_classes}
In this part, we define some more complexity classes including randomized classes and classes of search problems. 
This formalization is needed for arguing about the complexity of connectivity problems in \Cref{part:connectivity}.

\section{Randomized Classes}\label{sec:rand_classes}

In this section, we define the randomized version of the classes $\dP$
and $\dNP$ which are $\dBPP$ and $\dMA$ respectively.
The only difference between algorithms for problems in $\dP$ and
$\dBPP$ is that, for $\dBPP$, algorithms can be randomized, and
we only require the answers of the algorithms to be correct with probability
at least $1-1/n$ when the instance is of size $n$\footnote{Being correct with probability $\ge2/3$ is also enough. By maintaining
multiple data structures and then using the majority vote for each step, we can boost the probability of being correct arbitrarily.}. The same analogy goes for the difference between $\dNP$ and $\dMA$.

First, we formally define \emph{randomized algorithm-families}. A
randomized algorithm-family $\A$ is just a algorithm-family such
that, for each $n\ge1$, at step $t=0$, $\A_{n}$ is additionally
given a random string $r_{0}\in\{0,1\}^{\poly(n)}$,
and at step $t\ge1$, $\A_{n}$ is additionally given a random string
$r_{t}\in\{0,1\}^{\polylog(n)}$. The internal
states and the answers of $\A_{n}$ at each step can depend on previous
random strings. Hence, at the step $t$, the answer $x_{t}$ is a
random variable. We can formally define randomized verifier-families
as a randomized counterpart of deterministic verifier-family defined
in \Cref{def:verifier} exactly the same way.
\begin{defn}
[Class $\dBPP$]\label{def:BPP}A decision problem $\D$ is in $\dBPP$ iff it admits
a randomized algorithm-family $\A$ with update-time $\time_{\A}(n)=O(\polylog(n))$
and space-complexity $\text{Space}_{\A}(n)=O(\poly(n))$. On
an instance of size $n$, for each step $t$, the answer $x_{t}$
of $\A_{n}$ must be correct with probability at least $1-1/n$, i.e.
$\Pr[x_{t}=\D(I_{t})]\ge1-1/n$ for each $t$.
\end{defn}

\begin{defn}
[Class $\dMA$]A decision problem $\D$ is in $\dMA$ iff it admits
a randomized verifier-family $\V$ with update-time $\time_{\A}(n)=O(\polylog(n))$
and space-complexity $\text{Space}_{\A}(n)=O(\poly(n))$ which
satisfy the following properties for each $n\geq1$. Fix any instance-sequence
$(I_{0},\ldots,I_{k})$ of $\D_{n}$. Suppose that $\V_{n}$ gets
$I_{0}$ as input at step $t=0$, and $((I_{t-1},I_{t}),\pi_{t})$
as input at every step $t\geq1$. Then: 
\begin{enumerate}
\item For every proof-sequence $(\pi_{1},\ldots,\pi_{k})$, we have $x_{t}=0$
with probability at least $1-1/n$ for each $t\in\{0,\ldots,k\}$
where $\D_{n}(I_{t})=0$. 
\item If the proof-sequence $(\pi_{1},\ldots,\pi_{k})$ is \emph{reward-maximizing}
(defined as in \Cref{def:NP} and repeated below), then we have $x_{t}=1$ with probability at least
$1-1/n$ for each $t\in\{0,\ldots,k\}$ with $\D_{n}(I_{t})=1$.
\end{enumerate}
The proof-sequence $(\pi_1, \ldots, \pi_k)$ is {\em reward-maximizing} iff at each step $t \geq 1$, given the past history $(I_0, \ldots, I_t)$, $(r_0, \ldots, r_t)$, and $(\pi_1, \ldots, \pi_{t-1})$, the proof $\pi_t$ is chosen in such a way that maximizes  $y_t$ (when we think of $y_t$ as a $\polylog (n)$ bit integer). We say that such a proof $\pi_t$ is {\em reward-maximizing}.
\end{defn}
Note that the above class corresponds to the class of AM (Arthur-Merlin) in the static setting and not MA (Merlin-Arthur) because the proof $\pi_t$ at step $t$ can depend on all the random choices $r_0$,\dots,$r_t$ up to the current step $t$. We emphasize that it is the \emph{prover} who ``sees'' the previous random choices. The \emph{adversary} only sees the answers of the algorithm.

By definition of $\dMA$, we have the following:
\begin{prop}
\label{lem:NP(BPP) in MA}$(\dNP)^{\dBPP}\subseteq\dMA$.%
\end{prop}
Next, we list a randomized counterpart of \Cref{lem:NP:coNP,lem:NP:coNP:next,th:PH:collapse:next}. 
The proofs go exactly the same and hence are omitted.
\begin{prop}
\label{lem:removing intersection randomized oracle}$(\dMA\cap\dcoMA)^{\dMA\cap\dcoMA}=\dMA\cap\dcoMA$
and $(\dMA)^{\dMA\cap\dcoMA}=\dMA$.
\end{prop}

\begin{prop}
\label{lem:if NP in coMA}If $\dNP\subseteq\dcoMA$, then $\dPH\subseteq\dMA\cap\dcoMA$.
\end{prop}

\section{Search Problems}

\label{sec:search}

In this section, we define complexity classes of \emph{dynamic search
problems}. Intuitively, dynamic search problems are problems of maintaining
some objects under updates. For example, in graphs, there are problems
of maintaining spanning forests, matchings, shortest-paths trees,
etc. The important characteristic of dynamic search problems is they
the size of maintained object are usually larger than the update time.
Hence, at each step, the algorithm only outputs how the object should
be changed. 
Search problems are very natural in the literature of dynamic algorithms
and, perhaps, even more well-studied than their decision version.

\paragraph{Dynamic search problems and algorithms. }

A dynamic search problem $\D$ is a dynamic decision problem except
that the \emph{answer} on the instance $I$ denoted by $\D(I)$ may not be
only $0$ or $1$. For each $I$, $\D(I)$ can be an arbitrary set.
For any $x\in\D(I)$, we say that $x$ is a \emph{correct} answer for $I$.
An algorithm $\A$ for the search problem $\D$ is defined as an
algorithm for decision problems except for one difference. The difference
is that $\A$ will designate a part of its memory for the \emph{answer
	for the search problem}. At step $t$, we denote such answer by $x_{t}$
and $\A$ can update by reading and writing on $x_{t}$ at each step.
Note that $x_{t}$ can be a large string. We say that $\A$ solves
$\D$ if, every step $t$, $x_{t}\in\D(I_{t})$ for any instance sequence
$(I_{0},I_{1},\dots)$. We define a verifier and a randomized
verifier for search problems in similar way.

\paragraph{Complexity Classes.}

The search version of $\dP$, $\dNP$ and $\dBPP$ are denoted by
$\dFP$, $\dFNP$ and $\dFBPP$ respectively. Another important class
of search problems is the class $\dTFNP$. We can intuitively think
of this class as a search version of $\dNP\cap\dcoNP$. The definitions
are motivated from the definitions in the static setting (see e.g.
\cite{MegiddoP91,GoldwasserGH18} and \cite{Rich2008automata}).
\begin{defn}
[Class $\dFP$]A search problem $\D$ is in $\dFP$ iff it admits
a algorithm-family $\A$ with update-time $\time_{\A}(n)=O(\polylog(n))$
and space-complexity $\text{Space}_{\A}(n)=O(\poly(n))$. On
an instance of size $n$, for each step $t$, the answer $x_{t}$
of $\A_{n}$ must be correct, i.e.
$x_{t}\in\D(I_{t})$.
\end{defn}

\begin{defn}
[Class $\dFBPP$]A search problem $\D$ is in $\dFBPP$ iff it admits
a randomized algorithm-family $\A$ with update-time $\time_{\A}(n)=O(\polylog(n))$
and space-complexity $\text{Space}_{\A}(n)=O(\poly(n))$. On
an instance of size $n$, for each step $t$, the answer $x_{t}$
of $\A_{n}$ must be correct with probability at least $1-1/n$, i.e.
$\Pr[x_{t}\in\D(I_{t})]\ge1-1/n$ for each $t$.\end{defn}
\begin{rem}
[Oblivious vs. adaptive adversary]\label{rem:adversary}For search
problems that can be solved with randomized algorithms, there is an
important distinction between the two models of adversaries. We say
that an adversary is oblivious if the updates it generates must \emph{not}
depend on the previous outputs of the algorithms. Otherwise, we say
that an adversary is \emph{adaptive}. Observe that, in all definitions
in this paper, we never assume oblivious adversaries and so algorithms
must work against with adaptive adversaries.\end{rem}
\begin{defn}[Class $\dFNP$]\label{def:dFNP}
A search problem $\D$ is in $\dFNP$ iff it admits
a verifier-family $\V$ with update-time $\time_{\A}(n)=O(\polylog(n))$
and space-complexity $\text{Space}_{\A}(n)=O(\poly(n))$ which
satisfy the following properties for each $n\geq1$. Fix any instance-sequence
$(I_{0},\ldots,I_{k})$ of $\D_{n}$. Suppose that $\V_{n}$ gets
$I_{0}$ as input at step $t=0$, and $((I_{t-1},I_{t}),\pi_{t})$
as input at every step $t\geq1$. 
At each step $t \in\{0,\ldots,k\}$, $\V_{n}$ maintains the answer $x_t$ or outputs $\bot$ where $\bot$ is a special symbol.
\begin{itemize}
    \item If $\D_{n}(I_{t}) = \emptyset$, then $\V_{n}$ 
outputs $\bot$.  
    \item If $\D_{n}(I_{t}) \neq \emptyset$ and the proof-sequence $(\pi_{1},\ldots,\pi_{k})$ 
is \emph{reward-maximizing} (defined as in \Cref{def:NP}), then $x_{t}\in\D_{n}(I_{t})$.
    \item If $\D_{n}(I_{t}) \neq \emptyset$ and the proof-sequence $(\pi_{1},\ldots,\pi_{k})$ 
is not \emph{reward-maximizing}, then $x_t \in\D_{n}(I_{t})$ or $\bot$ is outputted.\qedhere
\end{itemize}
\end{defn}
In other words, if $V_n$ does not output $\bot$, then $x_t$ is always a correct answer. $V_n$ is allowed to output $\bot$ when there is no answer or when the proof sequence is not reward-maximizing.

Note that it make sense to say that a search problem $\D$ is $\dNP$-hard
under $\dP$-reduction. This just means that $\dNP\subseteq(\dP)^{\D}$,
i.e. any decision problem in $\dNP$ can be solved efficiently given
an oracle to $\D$. 
\begin{example}
Recall the decision problem $\dDNF$ from \Cref{def:range DNF}.  Define a {\em search} version of the same problem, $\sDNF$, where for each instance $I = (F, \phi)$, the set of feasible solutions $\mathcal{D}(I)$ is given by the collection of clauses $C_i \in \mathcal{C}$ that are satisfied under the assignment $\phi$. Since $\dDNF$ is $\dNP$-complete (see \Cref{thm:DNF complete}), it immediately follows that $\sDNF$ is $\dNP$-hard. It is easy to check that $\sDNF$ is in $\dFNP$.\footnote{We point out that the $\fDNF$ problem, as defined in \Cref{sub:DNF:NPhard}, is {\em not} known to be in $\dFNP$. This is because it is not clear how to design a verifier-family for $\fDNF$ which satisfies the last condition of \Cref{def:dFNP}.}
\end{example}

We say that a search problem $\D$ is\emph{ total }if for any instance
$I$, $\D(I)\neq\emptyset$. 
\begin{defn}
[Class $\dTFNP$]A search problem $\D$ is in $\dTFNP$ iff $\D\in\dFNP$
and $\D$ is total.
\end{defn}

In other words, every problem in $\dTFNP$ admits a verifier  that, given a proof sequence, maintains a correct answer  $x_t \in\D_{n}(I_{t})$ at all steps $t$, otherwise returns $\bot$, indicating the proof sequence is not reward-maximizing.

The class $\dTFNP$ is a search version of $\dNP\cap\dcoNP$ for the
same reason as in the static setting (see \cite{MegiddoP91}). 
\begin{example}
Dynamic spanning tree is not total, because there is no spanning tree
if a graph is not connected. However, dynamic spanning forest is total.
Dynamic spanning \emph{tree} is in $\dFNP$ using the same algorithm, which shows that dynamic connectivity $\conn \in \dNP$ from \Cref{prop:conn in NP}.  

It is not known if dynamic
spanning forest is in $\dTFNP$. 
Note that the verifier $\V_{NP}$ in the proof that $\conn\in\dNP$ from \Cref{prop:conn in NP} is not strong enough to show that dynamic spanning forest is in $\dTFNP$. This is because, although the $\V_{NP}$ does maintain a spanning forest whenever the proof sequence is reward-maximizing, when the proof is not reward-maximizing, $\V_{NP}$ does not guarantee to still correctly maintain a spanning forest or return $\bot$. (It might maintain a non-spanning forest.)
We will show in \ref{sec:is conn in coNP}
that dynamic spanning forest is in $\dTFNP$ iff dynamic connectivity
is in $\dcoNP$.
\end{example}
The proof of the following observation is basically the same as \Cref{lem:NP:coNP,lem:NP:coNP:next}.
We just have an oracle for search problem.
\begin{prop}
\label{prop:P to TNFP }$(\dP)^{\dTFNP}\subseteq\dNP\cap\dcoNP$. 
\end{prop}

\begin{prop}
\label{prop:NP to TNFP }$(\dNP)^{\dTFNP}=\dNP$ .\end{prop}

\newpage{}

\part{A Coarse-grained Approach to Connectivity Problems.}
\label{part:connectivity}

In contrast to the \emph{fine-grained} complexity which is problem-centric,
the \emph{coarse-grained} approach in this paper is resource-centric.
The goal of this part, however, is to show that this approach is helpful
for understanding the complexity of specific problems as well (even
though these problems are not known to be complete for any class).

As a show-case, we give new observations to the well-studied dynamic
connectivity and $k$-edge connectivity problems. More specifically,
we show a new ``nondeterministic equivalence'' between dynamic connectivity
and dynamic spanning forest and its consequences. We show that non-determinism
together with randomization can help speeding up the best current dynamic
$k$-edge connectivity algorithms from $\tilde{O}(\sqrt{n})$ to polylogarithmic
time. Then, by applying our results on coarse-grained complexity theory,
this implies a certain fine-grained ``non-hardness'' result of dynamic
$k$-edge connectivity and other problems.

\section{Is dynamic connectivity in $\protect\dcoNP$?}

\label{sec:is conn in coNP}

Dynamic connectivity ($\conn$) is one of the most well-studied dynamic
graph problems. In this section, we study this problem from non-deterministic
perspective.

\paragraph{Previous works.}

Here, we briefly review the best update time of this problem in the
worst-case setting. Kapron, King and Mountjoy \cite{KapronKM13} show
a Monte Carlo algorithm with polylogarithmic worst-case time for $\conn$.
In other words, they show the following (see \ref{prop:conn in NP}
for the formal proof):
\begin{thm}
[\cite{KapronKM13}]\label{lem:conn in BPP main}$\conn\in\dBPP$.
\end{thm}
Nanongkai, Saranurak and Wulff-Nilsen \cite{NanongkaiSW17} show a
Las Vegas algorithm for dynamic minimum spanning forest (which implies
dynamic connectivity). On an $n$-node graph, their algorithm has
$n^{o(\log\log\log n/\log\log n)}$ worst-case update time, which
is later slightly improved to $n^{O(\log\log n/\sqrt{\log n}})$ \cite{SaranurakW19}.
However, sub-polynomial update time are not captured by our complexity
classes. For deterministic algorithm, the best update time is still
$O(\sqrt{n}\cdot\frac{\log\log n}{\sqrt{\log n}})$ \cite{Kejlberg-Rasmussen16},
slightly improving the long-standing $O(\sqrt{n})$ bound by \cite{Frederickson85,EppsteinGIN92}.
The ultimate goal in this line of research is to show that there is
deterministic algorithm with $\mbox{polylog}(n)$ worst-case update
time, as in the case when amortized update time is allowed \cite{HolmLT98}.
In our language, this question translates to the following:
\begin{question}
Is $\conn\in\dP$? 
\end{question}
When we allow non-determinism, it is easy to show that $\conn\in\dNP$
(see \ref{prop:conn in NP}). 
\begin{prop}
\label{prop:conn in NP main}$\conn\in\dNP$.
\end{prop}
Surprisingly, it is not clear at all whether $\conn\in\dcoNP$. As
a stepping stone towards the answer whether $\conn\in\dP$, we believe
that studying this question might lead to further insight.
\begin{question}
Is $\conn\in\dcoNP$?
\end{question}
For the rest of this section, we study some consequences if $\conn\in\dcoNP$.
These consequences will later help us argue why some problems should
not be $\dNP$-hard in a non-straightforward way in \ref{sec:kconn}.

\subsection{Nondeterministic Equivalence of Connectivity and Spanning Forest}

Dynamic connectivity ($\conn$) and dynamic spanning forest ($\spanningForest$)
are closely related problems. If we can maintain a spanning
forest $F$ of a graph $G$, then we can implement the top tree \cite{AlstrupHLT05} on $F$ to count the number of
connected components; consequently, we know whether $G$ is connected or not. Moreover,
we can answer a connectivity query for any pair of nodes in logarithmic
time. Conversely, given that we can maintain if $G$ is connected,
it is not clear how to maintain a spanning forest of $G$. Interestingly,
all the previous dynamic connectivity algorithms actually maintain
dynamic spanning forest\footnote{There is one exception in a more restricted sensitivity setting \cite{PatrascuT07}.}.
It turns out that, when non-determinism is allowed, the two problems
are indeed equivalent :
\begin{thm}
$\conn\in\dNP\cap\dcoNP$ iff $\spanningForest\in\dTFNP$.\label{thm:conn spanning equiv}\end{thm}
\begin{proof}
($\impliedby$) Suppose that $\spanningForest\in\dTFNP$; i.e. there is a verifier-family $\V$ as in \Cref{def:dFNP}. Below we construct a verifier-family $\V'$ to show that $\conn\in \dcoNP$ (\Cref{def:coNP}). This implies $\conn\in\dNP\cap\dcoNP$ since we know that $\conn\in\dNP$ (\Cref{prop:conn in NP}).

Our verifier $\V'$ uses the same proof and reward as $\V$. If the graph maintained by $\V$ is a spanning tree\footnote{We can use the top tree to maintain this.} 
or $\V$ returns $\bot$, then $\V'$ returns 1 (to indicate that the input dynamic graph is connected); otherwise, $\V'$ returns 0 (``not connected''). 
Now we analyze the correctness of $\V'$.
\begin{itemize}
    \item If the proof sequence is reward-maximizing, then the subgraph maintained by $\V$ must be a spanning forest. Our verifier $\V'$ always returns the correct answer in this case. In particular, 
    the second condition in  \Cref{def:coNP} is always satisfied. 
    \item As noted under \Cref{def:dFNP}, if $\V$ does {\em not} return $\bot$ then the subgraph  maintained by $\V$ must be a spanning forest. In particular, for any proof sequence, when the dynamic input graph $G$ is connected, either $\V$ returns $\bot$ or maintains a spanning tree of $G$. $\V'$ returns 1 in both cases.  The first condition in  \Cref{def:coNP} is thus satisfied. 
\end{itemize}


($\implies$)
Suppose that $\conn\in\dNP\cap\dcoNP$. 
We construct an algorithm that given a proof sequence maintains a spanning forest $F$ of an input graph $G$ satisfying conditions in \Cref{def:dFNP} as follows. 
In our algorithm, if we ever return $\bot$ we will always return $\bot$ in the subsequent step. Consider step $t$ where we have not yet returned $\bot$. 
Suppose that
we have maintained a graph $G_{t}$ and its spanning forest $F_{t}$
up to time step $t$. Suppose that $F_{t}=\{T^{1},\dots,T^{k_{t}}\}$
has $k_{t}$ connected components. For $1\le i\le k_{t}$, let $u^{i}\in T^{k_{i}}$
be arbitrary node in $T^{k_{i}}$. We will maintain a graph $G'_{t}$
where $V(G'_{t})=V(G_{t})\cup\{s\}$ and $E(G'_{t})=E(G{}_{t})\cup\{(s,u^{i})\mid1\le i\le k_{t}\}$. 
Intuitively, $G'_t$ is obtained from $G_t$ by connecting each connected component represented 
by $u^i$ in $G$ to the node $s$ in $G'_t$. So $G'_t$ is connected at every time.

Now, given an edge insertion, it is easy to maintain $G_{t+1}$, $F_{t+1}$
and $G'_{t+1}$ using link-cut tree \cite{SleatorT81}. 
For edge deletion,
it is also easy if the deleted edge $e$ is not in $F_{t}$. 
So it
remains to deal with the case where $e\in F_{t}$. Suppose that 
%
deleting
$e$ disconnects $T^{j}$ into two trees $T_{L}^{j}$ and $T_{R}^{j}$.
Our goal is to find a \emph{replacement edge} $f=(u_{L},u_{R})$ where
$f\neq e$ and $u_{L}\in T_{L}^{j}$ and $u_{R}\in T_{R}^{j}$. 


We use the following observation:
\begin{claim}
$G'_{t}-e$ is connected iff there is a replacement edge.
\end{claim}

Informally, we can call a $\dNP\cap\dcoNP$-algorithm $\A$ for dynamic
connectivity on $G'_{t}$. If $\A$ returns that $G'_{t}-e$ is not
connected, then we do not need to update $G_{t+1}=G_{t}-e$ and $F_{t+1}=F_{t}-e$.
Either $T_{L}^{j}$ or $T_{R}^{j}$ is a new connected component,
and so we can update $G'_{t+1}$ accordingly. If $\A$ returns that
$G'_{t}-e$ is connected, then there \emph{is }a replacement edge
$f=(u_{L},u_{R})$ which a prover can provide us, and, more importantly,
we can check if $f\in E_{t}$, $f\neq e$, $u_{L}\in T_{L}^{j}$ and
$u_{R}\in T_{R}^{j}$, i.e. $f$ is indeed a replacement edge. Then,
we can update $G_{t+1}$, $F_{t+1}$ and $G'_{t+1}$ accordingly. This allows us to verify quickly that at each step $t$, $F_{t}$
is indeed a spanning forest.

More formally, we simulate some algorithms for verifying that $G'_{t}$ is connected or not as follows. 
%
Let $\V_{NP}$ be the verifier as in the proof of \Cref{prop:conn in NP} where we show that $\conn\in \dNP$; recall that $\V_{NP}$ maintains a forest $F'_t$ of the input dynamic graph ($G'_{t}$ in this case) that is a spanning forest when the proof sequence is reward maximizing.
%
Let $\V_{coNP}$ be the verifier satisfying \Cref{def:coNP}. ($\conn \in \dcoNP$ implies that $\V_{coNP}$ exists.) 
At each step $t$, $\V_{NP}$ and $\V_{coNP}$ are expected to receive some updates to $G_t'$ and proofs $\pi_t^{NP}$  and $\pi_t^{coNP}$. 
They then return whether $H$ is connected or not and the rewards, say $r_t^{NP}$  and $r_t^{coNP}$. To simulate them, we take $(\pi_t^{NP},\pi_t^{coNP})$ as our own proof and give $\pi_t^{NP}$ and $\pi_t^{coNP}$ to $\V_{NP}$ and $\V_{coNP}$ respectively. We also update $G_t'$ in the same way we update $G_t$ (i.e. we add/delete edges in $G'_t$ as we add/delete edges in $G'_t$). 
The reward of our algorithm is $r_t^{NP}+r_t^{coNP}$. Observe that the proof $(\pi_t^{NP},\pi_t^{coNP})$ is a reward-maximizing proof for our algorithm if and only if $\pi_t^{NP}$ and $\pi_t^{coNP}$  are reward-maximizing proofs for $\V_{NP}$ and $\V_{coNP}$ 
 respectively.

 Now, consider when an edge $e\in F'_t$ is deleted (other cases can be easily handled as described earlier). (i) If  $\V_{NP}$  returns 1 (indicating that $G'_{t+1}$ is still connected after the edge deletion), then 
 $F'_{t+1}$ is a spanning forest of $G'_{t+1}$ (this follows from the proof of \Cref{prop:conn in NP});  
 we can argue by induction that  $F_{t+1}=F'_{t+1}\cap G_{t+1}$ is a spanning forest of $G_t$. (We can easily maintain such an $F_{t+1}$.)
(ii) If  $\V_{coNP}$  returns 0, we know that $G'_t$ is not connected. Thus we add an edge from $s$ to one of the end-vertices of $e$ to keep $G'_t$ connected and call $\V_{NP}$ and $\V_{coNP}$  with appropriate updates and proofs (e.g., the update and the proof for $\V_{NP}$ according to the proof of \Cref{prop:conn in NP} is simply the edge we just insert).  
(iii) If $\V_{NP}$ returns 0 and  $\V_{coNP}$  returns 1, we can conclude that one of the proof sequences for $\V_{NP}$ and $\V_{coNP}$  are {\em not} reward maximizing (otherwise, $\V_{NP}$ and $\V_{coNP}$ must return correct answers); thus, we return $\bot$ from now on. 
(Note that it is impossible that $\V_{NP}$ returns 1 and  $\V_{coNP}$  returns 0 since the former case happens only if $G'_t$ is connected and the latter happens only if $G'_t$ is not connected).
\end{proof}

\subsection{Consequences of $\protect\conn\in\protect\dcoNP$}
\begin{cor}
[\cite{HenzingerK99}]\label{cor:consequence conn in coNP}Suppose
that $\conn\in\dcoNP$. Dynamic bipartiteness is $\dNP\cap\dcoNP$.
Moreover, the following problems are in $\dTFNP$:
\begin{enumerate}
\item dynamic minimum spanning forest on $d$ distinct-weight graphs where
$d=\polylog(n)$
\item dynamic $(1+\epsilon$)-approximate minimum spanning forest where
$\epsilon>1/\polylog(n)$, and
\item dynamic $k$-edge connectivity certificate where $k=\polylog(n)$\footnote{A $k$-connectivity certificate $H$ of $G$ is a subgraph of $G$
where 1) $H$ has $O(kn)$ edges, and 2) $H$ is $k$-edge connected
iff $G$ is $k$-edge connected. }.
\end{enumerate}
The number $n$ above is the number of nodes in the underlying graphs. \end{cor}
\begin{proof}
As we already know that $\conn\in\dNP$, if $\conn\in\dcoNP$, then
$\spanningForest\in\dTFNP$ by \ref{thm:conn spanning equiv}. Henzinger
and King \cite{HenzingerK99} show several reduction from spanning
forest to the problems  above. 
\end{proof}
From this, we can infer the ``non-hardness'' of the above problems
in \ref{cor:consequence conn in coNP} unless $\dPH$ collapses:
\begin{cor}
\label{cor:conn easy or something easy}Unless $\dPH=\dNP\cap\dcoNP$,
the following two statements cannot hold simultaneously
\begin{enumerate}
\item $\conn\in\dcoNP$.
\item Some problem $\D$ in the list of \ref{cor:consequence conn in coNP}
is $\dNP$-hard.
\end{enumerate}
\end{cor}
\begin{proof}
Suppose that $\conn\in\dcoNP$. Then, dynamic bipartiteness is in
$\dNP\cap\dcoNP$, and all the search problems $\D$ from \ref{cor:consequence conn in coNP}
are in $\dTFNP$. If dynamic bipartiteness is $\dNP$-hard, then $\dNP\subseteq(\dP)^{\dNP\cap\dcoNP}=\dNP\cap\dcoNP$
by \ref{lem:NP:coNP:next}. If any problem $\D$ from \ref{cor:consequence conn in coNP} is $\dNP$-hard,
then $\dNP\subseteq(\dP)^{\dTFNP}\subseteq\dNP\cap\dcoNP$ by \ref{prop:P to TNFP }.
Both of these imply that $\dPH$ collapses to $\dNP\cap\dcoNP$ by
\ref{th:PH:collapse:next}.
\end{proof}

\section{Complexity of dynamic $k$-edge connecitivty}

\label{sec:kconn}

Recall that the dynamic $k$-edge connectivity problem is to maintain
whether a graph is $k$-edge connected. The current best worst-case
upper bound for this problem is $\tilde{O}(\sqrt{n})$ by a deterministic
algorithm \cite{Thorup01}. For $k\le3$, the problem has been extensively
studied in the amortized update time setting (see \cite{HolmLT98}
for the history). However, for $k\ge4$, there is no better algorithms
with amortized update time. It remains open whether the $\tilde{O}(\sqrt{n})$-bound
can be improved.

It turns out that non-determinism and randomization can significantly
speed up the update time to polylogarithmic worst-case. More precisely,
we show that the following:
\begin{thm}
\label{thm:kconn in MAcoMA}$\kconn\in\dMA\cap\dcoMA$.
\end{thm}
Moreover, randomization can be removed if $\conn\in\dcoNP$:
\begin{thm}
\label{thm:kconn in NPcoNP}If $\conn\in\dcoNP$, then $\kconn\in\dNP\cap\dcoNP$. 
\end{thm}

\subsection{Non-reducibility from $\protect\dDNF$ to $\protect\kconn$}

Before proving the above theorems, we discuss that this shows some
``evidences'' that dynamic $k$-edge connectivity should not be
$\dNP$-hard. This also implies some interesting consequences to the
fine-grained complexity of $\kconn$.

\ref{thm:kconn in NPcoNP} adds another problem to the list of \ref{cor:consequence conn in coNP}.
Using the same proof as \ref{cor:conn easy or something easy}, we
have
\begin{cor}
Unless $\dPH=\dNP\cap\dcoNP$, either $\conn\in\dcoNP$ or $\kconn$
is \emph{not} $\dNP$-hard.
\end{cor}
The ``evidence'' that $\kconn$ should not be $\dNP$-hard is relatively
more interesting than other problems in \ref{cor:consequence conn in coNP}.
This is because the problems in \ref{cor:consequence conn in coNP}
already admits very fast algorithms. Using the dynamic spanning forest algorithm 
against an oblivious adversary from \cite{KapronKM13}, there
are Monte Carlo algorithms against an oblivious adversary with polylogarithmic
worst-case update time. Using \cite{NanongkaiSW17}, there are Las
Vegas algorithms against adaptive adversaries with sub-polynomial
worst-case update update time for these problems. On the contrary,
the current best update time for $\kconn$ is still $\tilde{O}(\sqrt{n})$
by \cite{Thorup01}. 

We can also show a similar theorems without the condition whether
$\conn\in\dcoNP$:
\begin{cor}
Unless $\dPH\subseteq\dMA\cap\dcoMA$, then $\kconn$ is not $\dNP$-hard.\end{cor}
\begin{proof}
If $\kconn$ is $\dNP$-hard, then \ref{thm:kconn in MAcoMA} implies
that $\dNP\subseteq(\dP)^{\dMA\cap\dcoMA}=\dMA\cap\dcoMA$ by \ref{lem:removing intersection randomized oracle}.
This implies that $\dPH\subseteq\dMA\cap\dcoMA$ by \ref{lem:if NP in coMA}.
\end{proof}

\paragraph{Consequences to fine-grained complexity of $\protect\kconn$.}

Suppose that we believe that $\kconn$ is not $\dNP$-hard. By definition,
this means that there is no reduction from $\dDNF$ to (many
instances of) $\kconn$ \emph{with any polynomial-size blow up}. 
This is a useful information about fine-grained complexity of $\kconn$ because of the following reason.

Without this information, it is conceivable that there might exist
a reduction from
$\dDNF$ to $\kconn$ where the size of the instance is blown up by
a quadratic factor. That is, we can reduce $\dDNF$ with $m$ clauses
to $\kconn$ with $m^{2}$ nodes. By SETH, this would immediately imply
that there is a tight lower bound of $\Omega(n^{0.5-o(1)})$ for $\kconn$.
In this section, we show that such reduction is unlikely assuming that $\dPH\not\subseteq\dMA\cap\dcoMA$ for example.

\subsection{Proof outline}

First, we prove the easy part of the above theorems:
\begin{lem}
$\kconn\in(\dcoNP)^{\conn}$. In particular, $\kconn\in\dcoMA$ and
$\kconn\in\dPi_{2}$. Moreover, if $\conn\in\dcoNP$, then $\kconn\in\dcoNP$.\end{lem}
\begin{proof}
[Proof sketch]At any time step, if a graph $G=(V,E)$ is not $k$
connected, then a prover sends the cut set $C\subset E$ of size less
than $k$ to a verifier. Then, a verifier can try deleting edges from
$C$ and see if $G-C$ is connected using the dynamic connectivity
oracle. That is, we can verify quickly that $G$ is \emph{not }$k$-edge
connected. As $\conn\in\dBPP$, $\kconn\in\dcoMA$ by \ref{lem:NP(BPP) in MA}.
As $\conn\in\dNP$, $\kconn\in\dPi_{2}$ by definition. If $\conn\in\dcoNP$,
then $\conn\in\dNP\cap\dcoNP$ by \ref{prop:conn in NP main}. 
We we have $\kconn\in (\dcoNP)^{\dNP\cap\dcoNP}=\dcoNP$ by \ref{lem:NP:coNP}.
\end{proof}
It remains to prove the following:
\begin{lem}
\label{lem:kconn in MA}$\kconn\in\dMA$. Moreover, if $\conn\in\dcoNP$,
then $\kconn\in\dNP$.
\end{lem}
The proof is based on the previous algorithm by Thorup \cite{Thorup01}.
There are three parts. The first part is to review Thorup's algorithm.
The second part is to show that the above lemma follows if we can
show that another dynamic problem called $\kminTreeCut$ is in $\dNP$.
The last part is to show that the YES-instance of $\kminTreeCut$
can be verified quickly, i.e. $\kminTreeCut\in\dNP$. This is done
by adjusting the algorithm of \cite{Thorup01} and use non-determinism
to bypass the $\tilde{O}(\sqrt{n})$-time bottleneck in \cite{Thorup01}.

\subsection{Reviewing Thorup's reduction to $\protect\kminTreeCut$}

Before formally defining the problem $\kminTreeCut$, we need the
following definitions. Let $G=(V,E)$, $F\subseteq E$ and $e\in F$.
Suppose $T\ni e$ is the connected component of $F$. Write $T=T_{L}\cup\{e\}\cup T_{R}$.
The \emph{cover number} of $e$ is $cover_{(G,F)}(e)=|\{(u,v)\in E\mid u\in T_{L}$
and $v\in T_{R}\}|$.
\begin{defn}
[$\kminTreeCut$]In the dynamic problem called dynamic $k$ min tree
cut problem ($\kminTreeCut$), we have to maintain a data structure
on a graph $G$ and a forest $F$ that can handle the following operations:
\begin{enumerate}
\item Update edge insertion or deletions in $G$.
\item Update edge insertion or deletions in $F$ as long as $F\subseteq G$
and $F$ is a forest.
\end{enumerate}
After each update, we must return whether all edges $e\in F$ have
cover number at least $k$, i.e. $\min_{e\in F}cover_{(G,F)}(e)\ge k$. 
\end{defn}

Now, we review how Thorup reduces $\kconn$ to $\kminTreeCut$. The
main goal in \cite{Thorup01} is to show a dynamic randomized algorithm
for solving $(1+\epsilon)$-approximate mincut with $\tilde{O}(\sqrt{n})$
worst-case update time on a graph with $n$ nodes. To do this, Thorup
first shows a randomized reduction based on Karger's sampling \cite{Karger99}
from $(1+\epsilon)$-approximate mincut to $k$-edge connectivity
problem where $k=O(\log n)$. This is the only randomized part in
his algorithm. Then, he actually gives a dynamic deterministic algorithm
for $\kconn$ with $O(\poly(k\log n)\sqrt{n})$ worst-case update
time. 

To solve $\kconn$, he maintains the \emph{greedy tree packing }with
$d=\poly(k\log n)$ many forests. The greedy tree packing \emph{$\T=\{F_{1},\dots,F_{d}\}$}
on $G$ is a collections of forests in $G$, i.e. $F_{i}\subseteq G$.
If $G$ is connected, then each $F_{i}$ is a tree. We omit the precise
definition here. We only need to state the crucial property of the greedy tree packing $\T$
as follows:
\begin{thm}
[\cite{Thorup01}]For some $d=\poly(k\log n)$, let $\T$ be
a greedy tree packing on $G$ with $d$ forests. $G$ is $k$-edge
connected iff $G$ is connected and $\min_{e\in F}cover_{(G,F)}(e)\ge k$
for all $F\in\T$.\label{thm:kconn to many kMinTreeCut}
\end{thm}
That is, given a greedy tree packing, $\kconn$ is reduced to the
``AND'' of many instances of $\kminTreeCut$. Another property 
about maintaining a greedy tree packing is the following:
\begin{lem}
[\cite{ThorupK00,Thorup01}]For any $d=\polylog n$, dynamic
greedy tree packing with $d$ forests is $\dP$-reducible to dynamic
minimum spanning forests on graph with $d$-distinct edge weights.\label{lem:packing to MST}
\end{lem}
Recall that the above two problems in \ref{lem:packing to MST} are search problems. But the statement
makes sense by the discussion in \ref{sec:search}.

\subsection{Proof of \ref{lem:kconn in MA}}

We will show later in the next section that $\kminTreeCut\in\dNP$
in \ref{thm:kMinTreeCut}. Given this, our goal in this section is
to prove \ref{lem:kconn in MA} which implies the main results (\ref{thm:kconn in MAcoMA}
and \ref{thm:kconn in NPcoNP}).

\paragraph{Proof: if $\protect\conn\in\protect\dcoNP$, then $\protect\kconn\in\protect\dNP$.}

Here, we assume that $\conn\in\dcoNP$. By \ref{cor:consequence conn in coNP}
and \ref{lem:packing to MST}, there are $\dTFNP$-algorithms for
dynamic minimum spanning forests on graph with $d$-distinct edge
weights where $d=\mbox{polylog}(n)$, and for dynamic greedy tree
packing $\T$ with $d$ forests, respectively. 

From \ref{thm:kconn to many kMinTreeCut}, we need to verify if $\min_{e\in F}cover_{(G,F)}(e)\ge k$
for every $F\in\T$. Let us denote this problem (i.e. the ``AND''
of $d=\mbox{polylog}(n)$ many instances of $\kminTreeCut$) by ``$d$-AND-$\kminTreeCut$''.
The key observation is that, if $\kminTreeCut\in\dNP$ by \ref{thm:kMinTreeCut}
as will be shown in \ref{thm:kMinTreeCut}, then $d$-AND-$\kminTreeCut$
is also in $\dNP$. This is true simply by running $d$ instances
of the verifiers for $\kminTreeCut$.

\ref{thm:kconn to many kMinTreeCut} implies that, to solve $\kconn$,
it is enough to maintain a dynamic greedy tree packing $\T$ with
$d$ forests, feed the graph and $d$ forests from $\T$ as an input
to the subroutine for $d$-AND-$\kminTreeCut$, and then return the
\emph{same} YES/NO answer returned by the subrouine for $d$-AND-$\kminTreeCut$,
in every step. 

We claim that $\kconn\in\dNP$. 
To see this, first, since $\T$ can be maintained by a $\dTFNP$-algorithm, this means that, given a proof-sequence, we can always maintain  $\T$ correctly, otherwise we can correctly detect that the proof sequence is not reward-maximizing and return $\bot$. If we detect that once the proof sequence is not reward-maximizing, we can simply return $0$ as an answer for $\kconn$ from now.\thatchaphol{This is where I edit.}
Second, assuming that $\T$ is correctly maintained at every step, given a sequence of reward-maximizing proof sequence for $d$-AND-$\kminTreeCut$, the verifier algorithm must return a correct answer for $d$-AND-$\kminTreeCut$ (as $d$-AND-$\kminTreeCut\in\dNP$). This answer is the same as the one for $\kconn$, so we obtain an $\dNP$-algorithm for $\kconn$.

\paragraph{Proof: $\protect\kconn\in\protect\dMA$.}

The proof goes in almost the same way as above. However, we need to
be more very careful about the notion of oblivious adversary.

We first claim that dynamic greedy tree packing $\T$ with $d=\polylog(n)$
forests can be maintained in $\polylog(n)$ time \emph{against an
	oblivious adversary}. This is true by observing that the algorithm
by Kapron, King and Mountjoy \cite{KapronKM13} which shows $\conn\in\dBPP$
can actually maintain dynamic spanning forest \emph{against oblivious
	adversary}. This result further can be extended to maintaining dynamic
minimum spanning forest \emph{against oblivious adversary }with $d$-distinct
edge weights where $d=\mbox{polylog}(n)$\footnote{We note that this cannot be obtain using the reduction in \cite{HenzingerK99}
	because the reduction require the algorithm to work against adaptive
	adversary. Nevertheless, the algorithm in \cite{KapronKM13} can be
	directed extend to $d$-weight minimum spanning forest (roughly by
	leveling edges of each weight in the increasing order and finding
	a spanning forest on a graph with small weight first).}. As greedy tree packing is reducible to $d$-weight minimum spanning
forest by \ref{lem:packing to MST}, we obtain the claim. 

As before, \ref{thm:kconn to many kMinTreeCut} implies that, to solve
$\kconn$, it is enough to maintain a dynamic greedy tree packing
$\T$ with $d$ forests, feed the graph and $d$ forests from $\T$
as an input to the subroutine for $d$-AND-$\kminTreeCut$, and then
return the \emph{same} YES/NO answer returned by the subrouine for
$d$-AND-$\kminTreeCut$, in every step. 

We claim that, in this algorithm, the greedy tree packing $\T$ can
be correctly maintained with high probability although the algorithm
by \cite{KapronKM13} only works against an oblivious adversary. This
is because (1) the update sequence generated to the dynamic greedy
tree packing algorithm only comes from the adversary. (We do not \emph{adaptively}
generate more updates that depend on the answer of the subroutine
for $d$-AND-$\kminTreeCut$.) and (2) assuming that the algorithm
has been returning only correct YES/NO answers, the adversary learns
nothing about the internal random choices of the algorithm. So adaptive
adversaries do not have more power than oblivious adversaries against
this algorithm. So, we have that $\T$ is correctly maintained with
high probability.

Next, we claim that $\kconn\in\dMA$. This is because, we can assume
with high probability that $\T$ is correctly maintained at every
step. So given a sequence of reward-maximizing proof sequence for
$d$-AND-$\kminTreeCut$, the verifier algorithm must return a correct
answer for $d$-AND-$\kminTreeCut$ (as $d$-AND-$\kminTreeCut\in\dNP$).
This answer is the same as the one for $\kconn$, so we obtain an
$\dMA$-algorithm for $\kconn$.

\subsection{$\protect\kminTreeCut\in\protect\dNP$}

In this section, it remains to prove the following:
\begin{thm}
\label{thm:kMinTreeCut}$\kminTreeCut\in\dNP$.
\end{thm}
We essentially use the algorithm in \cite{Thorup01}, so we will refer
most of the definitions to \cite{Thorup01}. We will only point out
which part of the algorithm can be speed up using non-determinism.

\paragraph{Review of \cite{Thorup01}.}

Let $G=(V,E)$ and $F\subseteq E$ be a graph and a forest we are
maintaining. We will implement \emph{top tree} on $F$. The top tree \cite{AlstrupHLT05} 
hierarchically decomposes $F$ into \emph{clusters}. Each cluster $C$ is
a connected subgraph of $F$. Clusters are edge disjoint. The hierarchy
of clusters forms a binary tree where the root cluster $C$ corresponds
to $F$ itself. For each cluster $C$, there are two children $C_{L}$
and $C_{R}$ which partitions edge of $C$ into two connected parts.
The important property of clusters is such that each cluster $C$
share at most $2$ nodes with other clusters which are not descendant
of $C$. Hence, there are two types of clusters: \emph{path clusters}
(i.e. ones that share two nodes) and \emph{point clusters} (i.e. ones
that share one node). For a path cluster $C$, let $a$ and $b$ be
the two nodes which $C$ share with other non-descendant clusters.
We call $a$ and $b$ \emph{boundary nodes} of $C$. Let $\pi=(a,\dots,b)\subseteq C$
be a path connecting $a$ and $b$. We call $\pi$ a \emph{cluster
path} of $C$.

After each update to $G$ and $F$, there is an algorithm which dictates
how the hierarchy of top tree clusters should change so that the binary
tree corresponds to the hierarchy has logarithmic depth. This is done
by ``joining'' and ``destroying'' $O(\log n)$ clusters. This
normally implies efficiency except that we will also maintain some
information on each cluster. Then, the problem reduces to how to obtain
such information when a new cluster $C$ is created from joining its
two children clusters $C_{L}$ and $C_{R}$.

Recall that our goal is to know if $\min_{e\in F}cover_{(G,F)}(e)\ge k$.
So it suffices if for each edge $e\in F$ we implicitly maintain $cover'_{(G,F)}(e)=\min\{cover_{(G,F)}(e),k\}$
such that the implicit representation allows us to $\min_{e\in F}cover'_{(G,F)}(e)$,
because $\min_{e\in F}cover'_{(G,F)}(e)\ge k$ iff $\min_{e\in F}cover_{(G,F)}(e)\ge k$.

Towards this implicit representation, we do the following. For every
cluster $C$, we implicitly maintain a \emph{local cover number w.r.t.
$C$} (using top tree as well). This local cover numbers
are such that, for a root cluster $C$, the local cover number w.r.t.
$C$ of $e$ denoted by $lcover_{(G,C)}(e)=cover'_{(G,F)}(e)$. Moreover,
we can obtain $\min_{e\in C}lcover{}_{(G,C)}(e)=\min_{e\in F}cover'_{(G,F)}(e)$
when $C$ is the root cluster. Therefore, the main task is how to
maintain local cover numbers on each cluster. 

As described above, to maintain information on each cluster, the only
task which is non-trivial is when we create a new cluster $C$ from
two clusters $C_{L}$ and $C_{R}$ where $C=C_{L}\cup C_{R}$. It
turns out that if $C_{L}$ and $C_{R}$ are point clusters, then we
need to do nothing. We only need to describe what to do when $C_{L}$
and/or $C_{R}$ are path clusters. We only describe the case for $C_{R}$
because it is symmetric for $C_{L}$. Let $\pi=(a_{0},\dots,a_{s})$
be a path cluster of $C_{R}$ where $C_{L}$ and $C_{R}$ share a
node at $a_{0}$. Let $E'(C_{L},C_{R})=\{(u_{L},u_{R})\in E-F\mid u_{L}\in C_{L},u_{R}\in C_{R}\}$
be the set of non-tree edges whose endpoints are in $C_{L}$ and $C_{R}$. 

We will state without proof that what we need is just to find $k$ non-tree
edges from $E'(C_{L},C_{R})$ that ``cover'' $\pi$ as much as possible
(see \cite{Thorup01} for the argument). More precisely, for $e\in E'(C_{L},C_{R})$,
let $P_{e}\subseteq F$ be the path in $F$ connecting two endpoints
of $e$. Observe that $P_{e}\cap\pi=(a_{0},\dots,a_{s_{e}})$ for
some $0\le s_{e}\le s$. We want to find different $k$ edges $e$
from $E'(C_{L},C_{R})$ whose $s_{e}$ is as large as possible. This
is the bottleneck in \cite{Thorup01}. To do this, Thorup builds a
data structure based on the 2-dimensional topology tree of Frederickson
\cite{Frederickson85,Frederickson97} for obtaining these edges. This
results in $\tilde{O}(\sqrt{n})$ update time.

\paragraph{Using proofs from the prover.}

We will let the prover give us such $k$ non-tree edges from $E'(C_{L},C_{R})$
that ``cover'' $\pi$ as much as possible. Given an edge $e$, the
verifier can check if $e\in E'(C_{L},C_{R})$ and can compute the
value $s_{e}$. Then, the verifier will use $s_{e}$ as a ``reward''
for the prover (see \ref{def:NP}). Therefore, the ``honest'' prover
will always try to maximize $s_{e}$. 

To summarize, after each update, there will be $O(\log n)$ many new
clusters created. For each cluster, the verifier needs $O(k)$ edges
maximizing the rewards as defined above. By, for example, concatenating
the rewards of all $O(k\log n)$ into one number, the prover can provide
these $O(k\log n)$ edges as we desired. If at any time $\min_{e\in F}cover_{(G,F)}(e)\ge k$,
with the help from the prover, the verifier can indeed verify that
$\min_{e\in F}cover_{(G,F)}(e)\ge k$. Hence, $\kminTreeCut\in\dNP$.

\part{Back to RAM}
\label{part:RAM}

\section{Completeness of $\protect\dDNF$ in RAM with Large Space}

\label{sec:RAM complete}
The main goal of this section to prove that $\dNP$-completeness of
$\dDNF$ also holds in the word-RAM model in the setting 
when we allow quasipolynomial preprocessing time and space.
Given this main result
and that other results never exploit that non-uniformity of the bit-probe
model, we conclude that other results in this paper also transfer
to the word-RAM model.

\paragraph{The word-RAM model and complexity classes.}

The standard assumption for most algorithms in the word-RAM model
is that an index to a cell in memory fits a word. More precisely,
if the space is $s$, then the word size $w\ge\lg s$. We will also
adopt this assumption here. So when the space can be $s=2^{\polylog(n)}$,
the word size is $w=\polylog(n)$. Recall that the cost in the word-RAM
model is the number of ``standard'' operations (e.g. addition, multiplication)
on words. From now, we just say the RAM model.

For each complexity classes formally defined in the bit-probe model
in this paper, we can easily define an analogous class in the word-RAM
model. The only difference is just that there is a concept of preprocessing
time in the word-RAM model. When given a size-$n$ instance, we allow
algorithms to use $2^{\polylog(n)}$ preprocessing time and space.
To be more precise, we denote the class $\dP$ in the RAM model as
$\dPram$ and we do similarly for other classes. 
As it will be very tedious to formally rewrite all the definitions of complexity classes
in the word-RAM model again, we will omit it in this version. 

The main result we prove this in section is as follows: 
\begin{thm}
	\label{thm:DNF complete RAM}If $\dDNF\in\dPram$, then $\dPram=\dNPram$. 
\end{thm}
\ref{thm:DNF complete RAM} implies that, unless $\dPram=\dNPram$,
there is no algorithm for $\dDNF$ with polylogarithmic update time,
even when allowing quasi-polynomial space and preprocessing time.

It is clear that $\dDNF\in\dNPram$. 
So, this will imply that $\dDNF$ is a $\dNPram$-complete in the sense of \Cref{def:completeness}.

\paragraph{Why one might believe that $\protect\dPram\protect\neq\protect\dNPram$.}

Consider the following problem: We are given a graph $G$ with $n$
nodes undergoing edge insertions and deletions. Then, after each update,
check if $G$ contains a clique of size at least $s=\polylog(n)$
in $G$. This problem is in $\dNPram$\footnote{In the preprocessing step, we check all set of nodes of size $s$
	trivially if they form a clique. In each update step, given a proof
	indicating the set $S$ of nodes forming a clique, we can check in
	polylogarithmic time of $S$ is indeed a clique. This show that the
	problem is in $\dNPram$. }. Now, if $\dNPram=\dPram$, then this implies that we can do the
following. Given an \emph{empty }graph with $n$ nodes, we preprocess
this empty graph using quasi-polynomial space and time. Then, given
an online sequence of arbitrary graphs $G_{1},G_{2},\dots$ with $n$
nodes and $m$ edges, we can decide if $G_{i}$ contains a clique
of size $\polylog(n)$ in $O(\polylog(n))$ time before $G_{i+1}$
arrives. This means that, by spending quasi-polynomial time and space
for preprocessing even without any knowledge about a graph, we can
then solve the $\polylog(n)$-clique problem in near-linear time,
which would be quite surprising. Moreover, the same kind of argument
applies to every problem in $\dNPram$

The rest of this section is for proving \ref{thm:DNF complete RAM}.
The proof is essentially the same as the $\dNP$-completeness of $\dDNF$
in the bit-probe model. This is because all the reduction are readily
extended to the relaxed RAM model.

\subsection{Proof of the Completeness.}

Suppose that $\dDNF\in\dPram$. There is an $\dPram$-algorithm $\O$
for the First Shallow Decision Tree ($\fDT$) problem (recall the
definition from \ref{sub:DT def}). This is because all reductions
in \ref{sub:DNF:NPhard} extend to the RAM model, 

Let $\D\in\dNPram$. Let $\V$ be the $\dNPram$-verifier for $\D$.
We want to devise a $\dPram$-algorithm $\A$ for $\D$ using $\V$
and $\O$. The idea from now is just to repeat the reduction which
shows that $\fDT$ is $\dNPprb$-hard \ref{sub:DT hard}. Most of
the reductions are readily extended to the RAM model. We only need
to bound the preprocessing time (which gives the bound for space).
We describe how $\A$ works as follows.

\paragraph{Preprocessing.}

In the preprocessing step, given an $n$-size initial instance $I_{0}$
of the problem $\D$. $\A$ calls $\V$ to preprocess $I_{0}$ which
takes time $2^{\polylog(n)}$. We know that when $\V$ is given an
input in each update step from now, $\V$ takes time at most $\polylog(n)$.
Hence, the process how $\V$ works can be described with the decision
tree $T_{\V}$. By simulating $\V$ step by step and by branching
whenever $\V$ reads any cell in the memory, $\A$ can construct the
decision tree $T_{\V}$ in time $2^{\polylog(n)}$ (as $T_{\V}$ can
contains at most $2^{\polylog(n)}$ nodes). Now, we construct the
collection $\T_{\V}=\{T_{\pi}\mid\pi\in\{0,1\}^{\polylog(n)}\}$ where
$T_{\pi}$ is obtained from $T_{\V}$ by ``fixing'' the proof-update
cell in $\mem_{\V}$ to be $\pi$. See \ref{sub:hardness init oracle}
for the same construction.

Next, we basically follow the same process as in \ref{sub:hardness update}.
That is, we construct an initial instance $I'$ of the problem $\fDT$
(recall the definition of $\fDT$ from \ref{sub:DT def}) for feeding
to the algorithm $\O$. The approach is the same as in \ref{sub:hardness init oracle}.
Let $I'=(\mem_{I'},\T_{I'})$ where $\mem_{I'}=\mem_{\V}(0)$
is the memory state of $\V$ after $\V$ preprocess $I_{0}$, and $\T_{I'}=\T_{\V}$.
Observe that the instance $I'$ can be described using $2^{\polylog(n)}$
bits, and $I$ can be constructed in $2^{\polylog(n)}$ time as well.
Now, we call $\O$ to preprocess $I'$. This takes time $2^{\polylog(2^{\polylog(n)})}=2^{\polylog(n)}$.
In total the preprocessing time is $2^{\polylog(n)}$.

\paragraph{Update.}

We follow the same process as in \ref{sub:hardness update}. That
is, at step $t$, given an instance-update $u(I_{t-1}I_{t})$ of the
problem $\D$, $\A$ does the followings: 
\begin{enumerate}
	\item $\A$ write $u(I_{t-1}I_{t})$ in the input-memory part $\mem_{\V}^{\inp}$
	of $\V$. 
	\item $\A$ calls $\O$ $|u(I_{t-1}I_{t})|=\lambda_{\D}(n)$ many times
	to update $\mem_{I'}=\mem_{\V}$. 
	\item $\O$ returns the guaranteed proof-update $\pi_{t}$. 
	\item $\A$ calls $\V$ with input $(u(I_{t-1}I_{t}),\pi_{t})$. 
	\item $\A$ calls $\O$ $\polylog(n)$ many times to update $\mem_{I'}=\mem_{\V}$. 
	\item $\A$ outputs the same output as $\V$'s. 
\end{enumerate}
As $\V$ always get the guaranteed proof-update, the answer from $\V$
is correct and so is $\A$. Now, we analyze the update time. Each
call to $\O$ takes 
\[
\polylog(2^{\polylog(n)}))=\polylog(n).
\]
The total number of calls is $\polylog(n)$. So the total time is
$\polylog(n)$. Other operations are subsumed by this.

To conclude, $\A$ takes $2^{\polylog(n)}$ preprocessing time and
$\polylog(n)$ update time, and $\A$ returns a correct answer for
$\D$ at every step. Therefore, $\D\in\dPram$. Hence, $\dPram=\dNPram$.

\section*{Acknowledgement} 
Nanongkai and Saranurak thank Thore Husfeldt for bringing \cite{HusfeldtR03} to their attention.  We thank Thomas Schwentick for pointing out the mistakes in the definition of FNP and TFNP defined in the previous version of this paper.

This project has received funding from the European Research Council (ERC) under the European Union's Horizon 2020 research and innovation programme under grant agreement No 715672. Nanongkai and Saranurak were also partially supported by the Swedish Research Council (Reg. No. 2015-04659.)

\part{Appendices}

\appendix
\section{Extending our results to promise problems}
\label{sec:promise}

For simplicity of exposition, we focussed on decision problems throughout this paper. However, almost all the results derived in this paper -- {\em except} the ones from~\ref{sec:PH:promise} -- hold for decision problems {\em with a promise} as well.\footnote{A decision problem {\em with a promise} allows for ``don't care'' instances, in addition to YES and NO instances. An algorithm for such a problem is allowed to provide any arbitrary output on don't care instances. } The results from~\ref{sec:PH:promise} don't hold for promise problems because of the following reason:~\ref{lem:NP:coNP} about $\dNP \cap \dcoNP$ does not hold for promise problems (see \Cref{sec:issue promise} for deeper explanation).
For the similar reason, we have that \Cref{lem:removing intersection randomized oracle,lem:if NP in coMA} do not hold for promise problems.

\section{Proofs: putting problems into classes}

\label{sec:app:classification}

\subsection{$\protect\dBPP$}

\begin{prop}
	\label{lem:conn in BPP}Dynamic connectivity is in $\dBPP$.\end{prop}
\begin{proof}
	The goal of this proof is give only to verify that the algorithm by
	Kapron, King and Mountjoy \cite{KapronKM13} really satisfies our
	definition of $\dBPP$ in \ref{def:BPP}. To point out some small
	difference, from \ref{def:BPP}, the algorithm needs to handle infinite
	sequence of updates, and the algorithm should be correct on \emph{each}
	update with good probability. However, from \cite{KapronKM13}, there
	is an dynamic connectivity algorithm such that on $n$-node graphs,
	the algorithm is correct on \emph{all }of the first $\mbox{poly}(n)$
	updates with $1-1/n^{c}$ for any constant $c$. This algorithm has
	$\tilde{O}(cm)=\tilde{O}(cn^{2})$ preprocessing time if the initial
	graph has $m$ edges and has $\tilde{O}(1)$ update time. 
	
	Because the preprocessing time is small enough, we can apply a standard
	technique of periodically rebuilding a data structure every, say $\tilde{O}(n^{3})$
	updates. The work for rebuilding the data structure is spread over
	man updates so that the worst-case update time is still $\tilde{O}(1)$.
	For each period, the answer of all updates is correct with high probability.
	Hence, each answer is trivially correct high probability as well.
	This holds for every period. So this shows that dynamic connectivity
	is in $\dBPP$.\end{proof}

\begin{prop}
\label{prop:approx mincut in BPP}For any constant $\epsilon>0$,
there is a dynamic algorithm that can maintain the $(2+\epsilon)$-approximate
size of min cut in an unweighted graph with $n$ nodes undergoing
edge insertions and deletions with $O(\mbox{polylog}(n))$ update
time. The answer at each step is correct with probability at least
$1-1/\mbox{poly}(n)$. In particular, the dynamic $(2+\epsilon)$-approximate
mincut problem is in $\dBPP$.\end{prop}
\begin{proof}
[Proof sketch]In \cite{ThorupK00}, Karger and Thorup shows that
a dynamic algorithm for $(2+\epsilon)$-approximating size of min
cut with polylogarithmic worst-case update time can be reduced to
a dynamic algorithm for maintaining a minimum spanning forest in a
graph of size $n$ where the weight of each edge is between $1$ and
$\mbox{polylog}n$. Kapron, King and Mountjoy \cite{KapronKM13} show
an algorithm for this problem with polylogarithmic worst-case time
and their algorithm is correct at each step with high probability\footnote{In \cite{KapronKM13}, they only claim that their algorithm can handle
polynomial many updates. But this restriction can be easily avoided
using a standard trick of periodically rebuilding the data structure.}. 
\end{proof}

\subsection{$\protect\dNP$}
\label{sec:problem in rankNP}

In this section, we show that all the problems listed in~\ref{table:BPP-NP} belong to the class $\dNP$. In particular,  we need to show that all these problems admit verifier-families satisfying~\ref{def:NP}. We say that the \emph{proof} at each step is given to
the verifier from a \emph{prover} $P$. We think of ourselves as a
verifier $\V$. At each step $t$, the prover gives us some proof $\pi_t$. At
some step, the prover also claims that the current instance is an YES instance,
and must convince us that it is indeed is the case.

\begin{prop}
\label{prop:conn in NP}The dynamic graph connectivity problem is in $\dNP$.
\end{prop}
\begin{proof}(Sketch) 
This problem is in $\dNP$ due to the following verifier. After every update, the verifier with proofs from the prover maintain a forest $F$ of $G$. A proof (given after each update) is an edge insertion to $F$ or an $\bot$ symbol indicating that there is no update to $F$.  It handles each update as follows.
\begin{itemize}[noitemsep]
	\item After an edge $e$ is inserted into $G$, the verifier checks if $e$ can be inserted into $F$ without creating a cycle. This can be done in $O(\log(n))$ time using a link/cut tree data structure \cite{SleatorT81}. It outputs reward $y=0$. (No proof from the prover is needed in this case.)
	\item After an edge $e$ is deleted from $G$,  the verifier checks if $F$ contains $e$. If not, it outputs  reward $y=0$ (no proof from the prover is needed in this case). If $e$ is in $F$, the verifier reads the proof (given after $e$ is deleted). If the proof is $\bot$ it outputs reward $y=0$. Otherwise, let the proof be an edge $e'$. The verifier checks if $F'=F\setminus\{e\}\cup \{e'\}$ is a forest; this can be done in $O(\log(n))$ time using a link/cut tree data structure \cite{SleatorT81}. If $F'$ is a forest, the verifier sets $F\gets F'$ and outputs reward $y=1$; otherwise, it outputs reward $y=-1$.  
\end{itemize}
When the verifier receives two nodes $u, v$ in $G$ as part of a query, it outputs $x=YES$ if and only if both $u$ and $v$ belong to the same component of $F$ (this again can be checked in $O(\log (n))$ time using a link/cut tree data structure).

Observe that if the prover gives a proof that maximizes the reward after every update, the forest $F$ will always be a spanning forest (since inserting an edge $e'$ to $F$ has higher reward than giving $\bot$ as a proof). Thus, the verifier will always output $x=YES$ for YES-instances in this case. It is not hard to see that the verifier never outputs $x=YES$ for NO-instances, no matter what the prover does. 
\end{proof}

\begin{prop}
\label{prop:matching in NP}The dynamic problem of maintaining a $(1+\epsilon)$-approximation to the size of the maximum matching  is in $\dNP$.
\end{prop}
\begin{proof}(Sketch)
We consider the gap version of the problem. Let $G$
be an $n$-node graph undergoing edge updates and $k$ be some number.
Let $\opt(G)$ be the size of maximum matching in $G$. We say that
$G$ is an  YES-instance if $\opt(G)>(1+\epsilon)k$, $G$ is a NO-instance
if $\opt(G)\le k$. Suppose that $M$ is a matching of size $|M|\le k$.
If $\opt(G)>(1+\epsilon)k$, then it is well-known that there is an augmenting
path $P$ for $G$ of length at most $O(1/\epsilon)$. Now, our verifier for this problem 
works as follows.

In the preprocessing step, we compute the maximum matching $M$ of
$G$. After each edge update, the prover gives us an augmenting path
$P$ to $M$ of length $O(1/\epsilon)$. We (as the verifier) can check if $P$ is
indeed an augmenting path for $M$ of length $O(1/\epsilon)$ in $O(1/\epsilon)$ time. Then we do the following.
\begin{itemize}
\item If $P$ is not an augmenting path, then we output a reward $y = -1$.
\item Else if $P$ is a valid and nontrivial augmenting path of length $O(1/\epsilon)$, then we output a reward $y = 1$. In this scenario, we also augment $P$ to $M$ in $O(1/\epsilon)$
time.
\item Otherwise, if $P$ is a valid but trivial augmenting path (i.e., it has length zero), then we output a reward $y = 0$.
\end{itemize}
At this point, we check the size of our matching $M$. If we find that $|M|>k$, then we
(as a verifier) output $x = 1$ (YES). Otherwise, we output $x = 0$ (NO).

If $\opt(G)\le k$, then obviously we always output $x = NO$. If $\opt(G)>(1+\epsilon)k$
and the prover so far has always maximized the reward $y$ (by giving us an augmenting path of length
$O(1/\epsilon)$ whenever it is possible), then it must follow that
$|M|>k$ and so we will output $x = YES$. This shows that this problem is
in $\dNP$.
\end{proof}

\begin{prop}
\label{prop:small cert in NP}The following dynamic problems from
\ref{table:BPP-NP} are in $\dNP$ (and hence in $\dNP$)
\begin{itemize}
\item \textup{subgraph detection,}
\item \textup{$uMv$,}
\item \textup{3SUM, and}
\item \textup{planar nearest neighbor.}
\item \textup{Erickson's problem.}
\item \textup{Langerman's problem.}
\end{itemize}
\end{prop}
\begin{proof}(Sketch) 
The general idea for showing that the above problems
are in $\dNP$ is as follows. The prover is supposed to say nothing when the
current instance is a NO-instance. On the other hand, when the current instance is a
YES-instance, the prover is supposed to give the whole certificate (which is
small) and we quickly verify the validity of this certificate. 

\medskip
\noindent \textbf{(subgraph detection)}: This problem is in $\dNP$ due to the following verifier: the verifier outputs $x=YES$ if and only if the proof (given after each update) is a mapping of the edges in $H$ to the edges in a subgraph of $G$ that is isomorphic to $H$. With output $x=YES$, the verifier gives  reward $y=1$. With output $x=NO$, the verifier gives reward $y=0$.  Observe that the proof  is of polylogarithmic size (since $|V(H)| = \polylog (|V(G)|)$), and the verifier can calculate its outputs $(x,y)$ in polylogarithmic time. Observe further that:
\begin{itemize}
\item (1) If the current input instance is a YES-instance, then the reward-maximizing proof is a mapping between $H$ and the subgraph of $G$ isomorphic to $H$, causing the verifier to output $x=YES$.
\item (2) If the current input instance is a NO-instance, then no proof will make the verifier output $x=YES$.  
\end{itemize}

\medskip
\noindent The arguments for the problems below are exactly similar to the one sketched above.

\medskip
\noindent 
\textbf{($uMV$):} After each update, if $uMv=1$, then the prover
gives indices $i$ and $j$ and we check if $u_{i}\cdot M_{ij}\cdot v_{j}=1$.

\medskip
\noindent
\textbf{(3SUM): }After each update, if there are $a,b,c\in S$ where
$a+b=c$, then the prover just gives pointers to $a,b$ and $c$.

\medskip
\noindent
\textbf{(planar nearest neighbor):} We consider the decision version
of the problem. Given a queried point $q$, we ask if there is a point
$p$ where $d(p,q)\le k$ for some $k$. After each update, if there
is a point $p$ where $d(p,q)\le k$, then the prover gives the point
$p$ to us and we verified that $d(p,q)\le k$ in $O(1)$ time.

\medskip
\noindent
\textbf{(Erickson's problem):} We consider the decision version of the problem. Given a queried value $\alpha$, we need to answer if there is an entry in the matrix with value at least $\alpha$. After each update, if there is an entry in the matrix with value at least $\alpha$, then the prover simply points to that entry in the matrix.

\medskip
\noindent
\textbf{(Langerman's problem):} After each update, if there is an index $k$ such that $\sum_{i=1}^k A[i] = 0$, then the prover simply points to such an index $k$. We (as the verifier) then check if it is indeed the case that $\sum_{i=1}^k A[i] = 0$. This takes $O(\log n)$ time if we use a standard balanced tree data structure built on top of the array $A[1 \ldots n]$.
\end{proof}

\subsection{$\dPH$}
\label{sec:PH:examples}

In this section, we show that all the problems listed in~\ref{table:PH} are in $\dPH$.

\paragraph{Small dominating set:} An instance of this problem consists of an undirected graph $G = (V, E)$ on $|V| = N$ nodes (say) and a parameter $k = O(\polylog (N))$. It is an YES instance if there is a dominating set $S \subseteq V$ in $G$ that is of size at most $k$, and a NO instance otherwise. An instance-update consists of inserting/deleting an edge in the graph $G$. 

\begin{lem}
\label{lem:small dominating set} 
The small dominating set problem described above is in the complexity class $\sig_2$. 
\end{lem}

\begin{proof}(Sketch)
We first define a new dynamic problem $\D$ as follows. 
\begin{itemize}
\item {\bf $\D$:} An instance of $\D$ is  an ordered pair $(G, S)$, where $G = (V, E)$ is a graph on $|V| = N$ nodes (say), and $S \subseteq V$ is a subset of nodes in $G$ of size at most $k$ (i.e., $|S| \leq k$). It is an YES instance if $S$ is {\em not} a dominating set of $G$, and a NO instance otherwise. An instance-update consists of either inserting/deleting an edge in $G$, or inserting/deleting a node in $S$ (the node-set $V$ remains unchanged).

It is easy to check that the problem $\D$ is in $\dNP$ as per~\ref{def:NP}. Specifically, in this problem a proof $\pi_t$ for the verifier  encodes a node  $v \in V$. Upon receiving this proof, the verifier checks whether or not $v \notin S$ and it also goes through all the neighbors of $v$ in $S$ (which takes $O(\polylog (N))$ time since $|S| \leq k = O(\polylog (N))$.  The verifier then outputs $(x_t, y_t)$ where $x_t = y_t \in \{0,1 \}$, and  $x_t = 1$ iff $v \notin S$ and $v$ does not have any neighbor in $S$  (YES instance). 
\end{itemize}
The lemma holds since  the small  dominating set problem is in $\dNP$ if the verifier $\V'$ (for small dominating set) can use an oracle for the problem $\D$. We now explain why this is the case. Specifically, a proof $\pi'_t$ for $\V'$ encodes a subset $S \subseteq V$ of nodes in $G$ of size at most $k$. Since $k = O(\polylog (N))$, the proof $\pi'_t$ also requires only $O(\polylog (N))$ bits. Upon receiving this proof $\pi'_t$, the verifier $\V'$ calls the oracle for $\D$ on the instance $(G, S)$. If this oracle returns a NO answer, then we know that $S$ is  a dominating set in $G$, and the verifier $\V'$ outputs $(x'_t, y'_t)$ where $x'_t = y'_t = 1$. Otherwise, if the oracle returns YES, then the verifier $\V'$ outputs $(x'_t, y'_t)$ where $x'_t = y'_t = 0$.
\end{proof}

\paragraph{Small vertex cover:} An instance of this problem consists of an undirected graph $G = (V, E)$ on $|V| = N$ nodes (say) and a parameter $k = O(\polylog (N))$. It is an YES instance if there is a vertex cover $S \subseteq V$ in $G$ that is of size at most $k$, and a NO instance otherwise. An instance-update consists of inserting/deleting an edge in the graph $G$. 

\begin{lem}
\label{lem:small vertex cover} 
The small vertex cover problem described above is in the complexity class $\sig_2$. 
\end{lem}

\begin{proof}(Sketch)
We first define a new dynamic problem $\D$ as follows. 
\begin{itemize}
\item {\bf $\D$:} An instance of $\D$ is  an ordered pair $(G, S)$, where $G = (V, E)$ is a graph on $|V| = N$ nodes (say), and $S \subseteq V$ is a subset of nodes in $G$. It is an YES instance if $S$ is {\em not} a vertex cover of $G$, and a NO instance otherwise. An instance-update consists of either inserting/deleting an edge in $G$, or inserting/deleting a node in $S$ (the node-set $V$ remains unchanged).

It is easy to check that the problem $\D$ is in $\dNP$ as per~\ref{def:NP}. Specifically, in this problem a proof $\pi_t$ for the verifier  encodes an edge $(v, u) \in E$. The verifier outputs $(x_t, y_t)$ where $x_t = y_t \in \{0,1 \}$, and  $x_t = 1$ iff $u \notin S$ and $v \notin S$  (YES instance). 
\end{itemize}
The lemma holds since  the small vertex cover problem is in $\dNP$ if the verifier $\V'$ (for small vertex cover) can use an oracle for the problem $\D$. We now explain why this is the case. Specifically, a proof $\pi'_t$ for $\V'$ encodes a subset $S \subseteq V$ of nodes in $G$ of size at most $k$. Since $k = O(\polylog (N))$, the proof $\pi'_t$ also requires only $O(\polylog (N))$ bits. Upon receiving this proof $\pi'_t$, the verifier $\V'$ calls the oracle for $\D$ on the instance $(G, S)$. If this oracle returns a NO answer, then we know that $S$ is  a vertex cover in $G$, and the verifier $\V'$ outputs $(x'_t, y'_t)$ where $x'_t = y'_t = 1$. Otherwise, if the oracle returns YES, then the verifier $\V'$ outputs $(x'_t, y'_t)$ where $x'_t = y'_t = 0$.
\end{proof}

\paragraph{Small maximal independent set:} An instance of this problem consists of an undirected graph $G = (V, E)$ on $|V| = N$ nodes (say) and a parameter $k = O(\polylog (N))$. It is an YES instance if there is a maximal independent set $S \subseteq V$ in $G$ that is of size at most $k$, and a NO instance otherwise. An instance-update consists of inserting/deleting an edge in the graph $G$. 

\begin{lem}
\label{lem:small maximal independent set} 
The small maximal independent set problem described above is in the complexity class $\sig_2$. 
\end{lem}

\begin{proof}(Sketch)
We first define a new dynamic problem $\D$ as follows. 
\begin{itemize}
\item {\bf $\D$:} An instance of $\D$ is  an ordered pair $(G, S)$, where $G = (V, E)$ is a graph on $|V| = N$ nodes (say), and $S \subseteq V$ is a subset of nodes in $G$ of size at most $k$ (i.e., $|S| \leq k$). It is an YES instance if $S$ is {\em not} a maximal independent set of $G$, and a NO instance otherwise. An instance-update consists of either inserting/deleting an edge in $G$, or inserting/deleting a node in $S$ (the node-set $V$ remains unchanged).

It is easy to check that the problem $\D$ is in $\dNP$ as per~\ref{def:NP}. Specifically, in this problem a proof $\pi_t$ for the verifier  encodes a node  $v \in V$. Upon receiving this proof, the verifier checks whether or not $v \notin S$ and it also goes through all the neighbors of $v$ in $S$ (which takes $O(\polylog (N))$ time since $|S| \leq k = O(\polylog (N))$.  The verifier then outputs $(x_t, y_t)$ where $x_t = y_t \in \{0,1 \}$, and  $x_t = 1$ iff either (1) $v \notin S$ and $v$ does not have any neighbor in $S$ or (2) $v \in S$ and $v$ has at least one neighbor in $S$  (YES instance). 
\end{itemize}
The lemma holds since  the small  maximal independent set problem is in $\dNP$ if the verifier $\V'$ (for small maximal independet set) can use an oracle for the problem $\D$. We now explain why this is the case. Specifically, a proof $\pi'_t$ for $\V'$ encodes a subset $S \subseteq V$ of nodes in $G$ of size at most $k$. Since $k = O(\polylog (N))$, the proof $\pi'_t$ also requires only $O(\polylog (N))$ bits. Upon receiving this proof $\pi'_t$, the verifier $\V'$ calls the oracle for $\D$ on the instance $(G, S)$. If this oracle returns a NO answer, then we know that $S$ is  a maximal independent set in $G$ of size at most $k$, and the verifier $\V'$ outputs $(x'_t, y'_t)$ where $x'_t = y'_t = 1$. Otherwise, if the oracle returns YES, then the verifier $\V'$ outputs $(x'_t, y'_t)$ where $x'_t = y'_t = 0$.
\end{proof}

\paragraph{Small maximal matching:} An instance of this problem consists of an undirected graph $G = (V, E)$ on $|V| = N$ nodes (say) and a parameter $k = O(\polylog (N))$. It is an YES instance if there is a maximal matching $M \subseteq E$ in $G$ that is of size  at most $k$ (i.e., $|M| \leq k$), and a NO instance otherwise. An instance-update consists of inserting/deleting an edge in the graph $G$.

\begin{lem}
\label{lem:small maximal matching} 
The small maximal matching problem described above is in the complexity class $\sig_2$. 
\end{lem}

\begin{proof}(Sketch)
We first define a new dynamic problem $\D$ as follows. 
\begin{itemize}
\item {\bf $\D$:} An instance of $\D$ is  an ordered pair $(G, M)$, where $G = (V, E)$ is a graph on $|V| = N$ nodes (say), and $M \subseteq E$ is a matching in $G$ of size at most $k$ (i.e., $|M| \leq k$). It is an YES instance if $M$ is {\em not} a maximal matching in $G$, and a NO instance otherwise. An instance-update consists of either inserting/deleting an edge in $G$, or inserting/deleting an edge in $M$ (ensuring that $M$ remains a matching in $G$).

It is easy to check that the problem $\D$ is in $\dNP$ as per~\ref{def:NP}. Specifically, in this problem a proof $\pi_t$ for the verifier  encodes an edge $(u, v) \in E$. Upon receiving this proof, the verifier checks whether or not both the endpoints $u, v$ are currently matched in $M$.  The verifier then outputs $(x_t, y_t)$ where $x_t = y_t \in \{0,1 \}$, and  $x_t = 1$ iff both $u, v$ are currently unmatched in $M$  (YES instance). 
\end{itemize}
The lemma holds since  the small maximal matching problem is in $\dNP$ if the verifier $\V'$ (for small maximal matching) can use an oracle for the problem $\D$. We now explain why this is the case. Specifically, a proof $\pi'_t$ for $\V'$ encodes a matching $M \subseteq E$  in $G$ of size at most $k$. Since $k = O(\polylog (N))$, the proof $\pi'_t$ also requires only $O(\polylog (N))$ bits. Upon receiving this proof $\pi'_t$, the verifier $\V'$ calls the oracle for $\D$ on the instance $(G, M)$. If this oracle returns a NO answer, then we know that $M$ is  a maximal matching in $G$ of size at most $k$, and the verifier $\V'$ outputs $(x'_t, y'_t)$ where $x'_t = y'_t = 1$. Otherwise, if the oracle returns YES, then the verifier $\V'$ outputs $(x'_t, y'_t)$ where $x'_t = y'_t = 0$.
\end{proof}

\paragraph{Chan's subset union problem:} An instance of this problem consists of a collection of sets $X_1, \ldots, X_n$ from universe $[m]$, and a set $S \subseteq [n]$. It is an YES instance if $\cup_{i \in S} X_i = [m]$, and a NO instance otherwise. An instance-update consists of inserting/deleting an index in $S$.

\begin{lem}
\label{lem:subset union}
Chan's subset union problem described above is in the complexity class $\pid_2$.
\end{lem}

\begin{proof}(Sketch)
We first define a new dynamic problem $\D$ as follows.
\begin{itemize}
\item {\bf $\D$:} An instance of $\D$ is a triple $(\X, S, e)$, where $\X = \{X_1, \ldots, X_n\}$ is a collection of $n$ sets from the universe $[m]$, $S\subseteq [n]$ is a set of indices, and $e \in [m]$ is an element in the universe $[m]$. It is an YES instance if $e \in \cup_{i \in S} X_i$ and a NO instance otherwise. An instance-update consists of inserting/deleting an index in $S$ and/or changing the element $e$.

It is easy to check that the problem $\D$ is in $\dNP$ as per~\ref{def:NP}. Specifically, in this problem a proof $\pi_t$ for the verifier encodes an index $j \in S$. The verifier outputs $(x_t, y_t)$ where $x_t = y_t \in \{0,1\}$, and $x_t = 1$ iff $e \in X_j$.
\end{itemize}
The lemma holds since  Chan's subset union problem is in $\dcoNP$ if the verifier $\V'$ (for Chan's subset union) can use an oracle for the problem $\D$. We now explain why this is the case. Specifically, a proof $\pi'_t$ for $\V'$ encodes an element $e \in [m]$.  Upon receiving this proof $\pi'_t$, the verifier $\V'$ calls the oracle for $\D$ on the instance $(\X, S, e)$. If this oracle returns a NO answer, then we know that $e \notin \cup_{i \in S} X_j$, and hence the verifier $\V'$ outputs $(x'_t, y'_t)$ where $x'_t = 0$ and $y'_t = 1$. Otherwise, if the oracle returns YES, then the verifier $\V'$ outputs $(x'_t, y'_t)$ where $x'_t = 1$ and $y'_t = 0$.
\end{proof}

\paragraph{$3$-vs-$4$ diameter:} An instance of this problem consists of an undirected graph $G = (V, E)$ on $|V| = N$ nodes (say). It is an YES instance if the diameter of $G$ is at most $3$, and a NO instance if the diameter of $G$ is at least $4$. An instance-update consists of inserting/deleting an edge in $G$.

\begin{lem}
\label{lem:small diameter} 
The $3$-vs-$4$ diameter problem described above is in the complexity class $\pid_2$. 
\end{lem}

\begin{proof}(Sketch)
We first define a new dynamic problem $\D$ as follows.
\begin{itemize}
\item {\bf $\D$:} An instance of $\D$ is a triple $(G, u, v)$, where $G = (V, E)$ is a graph on $|V| = N$ nodes (say), and $u, v \in V$ are two nodes in $G$. It is an YES instance if the distance (length of the shortest path) between $u$ and $v$ is at most $3$, and a NO instance otherwise. An instance-update consists of inserting/deleting an edge in the graph $G$ and changing the nodes $u, v$ (the node-set $V$ remains unchanged).

It is easy to check that the problem $\D$ is in $\dNP$ as per~\ref{def:NP}. Specifically, in this problem a proof $\pi_t$ for the verifier  encodes a path $P$ between $u$ and $v$ in $G$. The verifier outputs $(x_t, y_t)$ where $x_t = y_t \in \{0,1 \}$, and  $x_t = 1$ iff the length of  $P$ is at most $3$  (YES instance). 
\end{itemize}
The lemma holds since  the $3$-vs-$4$ diameter problem is in $\dcoNP$ if the verifier $\V'$ (for $3$-vs-$4$ diameter) can use an oracle for the problem $\D$. We now explain why this is the case. Specifically, a proof $\pi'_t$ for $\V'$ encodes two nodes $u, v \in V$ in $G$.  Upon receiving this proof $\pi'_t$, the verifier $\V'$ calls the oracle for $\D$ on the instance $(G, u, v)$. If this oracle returns a NO answer, then we know that the distance between $u$ and $v$ is at least $4$ (hence, the diameter of $G$ is also at least $4$), and the verifier $\V'$ outputs $(x'_t, y'_t)$ where $x'_t = 0$ and $y'_t = 1$. Otherwise, if the oracle returns YES, then the verifier $\V'$ outputs $(x'_t, y'_t)$ where $x'_t = 1$ and $y'_t = 0$.
\end{proof}

\paragraph{Euclidean $k$ center:} An instance of this problem consists of a set  of points $X \subseteq \mathbb{R}^d$ and a threshold $T \in \mathbb{R}$. It is an YES instance if there is a subset $C \subseteq X$ of points of size at most $k$ (i.e., $|C| \leq k$)  such that every point in $X$ is within a distance of $T$ from some point in $C$; and it is a NO instance otherwise. The parameter $k$ is such that $k = O(\polylog (n))$ where $n$ bits are needed to encode an instance. An instance-update consists of inserting or deleting a point in $X$.

\begin{lem}
\label{lem:small k center} 
The Euclidean $k$ center problem described above is in the complexity class $\sig_2$. 
\end{lem}

\begin{proof}(Sketch)
We first define a new dynamic problem $\D$ as follows.
\begin{itemize}
\item {\bf $\D$:} An instance of $\D$ is a triple $(X, C, T)$, where $X \subseteq \mathbb{R}^d$ is a set of points, $C \subseteq X$ is a subset of points of size at most $k$ (i.e., $|C| \leq k$), and $T \in \mathbb{R}$ is a threshold. It is an YES instance if there is a point $x \in X$ that is more than $T$ distance away from every point in $C$, and a NO instance otherwise. An instance-update consists of inserting/deleting a point in $X$ and/or  inserting/deleting a point in $C$ (ensuring that $C$ remains a subset of $X$).

It is easy to check that the problem $\D$ is in $\dNP$ as per~\ref{def:NP}. Specifically, in this problem a proof $\pi_t$ for the verifier  encodes a point $p \in X \setminus C$. The verifier outputs $(x_t, y_t)$ where $x_t = y_t \in \{0,1 \}$, and  $x_t = 1$ iff the point $p$ is more than $T$ distance away from every point in $C$  (YES instance). Note that the verifier can check if this is the case in $O(k) = O(\polylog (n))$ time, since there are at most $k$ points in $C$. 
\end{itemize}
The lemma holds since  the Euclidean $k$ center problem is in $\dNP$ if the verifier $\V'$ (for Euclidean $k$ center) can use an oracle for the problem $\D$. We now explain why this is the case. Specifically, a proof $\pi'_t$ for $\V'$ encodes a subset $C \subseteq X$ of size at most $k$.  Upon receiving this proof $\pi'_t$, the verifier $\V'$ calls the oracle for $\D$ on the instance $(X, C, T)$. If this oracle returns a NO answer, then we know that every point in $X$ is within a distance of $T$ from some point in $C$, and the verifier $\V'$ outputs $(x'_t, y'_t)$ where $x'_t = y'_t = 1$. Otherwise, if the oracle returns YES, then the verifier $\V'$ outputs $(x'_t, y'_t)$ where $x'_t = y'_t = 0$.
\end{proof}

\paragraph{$k$-edge connectivity:} An instance of this problem consists of an undirected graph $G = (V, E)$ on $|V| = N$ nodes (say). It is an YES instance if the graph $G$ is $k$-edge connected, and a NO instance otherwise. We assume that $k = O(\polylog (N))$. An instance-update consists of inserting/deleting an edge in $G$ (the node-set $V$ remains unchanged).

\begin{lem}
\label{lem:k edge connectivity} 
The $k$-edge connectivity problem described above is in the complexity class $\pid_2$. 
\end{lem}

\begin{proof}(Sketch)
The lemma holds since  the $k$-edge connectivity problem is in $\dcoNP$ if the verifier $\V'$ (for $k$-edge connectivity) can use an oracle for the dynamic connectivity problem (which is in $\dNP$). We now explain why this is the case. Specifically, a proof $\pi'_t$ for $\V'$ encodes a subset $E' \subseteq E$ of at most $k$ edges in $G$.  Upon receiving this proof $\pi'_t$, the verifier $\V'$ calls the oracle for dynamic connectivity on the instance $G' = (V, E \setminus E')$. If this oracle returns a NO answer, then we know that the graph $G$ becomes disconnected if we delete the (at most $k$) edges $E'$ from $G$, and hence the graph $G$ is {\em not} $k$-edge connected. Therefore, in this scenario the verifier $\V'$ outputs $(x'_t, y'_t)$ where $x'_t = 0$ and $y'_t = 1$. Otherwise, if the oracle returns YES, then the verifier $\V'$ outputs $(x'_t, y'_t)$ where $x'_t = 1$ and $y'_t = 0$.
\end{proof}

\section{Reformulations of $\protect\dDNF$: proof}

\label{sec:DNF equiv proof}
\begin{prop}
[Restatement of \ref{prop:DNF equiv}]An algorithm with the update
and query time are at most $O(u(m))$ for any one of the following
problems implies algorithms with the same update and query time for
all other problems:
\begin{enumerate}
\item $\dDNF$ on an formula $F$ with $m$ clauses,
\item $\dAW$ on an graph $G=(L,R,E)$ where $|R|=m$ clauses,
\item $\dIndep$ on a hypergraph $H$ with $m$ edges, and
\item $\dOV$ on a matrix $V\in\{0,1\}^{n\times m}$.
\end{enumerate}
\end{prop}
\begin{proof}
We will show how to map any instance of each problem to an instance
of another problem. It is clear from the mapping how we should handle
the updates. 

($1\rightarrow2$): Given an instance $(F,\phi)$ of $\dDNF$ where
$F$ has $n$ variables and $m$ clauses, we can construct $G=(L,R,E)$
where $|L|=2n$ and $|R|=m$. We write $L=(l_{1},\dots,l_{2n})$ and
$R=(r_{1},\dots,r_{m})$. If a clause $C_{j}$ contains $x_{i}$,
then $(l_{2i-1},r_{j})\in E$. If $C_{j}$ contains $\neg x_{i}$,
then $(l_{2i},r_{j})$. If $\phi(x_{i})=1$, then $l_{2i-1}$ is white
and $l_{2i}$ is black. Otherwise, $l_{2i-1}$ is black and $l_{2i}$
is white. It is easy to see that $F(\phi)$ iff there is a node in
$R$ whose neighbors are all white. 

($2\rightarrow3$): Given an instance $G=(L,R,E)$ of $\dAW$ where
nodes in $L$ are colored, we construct a hypergraph $H=(V,E_{H})$
where $V=L$ and, for each $j$, the edge $e_{j}$ of $E_{H}$ is
a set of neighbors of $r_{j}$ in $G$. Let $S\subseteq V$ be the
set corresponding to white nodes in $L$. It holds that $S$ is \emph{not
}independent in $H$ iff there is a node in $R$ whose neighbors are
all white. 

($3\rightarrow1$): Given an instance $(H,S)$ of $\dIndep$ where
$H=(V,E)$ where $V=(v_{1},\dots,v_{n})$ and $E=(e_{1},\dots,e_{m})$,
we construct a formula $F$ with variables $\X=(x_{1},\dots,x_{n})$
and $m$ clauses. For each $e_{j}\in E$, we construct a clause $C_{j}$
where $C_{j}$ contains $x_{i}$ iff $v_{i}\in e_{j}$. For each $v_{i}\in S$,
we set $\phi(x_{i})=1$, otherwise $\phi(x_{i})=0$. Therefore, $F(\phi)=1$
iff $S$ is \emph{not }independent in $H$.

($2\rightleftarrows4$): Given an instance $G=(L,R,E)$ of $\dAW$
where nodes in $L$ are colored, the corresponding instance of $\dOV$
is the matrix $V\in\{0,1\}^{n\times m}$ where $V$ is a bi-adjacency
matrix of $G$, i.e. $v_{ij}=1$ iff $(l_{i},r_{j})\in E$. The vector
$u$ corresponds to all \emph{black }nodes i.e. $u_{i}=1$ iff $l_{i}$
is black. It hols that a node $r_{j}\in R$ whose neighbors are all
white iff $r_{j}$ has no black neighbor iff a vector $v_{j}\in V$
where $u^{T}v_{j}=0$.
\end{proof}

\section{An Issue Regarding Promise Problems and $\protect\dP$-reductions}

\label{sec:issue promise}

Even, Selman and Yacobi \cite[Theorem 4]{EvenSY84} show that there
is a promise problem $\P$ where $\P\in\mathsf{NP}\cap\mathsf{coNP}$
but $\P$ is $\mathsf{NP}$-hard under Turing reduction (see also
the survey by Goldreich \cite{Goldreich06a}). As $\dP$-reduction
can be viewed as a dynamic version of Turing reduction, in this section,
we prove that there is an dynamic problem with similar properties:
\begin{thm}
	\label{thm:intersection but hard}There is a promise problem $\D$
	such that $\D\in\dNP\cap\dcoNP$ but $\D$ is $\dNP$-hard.
\end{thm}
Before proving, we give some discussion. This seemingly contradictory
situations, for both static and dynamic problems, are only because
the reduction can ``exploit'' don't care instances of a promise
problem $\D$. If $\D$ does not contain any don't care instance,
then \ref{thm:intersection but hard} cannot hold assuming that $\dNP\not\subseteq\dcoNP$. 

Suppose by contradiction that $\D\in\dNP\cap\dcoNP$, $\D$
is $\dNP$-hard, and $\D$ does not contain any don't care instance.
\ref{thm:easiness transfering} (3) implies that $\dNP\subseteq\dNP\cap\dcoNP$,
contradicting the assumption that $\dNP\not\subseteq\dcoNP$.
In other words, assuming that $\dNP\not\subseteq\dcoNP$, this
explains why we need the condition that $\D$ does not contain any
don't care instance in \ref{thm:easiness transfering} (3). 

Now, we prove \ref{thm:intersection but hard}.
\begin{proof}
	The proof are analogous to the statement for static problems in \cite[Theorem 4]{EvenSY84}.
	The problem $\D$ is the problem where we are given two instances
	$(F_{1},\phi_{1}),(F_{2},\phi_{2})$ from the $\dDNF$ problem. $(F_{1},\phi_{1}),(F_{2},\phi_{2})$
	is a yes-instance if $F_{1}(\phi_{1})=1$ and $F_{2}(\phi_{2})=0$.
	It is a no-instance if $F_{1}(\phi_{1})=0$ and $F_{2}(\phi_{2})=1$.
	Otherwise, it is a don't care instance. At each step, we can update
	either a variable assignment $\phi_{1}$ or $\phi_{2}$. It is easy
	to see that this problem is in $\dNP\cap\dcoNP$. Now, there
	is a $\dP$-reduction from $\fDNF$ to $\D$ essentially by using
	the same proof as in \ref{lem:first to exist}. So we omitted it in
	this version.
\end{proof}

\bibliographystyle{alpha}
\bibliography{paper}

\end{document}